\documentclass[11pt]{article}
\usepackage[letterpaper]{geometry}
\author{
Paraschos Koutris\\
University of Wisconsin-Madison\\ 
WI, USA\\
\texttt{paris@cs.wisc.edu}\\
\and
Jef Wijsen\\
University of Mons\\
Mons, Belgium\\
\texttt{jef.wijsen@umons.ac.be}
}

\title{Consistent Query Answering for Primary Keys in Logspace}
\date{}

\usepackage{enumitem}
\usepackage[ruled]{algorithm2e} 
\usepackage[new]{old-arrows}
\usepackage{xspace}
\usepackage{setspace}
\usepackage{latexsym}
\usepackage{amsmath}
\usepackage{amssymb}
\usepackage{amsthm}
\usepackage{times}
\usepackage{xcolor}
\usepackage{graphicx}
\usepackage[bold]{complexity}
\usepackage{stmaryrd}
\usepackage{arydshln}
\usepackage{bigstrut}
\usepackage[normalem]{ulem}
\usepackage{subcaption}

\newcommand{\blockfiller}{\ast}

\newcommand{\catoms}[1]{{#1}^{\mathsf{cons}}}
\newcommand{\rename}[2]{{#1}^{(#2)}}
\makeatletter
\newcommand{\ssymbol}[1]{\@fnsymbol{#1}}
\newcommand{\supone}{\ssymbol{2}}
\newcommand{\supthree}{\ssymbol{4}}
\newcommand{\predicate}[1]{\mathsf{#1}}
\newcommand{\conedge}{\predicate{Link}}
\newcommand{\concomp}{\predicate{Trans}}

\newcommand{\anyone}{\predicate{Any1Emb}}
\newcommand{\relone}{\predicate{Rel1Emb}}
\newcommand{\irrone}{\predicate{Irr1Emb}}
\newcommand{\polyedge}{\predicate{E}}
\newcommand{\cpolyedge}[4]{\polyedge({#1},{#2},\,{#3},{#4})}
\newcommand{\polydistinct}{\predicate{Neq}}
\newcommand{\cpolydistinct}[4]{\polydistinct({#1},{#2},\,{#3},{#4})}
\newcommand{\polyconnected}{\predicate{UCon}}
\newcommand{\polylong}{\predicate{InLongCycle}}

\newcommand{\del}[1]{\predicate{Garbage{#1}}}
\newcommand{\good}[1]{\predicate{Rlvant{#1}}}
\newcommand{\keep}[1]{\predicate{Keep{#1}}}
\newcommand{\detu}[1]{\predicate{N{#1}}}

\newcommand{\eqpred}[1]{\predicate{Eq{#1}}}
\newcommand{\diseqpred}[1]{\predicate{Neq{#1}}}

\newcommand{\encodet}{\predicate{T}}

\newcommand{\pick}{\predicate{IdentifiedBy}}

\newcommand{\att}[1]{\mathsf{#1}}

\newcommand{\logspace}{\ComplexityFont{L}}

\newcommand{\calV}{{\mathcal{V}}}

\newcommand{\card}[1]{\left|{#1}\right|}
\newcommand{\formula}[1]{({#1})}
\newcommand{\lrformula}[1]{\left({#1}\right)}

\newcommand{\tuple}[1]{\langle{#1}\rangle}
\newcommand{\defeq}{\mathrel{\mathop:}=}
\newcommand{\signature}[2]{[{#1},{#2}]}
\newcommand{\bfo}{{\mathbf{o}}}
\newcommand{\bfp}{{\mathbf{p}}}
\newcommand{\db}{\mbox{${\mathbf{db}}$}}
\newcommand{\rep}{{\mathbf{r}}}
\newcommand{\sep}{{\mathbf{s}}}

\newcommand{\block}{{\mathbf{b}}}
\newcommand{\theblock}[2]{{\mathsf{block}}({#1},{#2})}
\newcommand{\repairs}[1]{\mathsf{rset}({#1})}
\newcommand{\atomvars}[1]{{\mathsf{vars}}({#1})}
\newcommand{\keyvars}[1]{{\mathsf{key}}({#1})}
\newcommand{\queryvars}[1]{\mathsf{vars}({#1})}
\newcommand{\sequencevars}[1]{\mathsf{vars}({#1})}
\newcommand{\substitute}[3]{{#1}_{[{{#2}\mapsto{#3}}]}}
\newcommand{\fd}[2]{{#1}\rightarrow{#2}}

\newcommand{\FD}[1]{{\mathcal{K}}({#1})}
\newcommand{\keycl}[2]{{#1}^{+,{#2}}}

\newcommand{\cqa}[1]{{\mathsf{CERTAINTY}}({#1})}
\newcommand{\attacksymbol}[1]{\stackrel{#1}{\rightsquigarrow}}
\newcommand{\nattacksymbol}[1]{\stackrel{#1}{\not\rightsquigarrow}}
\newcommand{\attacks}[1]{\attacksymbol{#1}}

\newcommand{\nattacks}[1]{\nattacksymbol{#1}}
\newtheorem{lemma}{Lemma}

\newtheorem{corollary}{Corollary}
\newtheorem{theorem}{Theorem}
\theoremstyle{definition}
\newtheorem{definition}{Definition}
\newtheorem{example}{Example}

\newcommand{\step}[1]{\stackrel{#1}{\smallfrown}}

\newcommand{\mgraphsymbol}{${\mathsf{M}}$}
\newcommand{\mgraph}{\mgraphsymbol-graph\xspace}
\newcommand{\mgraphs}{\mgraphsymbol-graphs\xspace}
\newcommand{\mGraphs}{\mgraphsymbol-Graphs\xspace}
\newcommand{\mGraph}{\mgraphsymbol-Graph\xspace}
\newcommand{\mcycle}{\mgraphsymbol-cycle\xspace}
\newcommand{\mcycles}{\mgraphsymbol-cycles\xspace}
\newcommand{\mCycles}{\mgraphsymbol-Cycles\xspace}
\newcommand{\markov}{\stackrel{\mathsf{{}_{M}}}{\longrightarrow}}

\newcommand{\ucq}{{\mathsf{UCQ}}}

\newcommand{\sjfbcq}{{\mathsf{sjfBCQ}}}

\newcommand{\rifi}[1]{{\textsf{Reify}}({#1})}
\newcommand{\isc}{\mathcal{S}}

\newcommand{\mhook}{\hookrightarrow}
\newcommand{\cmhook}[1]{\stackrel{_{#1}}{\hookrightarrow}}
\newcommand{\problem}[1]{{\textsf{#1}}}
\newcommand{\keyequal}{\sim}

\newcommand{\prewrite}{\operatorname{\mathbf{P-Rewrite}}}
\newcommand{\calL}{{\mathcal{L}}}

\newcommand{\rdir}{\att{Directors}}
\newcommand{\rmov}{\att{Movies}}

\newcommand{\constant}[1]{\mbox{`{#1}'}}
\newcommand{\qpruning}{q_{1}}
\newcommand{\qatomfun}[1]{{\mathsf{genre}}_{#1}}
\newcommand{\qatom}[2]{\qatomfun{#2}({#1})}
\newsavebox{\savepar}

\newenvironment{datalogpgm}
{\begin{small}
\begin{array}{l@{\,\leftarrow\,}l}}
{\end{array}
\end{small}
}
\newenvironment{datalogpgmns}
{
\begin{array}{l@{\,\leftarrow\,}l}}
{\end{array}
}
\newcommand{\spacebetweenrules}{\multicolumn{2}{c}{\mbox{}}\\}

\newcommand{\myparagraph}[1]{\vspace{0.15\baselineskip}\noindent\textbf{#1}.}
\newcommand{\smallcell}[1]{\begin{minipage}[t]{0.1\textwidth}{{#1}}\end{minipage}}
\newcommand{\cell}[1]{\begin{minipage}[t]{0.45\textwidth}{{#1}\vspace{2ex}}\end{minipage}}
\newcommand{\po}[1]{\widehat{#1}}

\newcommand{\dom}{{\mathbf{dom}}}
\newcommand{\cid}{{\mathit{cid}}}
\newcommand{\start}[1]{\mathsf{start}({#1})}
\newcommand{\fin}[1]{\mathsf{end}({#1})}
\newcommand{\CC}[1]{{x}_{#1},{y}_{#1}}
\newcommand{\mc}{\mathsf{c}}
\newcommand{\mi}{\mathsf{i}}
\newcommand{\ssdatalog}{{\mathit{SymStratDatalog}}}
\newcommand{\ssdatalogmin}{\ssdatalog^{\min}}
\newcommand{\qcycle}{C_{3}}
\newcommand{\common}{{\mathsf{shared}}}

\begin{document}

\maketitle


\begin{abstract}
We study the complexity of consistent query answering on databases that may violate primary key constraints.
A repair of such a database is any consistent database that can be obtained by deleting a minimal set of tuples. 
For every Boolean query $q$,
$\cqa{q}$ is the problem that takes a database as input and asks whether $q$ evaluates to true on every repair.
In~\cite{DBLP:journals/tods/KoutrisW17},
the authors show that for every self-join-free Boolean conjunctive query $q$,
the problem $\cqa{q}$ is either in $\P$ or $\coNP$-complete,
and it is decidable which of the two cases applies.
In this paper, we sharpen this result by showing that for every self-join-free Boolean conjunctive query $q$, the problem $\cqa{q}$ is either expressible in symmetric stratified Datalog or $\coNP$-complete.
Since symmetric stratified Datalog is in $\logspace$, we thus obtain a complexity-theoretic dichotomy between $\logspace$ and $\coNP$-complete.
Another new finding of practical importance is that $\cqa{q}$ is on the logspace side of the dichotomy for queries $q$ where all join conditions express foreign-to-primary key matches, which is undoubtedly the most common type of join condition.
\end{abstract}


\section{Motivation}

{\em Consistent query answering\/} (CQA) with respect to primary key constraints is the following problem.
Given a database $\db$ that may violate its primary key constraints,
define a repair as any consistent database that can be obtained by deleting a minimal set of tuples from $\db$.
For every Boolean query~$q$, the problem $\cqa{q}$ takes a database as input and asks whether $q$ evaluates to true on every repair of $\db$. 
In this paper, we focus on $\cqa{q}$ for queries $q$ in the class $\sjfbcq$, the class of self-join-free Boolean conjunctive queries.
In~\cite{DBLP:journals/tods/KoutrisW17},
the authors show that for every query~$q$ in $\sjfbcq$,
the problem $\cqa{q}$ is either in $\P$ or $\coNP$-complete.
This result is proved in a constructive way:
the authors introduce a syntactic (and decidable) property (call it $p$)\footnote{We will recall in Section~\ref{sec:preliminaries} that the property $p$ is: having an attack graph without strong cycles.}  for self-join-free Boolean conjunctive queries, and then show two things:
for queries~$q$ not possessing the property~$p$, $\cqa{q}$ is proved to be $\coNP$-complete;
and for queries $q$ possessing the property~$p$, 
an effective procedure is described that constructs a polynomial-time algorithm for $\cqa{q}$.
For clarity of exposition, we will call this effective procedure $\prewrite$ from here on.
Thus, $\prewrite$ takes as input a self-join-free Boolean conjunctive query $q$ possessing the property~$p$, and returns as output a polynomial-time algorithm that solves $\cqa{q}$.

Given an input query $q$, the procedure $\prewrite$ behaves differently depending on whether $\cqa{q}$ is in $\FO$ or in $\P\setminus\FO$. 
Membership of $\cqa{q}$ in $\FO$ is known to be decidable for queries $q$ in $\sjfbcq$.
If $\cqa{q}$ is in $\FO$, then the output of $\prewrite$ is a relational calculus query which can be encoded in SQL;
in this case, the functioning of $\prewrite$ is well understood and easily implementable (such implementations already exist~\cite{Pijcke2018}), as explained in~\cite[Section~5]{DBLP:journals/tods/KoutrisW17}.
The situation is different if $\cqa{q}$ is in $\P\setminus\FO$.
In~\cite{DBLP:journals/tods/KoutrisW17}, the authors show that if a query $q$ possesses the property $p$ but $\cqa{q}$ is not in $\FO$, then one can construct a polynomial-time algorithm for $\cqa{q}$. 
The construction in~\cite{DBLP:journals/tods/KoutrisW17}, however, is complicated and not amenable to easy implementation.
In this paper, we improve this situation;
we show that if a query $q$ possesses the property $p$,
then one can implement $\cqa{q}$ in symmetric stratified Datalog ($\ssdatalog$), which has logspace data complexity~\cite{DBLP:conf/lics/EgriLT07}.
We thus sharpen the complexity dichotomy of~\cite{DBLP:journals/tods/KoutrisW17} as follows: for every query $q$ in $\sjfbcq$, $\cqa{q}$ is either in $\logspace$ or $\coNP$-complete.
It is significant that Datalog is used as a target language for $\prewrite$, because this allows using optimized Datalog engines for solving $\cqa{q}$ whenever the problem lies on the logspace side of the dichotomy.
Rewriting into Datalog is generally considered a desirable outcome when consistent first-order rewritings do not exist (see, e.g., \cite[page~193]{DBLP:books/daglib/0041477}).
It is also worth noting that the SQL:1999 standard introduced linear recursion into SQL, which has been implemented in varying ways in existing DBMSs~\cite{DBLP:conf/fgit/PrzymusBBS10}. 
Since the Datalog programs in this paper never use non-linear recursion,
they may be partially or fully implementable in these DBMSs.

Throughout this paper, we use the term {\em consistent database\/} to refer to a database that satisfies all primary-key constraints, while the term {\em database\/} refers to both consistent and inconsistent databases.
This is unlike most database textbooks, which tend to say that databases must always be consistent.
The following definition introduces the main focus of this paper;
the complexity dichotomy of Theorem~\ref{the:dichotomy} is the main result of this paper.

\begin{definition}
Let $q$ be a Boolean query.
Let $\calL$ be some logic.
A {\em consistent $\calL$ rewriting for $q$\/} is a Boolean query $P$ in $\calL$ such that for every database $\db$,
$P$ is true in $\db$ if and only if $q$ is true in every repair of $\db$.
If $q$ has a consistent $\calL$ rewriting,
then we say that $\cqa{q}$ is expressible in $\calL$.
\end{definition}


\begin{theorem}\label{the:dichotomy}
For every self-join-free Boolean conjunctive query $q$, the problem $\cqa{q}$ is either $\coNP$-complete or expressible in $\ssdatalogmin$ (and thus in $\logspace$).
\end{theorem}

The language $\ssdatalogmin$ will be defined in Section~\ref{sec:preliminaries}; informally, the superscript $\min$ means that the language allows selecting a minimum (with respect to some total order) from a finite set of values.
Since $\cqa{q}$ is $\logspace$-complete for some queries $q\in\sjfbcq$,
the logspace upper bound in Theorem~\ref{the:dichotomy} is tight.
The proof of Theorem~\ref{the:dichotomy} relies on novel constructs and insights developed in this paper.

Our second significant result in this paper focuses on consistent query answering for foreign-to-primary key joins.
In Section~\ref{sec:keyjoin}, we define a subclass of $\sjfbcq$ that captures foreign-to-primary key joins, which is undoubtedly the most common type of join. We show that $\cqa{q}$ lies on the logspace side of the dichotomy for \emph{all} queries $q$ in this class.
Thus, for the most common type of joins and primary key constraints, CQA is highly tractable, a result that goes against a widely spread belief that CQA would be impractical because of its high computational complexity.

\myparagraph{Organization}
Section~\ref{sec:related} discusses related work.
Section~\ref{sec:preliminaries} defines our theoretical framework, including the notion of \emph{attack graph}.
To guide the reader through the technical development,
Section~\ref{sec:gt} provides a high-level outline of where we are heading in this paper, including examples of the different graphs used. 
Section~\ref{sec:glimpse} introduces a special subclass of $\sjfbcq$, called  \emph{saturated queries}, and shows that each problem $\cqa{q}$ can be first-order reduced to some $\cqa{q'}$ where $q'$ is saturated. 
Section~\ref{sec:mgraph} introduces the notion of \mgraph, a graph at the schema-level,
and its data-level instantiation, called $\mhook$-graph.
An important result, Lemma~\ref{lem:happy}, relates cycles in attack graphs to cycles in \mgraphs, for saturated queries only.
Section~\ref{sec:garbageset} introduces the notion of garbage set for a subquery.
Informally,  garbage sets contain facts that can never make the subquery hold true, and thus can be removed from the database without changing the answer to $\cqa{q}$.
Section~\ref{sec:scrub} focuses on cycles in the \mgraph of a query, and shows that garbage sets for such cycles can be computed and removed in symmetric stratified Datalog.
At the end of Section~\ref{sec:scrub}, we have all ingredients for the proof of our main theorem.
Finally, Section~\ref{sec:keyjoin} shows that foreign-to-primary key joins fall on the logspace side of the dichotomy.
Most proofs have been moved to an appendix.
Appendix~\ref{sec:notation} contains a list of notations for easy reference.

\section{Related Work}\label{sec:related}

Consistent query answering (CQA) was initiated by the seminal work by 
Arenas, Bertossi, and Chomicki~\cite{DBLP:conf/pods/ArenasBC99},
and is the topic of the monograph~\cite{DBLP:series/synthesis/2011Bertossi}.
The term $\cqa{q}$ was coined in~\cite{DBLP:conf/pods/Wijsen10} to refer to CQA for Boolean queries $q$ on databases that violate primary keys, one per relation,
which are fixed by $q$'s schema.
The complexity classification of $\cqa{q}$ for all $q\in\sjfbcq$ started with the ICDT 2005 paper of Fuxman and Miller~\cite{FUXMAN2005,DBLP:journals/jcss/FuxmanM07}, and has attracted much research since then.
These previous works (see ~\cite{DBLP:conf/foiks/Wijsen14} for a survey) were generalized by~\cite{DBLP:conf/pods/KoutrisW15,DBLP:journals/tods/KoutrisW17}, where it was shown that the set $\{\cqa{q}\mid q\in\sjfbcq\}$ exhibits a $\P$-$\coNP$-complete dichotomy.
Furthermore, it was shown that membership of $\cqa{q}$ in $\FO$ is decidable for queries $q$ in $\sjfbcq$.
The current paper culminates this line of research by showing that the dichotomy is actually between $\logspace$ and $\coNP$-complete, and---even stronger---between expressibility in symmetric stratified Datalog and $\coNP$-complete.

The complexity of $\cqa{q}$ for self-join-free conjunctive queries with negated atoms was studied in~\cite{DBLP:conf/pods/KoutrisW18}.
Little is known about $\cqa{q}$ beyond self-join-free conjunctive queries.
For $\ucq$ (i.e., unions of conjunctive queries, possibly with self-joins),
Fontaine~\cite{DBLP:conf/lics/Fontaine13} showed that a $\P$-$\coNP$-complete dichotomy in the set 
$\{\cqa{q}\mid q$ is a Boolean query in $\ucq\}$ 
implies Bulatov's dichotomy theorem for conservative CSP~\cite{DBLP:journals/tocl/Bulatov11}.
This relationship between CQA and CSP was further explored in~\cite{DBLP:conf/icdt/LutzW15}.
The complexity of CQA for aggregation queries with respect to violations of functional dependencies has been studied in~\cite{DBLP:journals/tcs/ArenasBCHRS03}. 

The counting variant of $\cqa{q}$, which is called $\#\cqa{q}$, asks to determine the number of repairs that satisfy some Boolean query $q$. 
In~\cite{DBLP:journals/jcss/MaslowskiW13}, the authors show a $\FP$-$\#\P$-complete dichotomy in $\{\#\cqa{q}\mid q\in\sjfbcq\}$.
For conjunctive queries $q$ with self-joins,
the complexity of $\#\cqa{q}$ has been established for the case that all primary keys consist of a single attribute~\cite{DBLP:conf/icdt/MaslowskiW14}.

The paradigm of CQA has been implemented in expressive formalisms, such as Disjunctive Logic Programming~\cite{DBLP:journals/tkde/GrecoGZ03} and Binary Integer Programming~\cite{DBLP:journals/pvldb/KolaitisPT13}.
In these formalisms, it is relatively straightforward to express an exponential-time algorithm for $\cqa{q}$.
The drawback is that the efficiency of these algorithms is likely to be far from optimal in case that certain answers are computable in logspace or expressible in first-order logic.


\section{Preliminaries}\label{sec:preliminaries}

We assume an infinite total order $(\dom,\leq)$ of {\em constants\/}.
We assume a set of {\em variables\/} disjoint with $\dom$.
If $\vec{x}$ is a sequence containing variables and constants, then $\sequencevars{\vec{x}}$ denotes the set of variables that occur in $\vec{x}$.
A {\em valuation\/} over a set $U$ of variables is a total mapping $\theta$ from $U$ to $\dom$.
At several places, it is implicitly understood that such a valuation $\theta$ is extended to be the identity on constants and on variables not in~$U$.
If $V \subseteq U$, then $\theta[V]$ denotes the restriction of $\theta$ to $V$.
If $\theta$ is a valuation over a set $U$ of variables, $x$ is a variable (possibly $x\notin U$), and $a$ is a constant,
then $\substitute{\theta}{x}{a}$ is the valuation over $U\cup\{x\}$ such that
$\substitute{\theta}{x}{a}(x)=a$ and for every variable $y$ such that $y\neq x$, $\substitute{\theta}{x}{a}(y)=\theta(y)$.

\myparagraph{Atoms and key-equal facts}
Each {\em relation name\/} $R$ of arity $n$, $n\geq 1$, has a unique {\em primary key\/} which is a set $\{1,2,\dots,k\}$ where $1\leq k\leq n$.
We say that $R$ has {\em signature\/} $\signature{n}{k}$ if $R$ has arity $n$ and primary key $\{1,2,\dots,k\}$. 
Elements of the primary key are called {\em primary-key positions\/},
while $k+1$, $k+2$, \dots, $n$ are {\em non-primary-key positions\/}. 
For all positive integers $n,k$ such that $1\leq k\leq n$, we assume denumerably many relation names with signature $\signature{n}{k}$.
Every relation name has a unique {\em mode\/}, which is a value in $\{\mc,\mi\}$.
Informally, relation names of mode~$\mc$ will be used for consistent relations,
while relations that may be inconsistent will have a relation name of mode~$\mi$.  
We often write $R^{\mc}$ to make clear that $R$ is a relation name of mode~$\mc$.

If $R$ is a relation name with signature $\signature{n}{k}$, then we call $R(s_{1},\dots,s_{n})$ an {\em $R$-atom\/} (or simply atom), where each $s_{i}$ is either a constant or a variable ($1\leq i\leq n$).
Such an atom is commonly written as $R(\underline{\vec{x}},\vec{y})$ where the primary-key value $\vec{x}=s_{1},\dots,s_{k}$ is underlined and $\vec{y}=s_{k+1},\dots,s_{n}$.
An {\em $R$-fact\/} (or simply fact) is an $R$-atom in which no variable occurs.
Two facts $R_{1}(\underline{\vec{a}_{1}},\vec{b}_{1}),R_{2}(\underline{\vec{a}_{2}},\vec{b}_{2})$ are {\em key-equal\/},
denoted $R_{1}(\underline{\vec{a}_{1}},\vec{b}_{1})\keyequal R_{2}(\underline{\vec{a}_{2}},\vec{b}_{2})$,
 if $R_{1}=R_{2}$ and $\vec{a}_{1}=\vec{a}_{2}$.

We will use letters $F,G,H$ for atoms.
For an atom $F=R(\underline{\vec{x}},\vec{y})$, we denote by $\keyvars{F}$ the set of variables that occur in $\vec{x}$,
and by $\atomvars{F}$ the set of variables that occur in $F$, that is, $\keyvars{F}=\sequencevars{\vec{x}}$ and $\atomvars{F}=\sequencevars{\vec{x}}\cup\sequencevars{\vec{y}}$.
We sometimes blur the distinction between relation names and atoms.
For example, if $F$ is an atom,
then the term $F$-fact refers to a fact with the same relation name as $F$.

\myparagraph{Databases, blocks, and repairs}
A {\em database schema\/} is a finite set of relation names.
All constructs that follow are defined relative to a fixed database schema.
A {\em database\/} is a finite set $\db$ of facts using only the relation names of the schema such that for every relation name $R$ of mode~$\mc$, no two distinct $R$-facts of $\db$ are key-equal.

A {\em relation\/} of $\db$ is a maximal set of facts in $\db$ that all share the same relation name.
A {\em block\/} of $\db$ is a maximal set of key-equal facts of $\db$.
A block of $R$-facts is also called an $R$-block.
If $A$ is a fact of $\db$, then $\theblock{A}{\db}$ denotes the block of $\db$ that contains $A$.
If $A=R(\underline{\vec{a}},\vec{b})$, then $\theblock{A}{\db}$ is also denoted by$R(\underline{\vec{a}},\blockfiller)$. 
A database $\db$ is {\em consistent\/} if no two distinct facts of $\db$ are key-equal 
(i.e., if no block of $\db$ contains more than one fact).
A {\em repair\/} of $\db$ is a maximal (with respect to set inclusion) consistent subset of $\db$. 
We write $\repairs{\db}$ for the set of repairs of $\db$.


\myparagraph{Boolean conjunctive queries}
A {\em Boolean query\/} is a mapping $q$ that associates a Boolean (true or false) to each database, 
such that $q$ is closed under isomorphism~\cite{DBLP:books/sp/Libkin04}.
We write $\db\models q$ to denote that $q$ associates true to $\db$, in which case $\db$ is said to {\em satisfy\/} $q$.
A Boolean query $q$ can be viewed as a decision problem that takes a database as input and asks whether $\db$ satisfies~$q$. 
In this paper, the complexity class $\FO$ stands for the set of Boolean queries that can be defined in first-order logic with equality and constants, but without other built-in predicates or function symbols.

A {\em Boolean conjunctive query\/} is a finite set 
$q=\{R_{1}(\underline{\vec{x}_{1}},\vec{y}_{1})$, $\dots$, $R_{n}(\underline{\vec{x}_{n}},\vec{y}_{n})\}$ of atoms, without equality or built-in predicates.
We denote by $\queryvars{q}$ the set of variables that occur in $q$.
The set $q$ represents the first-order sentence 
$$\exists u_{1}\dotsm\exists u_{k}\left(R_{1}(\underline{\vec{x}_{1}},\vec{y}_{1})\land\dotsm\land R_{n}(\underline{\vec{x}_{n}},\vec{y}_{n})\right),$$ where $\{u_{1}, \dots, u_{k}\}=\queryvars{q}$.
This query $q$ is satisfied by a database $\db$ if there exists a valuation $\theta$ over $\queryvars{q}$ such that for each $i\in\{1,\dots,n\}$, 
$R_{i}(\underline{\vec{a}},\vec{b})\in\db$ with $\vec{a}=\theta(\vec{x}_{i})$ and $\vec{b}=\theta(\vec{y}_{i})$. 

We say that a Boolean conjunctive query $q$ has a {\em self-join\/} if some relation name occurs more than once in $q$.
If $q$ has no self-join, then it is called {\em self-join-free\/}. 
We write $\sjfbcq$ for the class of self-join-free Boolean conjunctive queries.
If $q$ is a query in $\sjfbcq$ with an $R$-atom, then, by an abuse of notation, we sometimes write $R$ to mean the $R$-atom of $q$. 

Let $\theta$ be a valuation over some set $X$ of variables.
For every  Boolean conjunctive query $q$, we write $\theta(q)$ for the query obtained from $q$ by replacing all occurrences of each $x\in X\cap\queryvars{q}$ with $\theta(x)$;
variables in $\queryvars{q}\setminus X$ remain unaffected (i.e., $\theta$ is understood to be the identity on variables not in $X$).

\myparagraph{Atoms of mode~$\mc$}
The {\em mode\/} of an atom is the mode of its relation name (a value in $\{\mc,\mi\}$).
If $q$ is a query in $\sjfbcq$,
then $\catoms{q}$ is the set of all atoms of $q$ that are of mode~$\mc$.


\myparagraph{Functional dependencies}
Let $q$ be a Boolean conjunctive query.
A {\em functional dependency for $q$\/} is an expression $\fd{X}{Y}$ where $X,Y\subseteq\queryvars{q}$.
Let $\calV$ be a finite set of valuations over $\queryvars{q}$.
We say that $\calV$ \emph{satisfies} $\fd{X}{Y}$ if for all $\theta,\mu\in\calV$,
if $\theta[X]=\mu[X]$, then $\theta[Y]=\mu[Y]$.
Let $\Sigma$ be a set of functional dependencies for $q$.
We write $\Sigma\models\fd{X}{Y}$ if for every set $\calV$ of valuations over $\queryvars{q}$,
if $\calV$ satisfies each functional dependency in $\Sigma$, then $\calV$ satisfies $\fd{X}{Y}$. 
Note that the foregoing conforms with standard dependency theory if variables are viewed as attributes, and valuations as tuples.


\myparagraph{Consistent query answering}
Let $q$ be a query in $\sjfbcq$.
We define $\cqa{q}$ as the decision problem that takes as input a 
database $\db$, 
and asks whether every repair of $\db$ satisfies~$q$.

\myparagraph{The genre of a fact}
Let $q$ be a query in $\sjfbcq$.
For every fact $A$ whose relation name occurs in $q$,
we denote by $\qatom{A}{q}$ the (unique) atom of $q$ that has the same relation name as $A$. 
From here on, if $\db$ is a database that is given as an input to $\cqa{q}$,
we will assume that each relation name of each fact in $\db$ also occurs in $q$.
Therefore, for every $A\in\db$, $\qatom{A}{q}$ is well defined.
Of course, this assumption is harmless.

\myparagraph{Attack graph}
Let $q$ be a query in $\sjfbcq$.
We define $\FD{q}$ as the following set of functional dependencies:
$\FD{q}\defeq\{\fd{\keyvars{F}}{\atomvars{F}}\mid F\in q\}$.
%
%
%
For every atom $F\in q$, we define $\keycl{F}{q}$ as the set of all variables $x\in\queryvars{q}$ satisfying
$\FD{q\setminus\{F\}}\cup\FD{\catoms{q}}\models\fd{\keyvars{F}}{x}$.
Informally, the term $\FD{\catoms{q}}$ is the set of all functional dependencies that arise in atoms of mode~$\mc$.
The {\em attack graph\/} of $q$ is a directed graph whose vertices are the atoms of $q$.
There is a directed edge from $F$ to $G$ ($F\neq G$), denoted $F\attacks{q}G$, if there exists a sequence
\begin{equation}\label{eq:witness}
F_{0}\step{x_{1}}F_{1}\step{x_{2}}F_{2}\dotsm\step{x_{\ell}}F_{\ell}
\end{equation}
such that $F_{0}=F$, $F_{\ell}=G$, and for each $i\in\{1,\dots,\ell\}$,
$F_{i}$ is an atom of $q$ and $x_{i}$ is a variable satisfying
$x_{i}\in\lrformula{\atomvars{F_{i-1}}\cap\atomvars{F_{i}}}\setminus \keycl{F}{q}$.
The sequence~(\ref{eq:witness}) is also called a {\em witness\/} for $F\attacks{q}G$.
An edge $F\attacks{q}G$ is also called an {\em attack from $F$ to $G$\/}; we also say that {\em $F$ attacks $G$\/}.

An attack on a variable $x\in\queryvars{q}$ is defined as follows:
$F\attacks{q}x$ if $F\attacks{q\cup\{N(\underline{x})\}}N(\underline{x})$ where $N$ is a fresh relation name of signature~$\signature{1}{1}$.
Informally, $x$ is attacked in $q$ if $N(\underline{x})$ has an incoming attack in the attack graph of $q\cup\{N(\underline{x})\}$.

\begin{example}
Let $\qpruning=\{R(\underline{x},y)$, $S(\underline{y},z)$, $U(\underline{y,z,w},x)$, 
$T_{1}(\underline{z},w)$, $T_{2}(\underline{z},w)$, $T^{\mc}(\underline{z},w)\}$.
Using relation names for atoms, we have $\keycl{R}{\qpruning}=\{x\}$.
A witness for $R\attacks{\qpruning}U$ is $R\step{y}U$.
The attack graph of $\qpruning$ is shown in Fig.~\ref{fig:pruning}.
\end{example}

An attack $F\attacks{q}G$ is {\em weak\/} if $\FD{q}\models\fd{\keyvars{F}}{\keyvars{G}}$;
otherwise it is {\em strong\/}.
A cycle in the attack graph is {\em strong\/} if at least one attack in the cycle is strong.
It has been proved~\cite[Lemma~3.6]{DBLP:journals/tods/KoutrisW17} that if the attack graph contains a strong cycle, then it contains a strong cycle of length~$2$.
The main result in~\cite{DBLP:journals/tods/KoutrisW17} can now be stated. 

\begin{theorem}[\cite{DBLP:journals/tods/KoutrisW17}]\label{the:trichotomy}
For every query $q$ in $\sjfbcq$,
\begin{itemize} 
\item
if the attack graph of $q$ is acyclic, then $\cqa{q}$ is in $\FO$;
\item
if the attack graph of $q$ is cyclic but contains no strong cycle, then $\cqa{q}$ is $\logspace$-hard and in $\P$; and
\item
if the attack graph of $q$ contains a strong cycle, then $\cqa{q}$ is $\coNP$-complete.
\end{itemize}
\end{theorem}

\myparagraph{Sequential proof}
Let $q$ be a query in $\sjfbcq$.
Let $\fd{Z}{w}$ be a functional dependency for $q$ with a singleton right-hand side (where set delimiters $\{$ and $\}$ are omitted).
A {\em sequential proof for $\fd{Z}{w}$\/} is a (possibly empty) sequence
$F_{1},F_{1},\dots,F_{\ell}$ of atoms in $q$ such that for every $i\in\{1,\dots,\ell\}$, 
$\keyvars{F_{i}}\subseteq Z\cup\lrformula{\bigcup_{j=1}^{i-1}\atomvars{F_{j}}}$
and for some  $k\in\{1,\dots,\ell\}$, $w\in\atomvars{F_{k}}$.
Sequential proofs mimic the computation of a closure of a set of attributes with respect to a set of functional dependencies; see, e.g., \cite[p.~165]{DBLP:books/aw/AbiteboulHV95}.


\myparagraph{Notions from graph theory}
A directed graph is {\em strongly connected\/} if there is a directed path from any vertex to any other.
The maximal strongly connected subgraphs of a graph are vertex-disjoint and are called its {\em strong components\/}.
If $S_{1}$ and $S_{2}$ are strong components such that an edge leads from a vertex in $S_{1}$ to a vertex in $S_{2}$,
then $S_{1}$ is a {\em predecessor\/} of $S_{2}$ and $S_{2}$ is a {\em successor\/} of $S_{1}$.
A strong component is called {\em initial\/} if it has no predecessor.
For a directed graph, 
we define the length of a directed path as the number of edges it contains.
A directed path or cycle without repeated vertices is called {\em elementary\/}.
If $G$ is a graph, then $V(G)$ denotes the vertex set of $G$,
and $E(G)$ denotes the edge set of $G$.


\myparagraph{Linear stratified Datalog}
We assume that the reader is familiar with the syntax and semantics of Datalog.
We fix some terminology for Datalog programs, most of which is standard.
A predicate that occurs in the head of some rule is called an {\em intentional database predicate\/} (IDB predicate); otherwise it is an {\em extensional database predicate\/} (EDB predicate).

The following definition is slightly adapted from~\cite[p.~185]{DBLP:series/txtcs/GradelKLMSVVW07}.
A {\em stratified Datalog program\/} is a sequence $P=(P_{0},\dots,P_{r})$ of basic Datalog programs, which are called the {\em strata\/} of $P$, such that each of the IDB predicates of $P$ is an IDB predicate of precisely one stratum $P_{i}$ and can be used as an EDB predicate (but not as an IDB predicate) in higher strata $P_{j}$ where $j>i$. In particular, this means that
\begin{enumerate}
\item
if an IDB predicate of stratum $P_{j}$ occurs {\em positively\/} in the body of a rule of stratum $P_{i}$, then $j\leq i$, and
\item
if an IDB predicate of stratum $P_{j}$ occurs {\em negatively\/} in the body of a rule of stratum $P_{i}$, then $j<i$.
\end{enumerate}
Stratified Datalog programs are given natural semantics using semantics for Datalog programs for each $P_{i}$, where the IDB predicates of a lower stratum are viewed as EDB predicates for a higher stratum.
A rule is {\em recursive\/} if its body contains an IDB predicate of the same stratum. 


A stratified Datalog program is {\em linear\/} if in the body of each rule there is at most one occurrence of an IDB predicate of the same stratum (but there may be arbitrarily many occurrences of IDB predicates from lower strata).

\myparagraph{Symmetric stratified Datalog}
Assume that some stratum of a linear stratified Datalog program contains a recursive rule 

$$
L_{0}\leftarrow L_{1},L_{2},\dots,L_{m},\neg L_{m+1},\dots,\neg L_{n}
$$
such that $L_{1}$ is an IDB predicate of the same stratum. 
Then, since the program is linear, each predicate among $L_{2},\dots,L_{n}$ is either an EDB predicate or an IDB predicate of a lower stratum.
Such rule has a {\em symmetric rule\/}:

$$
L_{1}\leftarrow L_{0},L_{2},\dots,L_{m},\neg L_{m+1},\dots,\neg L_{n}.
$$
A stratified Datalog program is {\em symmetric\/} if it is linear and the symmetric of any recursive rule is also a rule of the program.

It is known (see, for example, \cite[Proposisition~3.3.72]{DBLP:series/txtcs/GradelKLMSVVW07}) that linear stratified Datalog is equivalent to Transitive Closure Logic.
The data complexity of linear stratified Datalog is in $\NL$ (and is complete for $\NL$).
A symmetric Datalog program can be evaluated in logarithmic space~\cite{DBLP:conf/lics/EgriLT07} and cannot express directed reachability~\cite{DBLP:conf/icalp/EgriLT08}.

We will assume that given a (extensional or intentional) predicate $P$ of some arity $2\ell$, we can express the following query 
(let $\vec{x}=\tuple{x_{1},\dots,x_{\ell}}$,
$\vec{y}=\tuple{y_{1},\dots,y_{\ell}}$, and
$\vec{z}=\tuple{z_{1},\dots,z_{\ell}}$):
\begin{equation}\label{eq:groupby}
\{\vec{x},\vec{y}
\mid P(\vec{x},\vec{y})
\land
\forall z_{1}\dotsm\forall z_{\ell}
\lrformula{
P(\vec{x},\vec{z})
\rightarrow
\vec{y}\leq_{\ell}\vec{z}
}\},
\end{equation}
where $\leq_{\ell}$ is a total order on $\dom^{\ell}$.
Informally, the above query groups by the $\ell$ leftmost positions,
and, within each group, takes the smallest (with respect to $\leq_{\ell}$) value for the remaining positions.
Such query will be useful in Section~\ref{sec:eliminationofmcycle}, where $P$ encodes an equivalence relation on a finite subset of $\dom^{\ell}$, and the query~(\ref{eq:groupby}) allows us to deterministically choose a representative in each equivalence class.
The order $\leq_{\ell}$ can be defined as the lexicographical order on $\dom^{\ell}$ induced by the linear order on $\dom$.
For example, for $\ell=2$, the lexicographical order is defined as $(y_{1},y_{2})\leq_{2}(z_{1},z_{2})$ if $y_{1}<z_{1}\lor\lrformula{\lrformula{y_{1}=z_{1}}\land\lrformula{y_{2}\leq z_{2}}}$.
Nevertheless, our results do not depend on how the order $\leq_{\ell}$ is defined.
Moreover, all queries in our study will be order-invariant in the sense defined in~\cite{DBLP:journals/tocl/GroheS00}.
The order is only needed in the proof of Lemma~\ref{lem:toT} to pick, in a deterministic way, an identifier from a set of candidate identifiers. 
In Datalog, we use the following convenient syntax for~(\ref{eq:groupby}):
$$
\begin{datalogpgmns}
\predicate{Answer}(\vec{x},\min(\vec{y}))
&
P(\vec{x},\vec{y}).
\end{datalogpgmns}
$$
Such rule will always be non-recursive.
Most significantly, if we extend a logspace fragment of stratified Datalog with queries of the form~(\ref{eq:groupby}), the extended fragment will also be in logspace. 
Therefore, assuming queries of the form~(\ref{eq:groupby}) is harmless for our complexity-theoretic purposes.
We use $\ssdatalog$ for symmetric stratified Datalog,
and $\ssdatalogmin$ for symmetric stratified Datalog that allows queries of the form~(\ref{eq:groupby}).

\section{The Main Theorem and an Informal Guide of its Proof}\label{sec:gt}

\begin{figure*}[t!]
\captionsetup[subfigure]{justification=centering}
\begin{subfigure}[t]{0.35\textwidth}\centering
\includegraphics[scale=0.8]{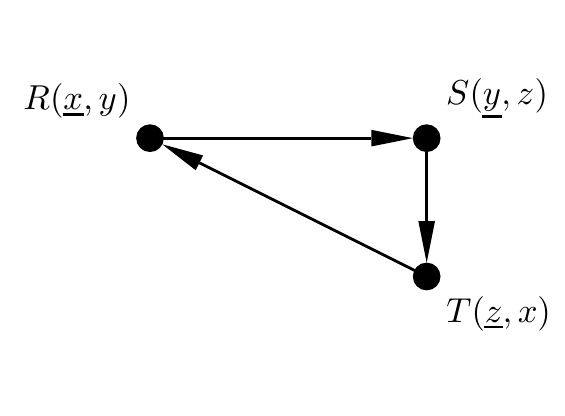}
\caption{\mgraph.}\label{fig:gtm}
\end{subfigure}
\begin{subfigure}[t]{0.6\textwidth}\centering
\includegraphics[scale=0.8]{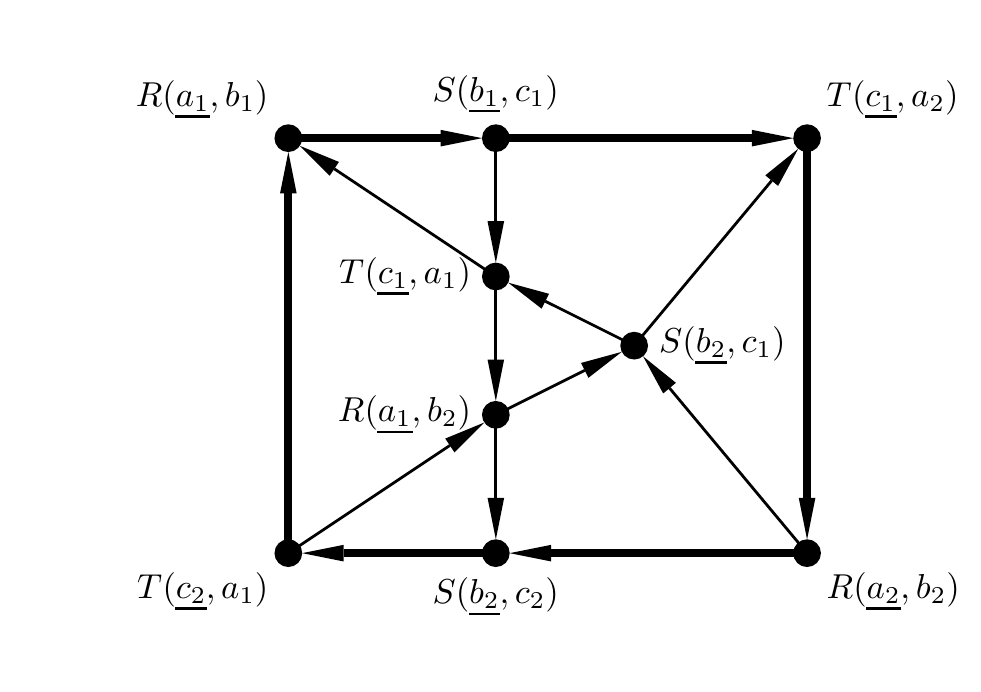}
\caption{$\cmhook{\qcycle}$-graph.}\label{fig:gthook}
\end{subfigure}
\newline
\begin{subfigure}[t]{\textwidth}\centering
\includegraphics[scale=0.8]{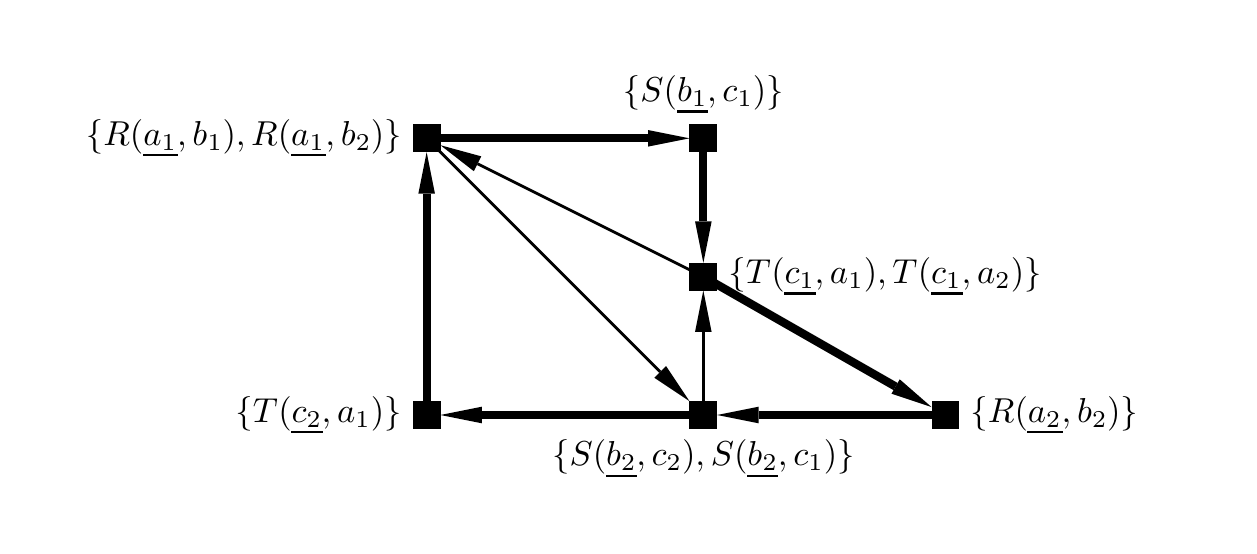}
\caption{Block-quotient graph.}\label{fig:gtbq}
\end{subfigure}
\caption{Examples of three different graphs used in this paper: \mgraph, $\cmhook{}$-graph, block-quotient graph. }\label{fig:gt}
\end{figure*}

In this paper, we prove the following main result.

\begin{theorem}[Main Theorem]\label{the:effectivedichotomy}
For every query $q$ in $\sjfbcq$,
\begin{itemize}
\item
if the attack graph of $q$ contains a strong cycle, then $\cqa{q}$ is $\coNP$-complete; and 
\item
if the attack graph of $q$ contains no strong cycle, 
then $\cqa{q}$ is expressible in $\ssdatalogmin$ (and is thus in $\logspace$).
\end{itemize}
\end{theorem}

The above result is stronger than Theorem~\ref{the:dichotomy}, because it also provides an effective criterion for the dichotomy between $\coNP$-complete and expressibility in symmetric stratified Datalog.

Before we delve into the proof in the next sections, we start with a guided tour that introduces our approach 
in an informal way. The focus of this paper is a logspace algorithm for $\cqa{q}$ whenever $\cqa{q}$ is in $\P$ but not in $\FO$ (assuming $\P\neq\coNP$).
An exemplar query is  $\qcycle\defeq\{R(\underline{x},y)$, $S(\underline{y},z)$,  $T(\underline{z},x)\}$, which can be thought of as a cycle of length~$3$.
For the purpose of this example, let $q$ be a query in $\sjfbcq$ that includes $\qcycle$ as a subquery (i.e., $\qcycle\subseteq q$).

An important novel notion in this paper is the \mgraph of a query (see Section~\ref{sec:mgraph}).
The \mgraph of $\qcycle$ is shown in Fig.~\ref{fig:gtm}.
Informally, a directed edge from an atom $F$ to an atom $G$, denoted $F\markov G$, means that every variable that occurs in the primary key of $G$ occurs also in $F$. In Fig.~\ref{fig:gtm}, we have $T(\underline{z},x)\markov R(\underline{x},y)$, because $R$'s primary key (i.e., $x$) occurs in the $T$-atom; there is no edge from $R(\underline{x},y)$ to $T(\underline{z},x)$ because $z$ does not occur in the $R$-atom.
Intuitively, one can think of edges in the \mgraph as foreign-to-primary key joins.
In what follows, we focus on cycles in the \mgraph, called \mcycles.

Figure~\ref{fig:gthook} shows an \emph{instantiation} of the \mgraph, called  $\cmhook{\qcycle}$-graph (see Definitions~\ref{def:mhook} and~\ref{def:cmhook}).
We write $A\cmhook{\qcycle}B$ to denote an edge from fact $A$ to fact $B$.
Each triangle in the $\cmhook{\qcycle}$-graph of Fig.~\ref{fig:gthook} instantiates the query $\qcycle$; for example, the inner triangle is equal to $\theta(\qcycle)$ where $\theta$ is the valuation such that $\theta(xyz)=a_{1}b_{2}c_{1}$.
We call such a triangle a $1$-embedding (see Definition~\ref{def:cmhook}). 
Significantly, some edges are not part of any triangle.
For example, the edge $S(\underline{b_1},c_1)\cmhook{\qcycle}T(\underline{c_1},a_2)$ is not in a triangle, but is present because the primary key of $T(\underline{c_1},a_2)$ occurs in $S(\underline{b_1},c_1)$.

Let $\db$ be a database that is input to $\cqa{q}$ such that $\db$ contains (but is not limited to) all facts of Fig.~\ref{fig:gthook}.
Since $\qcycle$ is a subquery of $q$, $\db$ will typically contain other facts with relation names in $q\setminus\qcycle$.
Furthermore, $\db$ can contain $R$-facts, $S$-facts, and $T$-facts not shown in Fig.~\ref{fig:gthook}.
Then, $\db$ has at least $2^{3}=8$ repairs,
because Fig.~\ref{fig:gthook} shows two $R$-facts with primary key $a_1$, two $S$-facts with primary key $b_2$, and two $T$-facts with primary key~$c_1$.
Consider now the outermost elementary cycle (thick arrows) of length~$6$, i.e., the cycle using the vertices in $\rep\defeq\{R(\underline{a_1},b_1)$, $S(\underline{b_1},c_1)$, $T(\underline{c_1},a_2)$, $R(\underline{a_2},b_2)$, $S(\underline{b_2},c_2)$, $T(\underline{c_2},a_1)\}$, which will be called a $2$-embedding in Definition~\ref{def:cmhook}. 
One can verify that $\rep$ does not contain distinct key-equal facts and does not satisfy  $\qcycle$ (because the subgraph induced by $\rep$ has no triangle).
Let $\bfo$ be the database that contains $\rep$ as well as all facts of $\db$ that are key-equal to some fact in~$\rep$.
A crucial observation is that if $\db\setminus\bfo$ has a repair that falsifies $q$, then so has $\db$ (the converse is trivially true).
Indeed, if $\sep$ is a repair of $\db\setminus\bfo$ that falsifies $q$,
then $\sep\cup\rep$ is a repair of $\db$ that falsifies $q$.
Intuitively, we can add $\rep$ to $\sep$ without creating a triangle in the $\cmhook{\qcycle}$-graph (i.e., without making $\qcycle$ true, and thus without making $q$ true),
because the facts in $\rep$ form a cycle on their own and contain no outgoing $\cmhook{\qcycle}$-edges to facts in $\sep$. 
In Section~\ref{sec:garbageset}, the set $\bfo$ will be called a \emph{garbage set:} its facts can be thrown away without changing the answer to $\cqa{q}$.
Note that the $\cmhook{\qcycle}$-graph of Fig.~\ref{fig:gthook} contains other elementary cycles of length~$6$, which, however, contain distinct key-equal facts: for example, the cycle with vertices $R(\underline{a_1},b_1)$, $S(\underline{b_1},c_1)$, $T(\underline{c_1},a_1)$, $R(\underline{a_1},b_2)$, $S(\underline{b_2},c_2)$, $T(\underline{c_2},a_1)$ contains both $R(\underline{a_1},b_1)$ and $R(\underline{a_1},b_2)$.

Garbage sets thus arise from cycles in the $\cmhook{\qcycle}$-graph that \emph{(i)}~do not contain distinct key-equal facts, and \emph{(ii)}~are not triangles satisfying $\qcycle$.
To find such cycles, we construct the quotient graph of the $\cmhook{\qcycle}$-graph with respect to the equivalence relation ``is key-equal to.''
Since the equivalence classes with respect to ``is key-equal to'' are the \emph{blocks} of the database, we call this graph the \emph{block-quotient graph} (Definition~\ref{def:quotient}).
The block-quotient graph for our example is shown in Fig.~\ref{fig:gtbq}.
The vertices are database blocks; there is an edge from block $\block_{1}$ to $\block_{2}$ if the $\cmhook{\qcycle}$-graph contains an edge from some fact in $\block_{1}$ to some fact in $\block_{2}$.
The block-quotient graph contains exactly one elementary directed cycle of length~$6$ (thick arrows);
this cycle obviously corresponds to the outermost cycle of length~$6$ in the $\cmhook{\qcycle}$-graph.
A core result (Lemma~\ref{lem:longcycle}) of this article is a deterministic logspace algorithm for finding elementary cycles in the block-quotient graph whose lengths are strict multiples of the length of the underlying \mcycle.
In our example, since the \mcycle of $\qcycle$ has length~$3$, we are looking for cycles in the block-quotient graph of lengths $6,9,12,\dots$\ Note here that, since the $\cmhook{\qcycle}$-graph is tripartite, the length of any cycle in it must be a multiple of~$3$.
Our algorithm can be encoded in symmetric stratified Datalog.
This core algorithm is then extended to compute garbage sets (Lemma~\ref{lem:scrub}) for \mcycles.

In our example, $\qcycle$ is a subquery of $q$. In general, \mcycles will be subqueries of larger queries.
The facts that belong to the garbage set for an \mcycle can be removed,
but the other facts must be maintained for computations on the remaining part of the query,
and are stored in a new schema that replaces the relations in the \mcycle with a single relation (see Section~\ref{sec:eliminationofmcycle}).
In our example, this new relation has attributes for $x$, $y$, and $z$, and stores all triangles that are outside the garbage set for $\qcycle$.

We can now sketch our approach for dealing with queries $q$ such that $\cqa{q}$ is in $\P\setminus\FO$.
Lemma~\ref{lem:happy} tells us that such query $q$ will have an \mcycle involving two or more atoms of mode~$\mi$.
The garbage set of this \mcycle is then computed,
and the facts not in the garbage set will be stored in a single new relation of mode $\mi$ that replaces the \mcycle.
In this way, $\cqa{q}$ is reduced to a new problem $\cqa{q'}$, where $q'$ contains less atoms of mode $\mi$ than $q$.
Lemma~\ref{lem:toT} shows that this new problem will be in $\P$ and that our reduction can be expressed in symmetric stratified Datalog.
We can repeat this reduction until we arrive at a query $q''$ such that $\cqa{q''}$ is in $\FO$.

To conclude this guided tour, we point out the role of atoms of mode~$\mc$ in the computation of the \mgraph, which was not illustrated by our running example.
In the \mgraph of Fig.~\ref{fig:toT} \emph{(right)}, we have $S(\underline{y},z)\markov U(\underline{y,z,w},x)$, even though $w$ does not occur in the $S$-atom. The explanation is that the query also contains the consistent relation $T^{\mc}(\underline{z},w)$, which maps each $z$-value to a unique $w$-value. So, even though $w$ does not occur as such in $S(\underline{y},z)$, it is nevertheless uniquely determined by $z$. It is thus important to identify all relations of mode~$\mc$, which is the topic of the next section. 

\section{Saturated Queries}\label{sec:glimpse}


In this section, we show that we can safely extend a query $q$ with new consistent relations. To achieve this,
we need to identify a particular type of functional dependencies for $q$, which are called \emph{internal}.

\begin{definition}\label{def:saturated}
Let $q$ be a query in $\sjfbcq$. 
Let $\fd{Z}{w}$ be a functional dependency for $q$.
We say that $\fd{Z}{w}$ is {\em internal to $q$\/} if the following two conditions are satisfied:
\begin{enumerate}
\item
there exists a sequential proof for $\FD{q}\models\fd{Z}{w}$ such that no atom in the sequential proof attacks a variable in $Z\cup\{w\}$; and
\item
for some $F\in q$, $Z\subseteq\atomvars{F}$. 
\end{enumerate}
We say that $q$ is {\em saturated\/} if for every every functional dependency $\sigma$ that is internal to $q$, we have $\FD{\catoms{q}}\models\sigma$.
\end{definition}

\begin{example}
Assume $q=\{S_{1}(\underline{z},u)$, $S_{2}(\underline{u},w)$, $R_{1}(\underline{z},u')$, $R_{2}(\underline{u'},w)$, $T_{1}(\underline{u},v)$, $T_{2}(\underline{v},w)\}$.
We have that $\tuple{S_{1},S_{2}}$ is a sequential proof for $\FD{q}\models\fd{z}{w}$
in which neither $S_{1}$ nor $S_{2}$ attacks $z$ or $w$.\footnote{Note that we use relation names as a shorthand for atoms.}
Indeed, $S_{1}$ attacks neither $z$ nor $w$ because $z,w\in\keycl{S_{1}}{q}$.
$S_{2}$ attacks no variable because $\atomvars{S_{2}}\subseteq\keycl{S_{2}}{q}$.
It follows that the functional dependency $\fd{z}{w}$ is internal to~$q$.
\end{example}

The next key lemma shows that we can assume without loss of generality that every internal functional dependency $\fd{Z}{w}$ is satisfied, i.e., that every $Z$-value is mapped to a unique $w$-value.
Therefore, whenever $\fd{Z}{w}$ is internal,  we can safely extend $q$ with a new consistent relation $N^{\mc}(\underline{Z},w)$ that materializes the mapping from $Z$-values to $w$-values.
Continuing the above example, we would extend $q$ by adding a fresh atom $N^\mc(\underline{z},w)$.

\begin{lemma}\label{lem:corsaturate}
Let $q$ be a query in $\sjfbcq$.
It is possible to compute a query $q'$ in $\sjfbcq$ with the following properties:
\begin{enumerate} 
\item
there exists a first-order reduction from $\cqa{q}$ to $\cqa{q'}$;
\item
if the attack graph of $q$ contains no strong cycle,
then the attack graph of $\cqa{q'}$ contains no strong cycle; and
\item
$q'$ is saturated.
\end{enumerate}
\end{lemma}

\section{\mGraphs and $\mhook$-Graphs}\label{sec:mgraph}

In this section, we introduce the {\em \mgraph\/} of a query $q$ in $\sjfbcq$,
which is a generalization of the notion of Markov-graph introduced in~\cite{DBLP:journals/tods/KoutrisW17} (hence the use of the letter \mgraphsymbol).
An important new result, Lemma~\ref{lem:happy}, expresses a relationship between attack graphs and \mgraphs.
Finally, we define $\mhook$-graphs, which can be regarded as data-level instantiations of \mgraphs. 

\begin{figure*}\centering
\begin{tabular}{cc}
\includegraphics[scale=0.8,trim={0 0 1cm 0},clip]{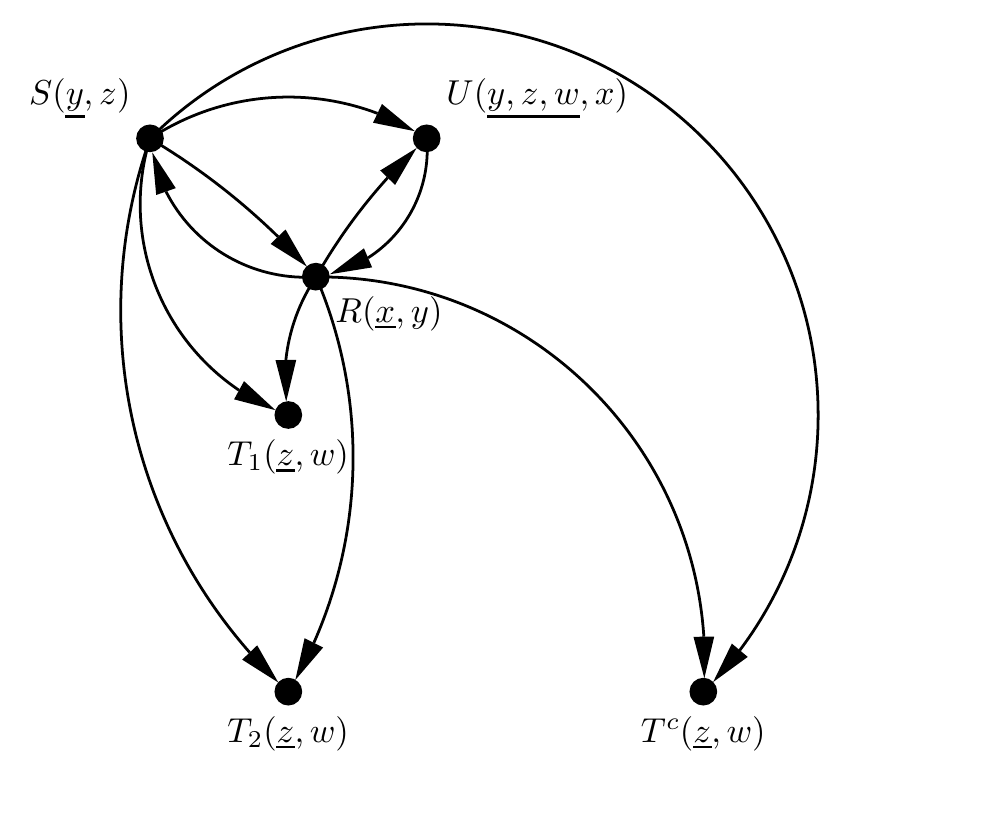}
&
\includegraphics[scale=0.8,trim={0 0 1cm 0},clip]{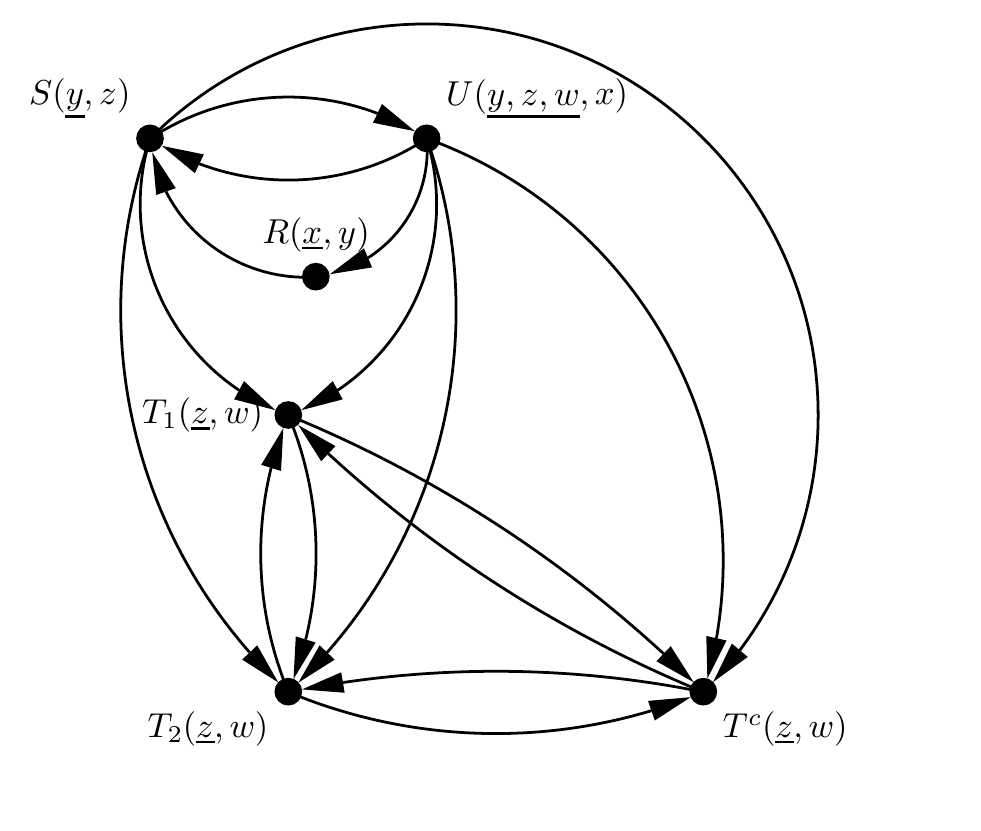}\\
Attack Graph
&
\mGraph
\end{tabular}
\vspace{-2ex}
\caption{Attack graph \emph{(left)} and \mgraph \emph{(right)} of the same query 
$\qpruning=\{R(\underline{x},y)$, $S(\underline{y},z)$, $U(\underline{y,z,w},x)$, 
$T_{1}(\underline{z},w)$, $T_{2}(\underline{z},w)$, $T^{\mc}(\underline{z},w)\}$.
It can be verified that all attacks are weak and that the query is saturated. 
The attack graph has an initial strong component containing three atoms ($R$, $S$, and $U$).
As predicted by Lemma~\ref{lem:happy}, the subgraph of the  \mgraph induced by $\{R,S,U\}$ is cyclic. 
}\label{fig:pruning}
\end{figure*}

\begin{definition}\label{def:mbis}
Let $q$ be a query in $\sjfbcq$ (which need not be saturated).
The {\em \mgraph\/} of $q$ is a directed graph whose vertices are the atoms of $q$.
There is a directed edge from $F$ to $G$ ($F\neq G$), denoted $F\markov G$,
if $\FD{\catoms{q}}\models\fd{\atomvars{F}}{\keyvars{G}}$.
A cycle in the \mgraph is called an \mcycle.
\end{definition}

\begin{example}
The notion of \mgraph is illustrated by Fig.~\ref{fig:pruning}.
We have $\catoms{\qpruning}=\{\fd{z}{w}\}$.
Since $\FD{\catoms{\qpruning}}\models\fd{\atomvars{S}}{\keyvars{U}}$, the \mgraph has a directed edge from $S$ to $U$.
\end{example}


\begin{lemma}\label{lem:happy}
Let $q$ be a query in $\sjfbcq$ such that $q$ is saturated and the attack graph of $q$ contains no strong cycle.
Let $\isc$ be an initial strong component in the attack graph of $q$ with $\card{\isc}\geq 2$.
Then, the \mgraph of $q$ contains a cycle all of whose atoms belong to $\isc$.
\end{lemma}


Given a query $q$, every database that instantiates the schema of $q$ naturally gives rise to an instantiation of the $\markov$-edges in $q$'s \mgraph, in a way that is captured by the following definition.

\begin{definition}\label{def:mhook}
The following notions are defined relative to a query $q$ in $\sjfbcq$ and a database $\db$.
The {\em $\mhook$-graph\/} of $\db$  is a directed graph whose vertices are the atoms of $\db$.
There is a directed edge from $A$ to $B$, denoted $A\mhook B$, if there exists a valuation $\theta$ over $\queryvars{q}$ and an edge $F\markov G$ in the \mgraph of $q$ such that $\theta(q)\subseteq\db$,  $A=\theta(F)$, and $B\keyequal\theta(G)$. 
A cycle in the $\mhook$-graph is also called a $\mhook$-cycle.
\end{definition}

\begin{figure}\centering
\begin{tabular}{cc}
\begin{small}
$
\setlength{\arraycolsep}{1ex}
\begin{array}{cc}
\begin{array}{c|cc}
S & \underline{y} & z\bigstrut\\\cline{2-3}
  & I & 1\bigstrut\\
  & I & 2\\
  & I & 3\\\cdashline{2-3}
\end{array}
&
\begin{array}{c|cc}
T^{\mc}& \underline{z} & w\\\cline{2-3}
  & 1 & a\bigstrut\\\cdashline{2-3}
  & 2 & b\bigstrut\\\cdashline{2-3}
\end{array}
\\
\\
\begin{array}{c|*{4}{c}}
U & \underline{y} & \underline{z} & \underline{w} & x\bigstrut\\\cline{2-5}
  & I & 1 & a & \chi\bigstrut\\\cdashline{2-5}
  & I & 2 & b & \chi\bigstrut\\\cdashline{2-5}
\end{array}
&
\begin{array}{c|cc} 
R & \underline{x} & y\bigstrut\\\cline{2-3}
  & \chi & I\bigstrut\\\cdashline{2-3}
\end{array}
\end{array}
$
\end{small}
&
\raisebox{-18ex}{
\includegraphics[scale=0.8]{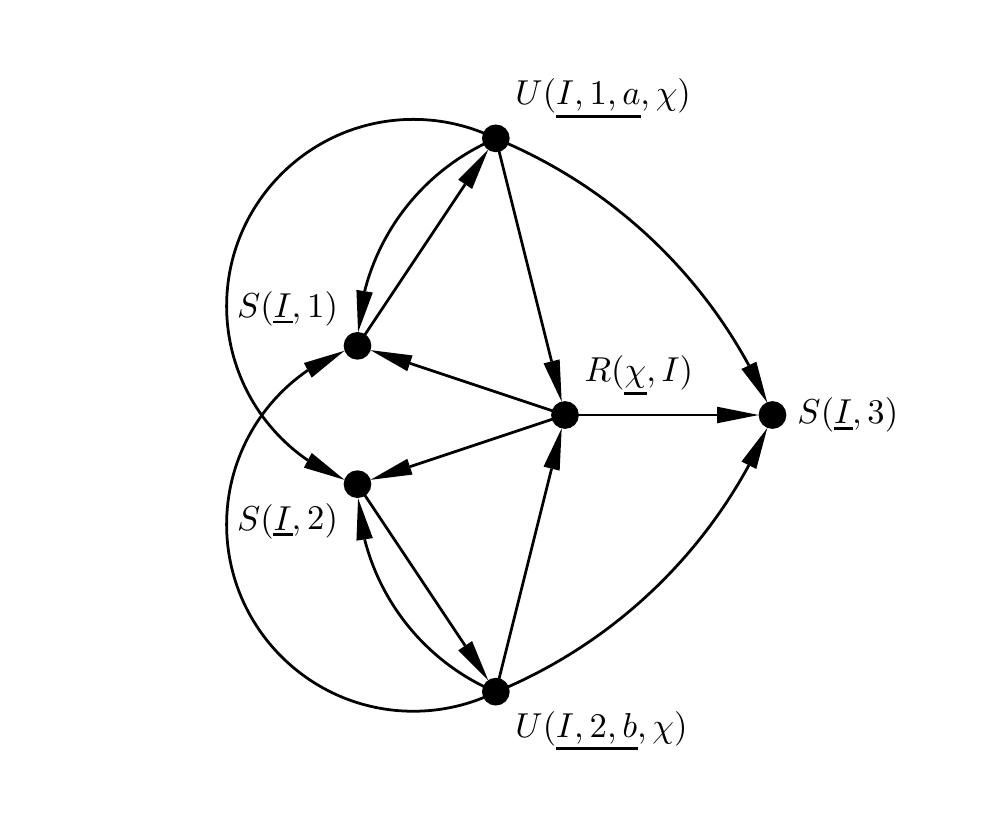}}
\end{tabular}
\vspace{-2ex}
\caption{
\emph{Left:} Database that is input to $\cqa{\qpruning}$ for the query $\qpruning$ in Fig.~\ref{fig:pruning}.
The relations for $T_{1}$ and $T_{2}$, which are identical to the relation for $T^{\mc}$, have been omitted.
\emph{Right:}
The $\mhook$-graph from which, for readability reasons,
$T_{1}$-facts, $T_{2}$-facts, and $T^{\mc}$-facts have been omitted.
}\label{fig:mhookgraph}
\end{figure}

The notion of $\mhook$-graph is illustrated by Fig.~\ref{fig:mhookgraph}.


\begin{lemma}\label{lem:hockey}
Let $q\in\sjfbcq$ and let $\db$ be a database.
Let $A,B\in\db$ and $F,G\in q$.
\begin{enumerate}
\item
if $A\mhook B$, then $A\mhook B'$ for all $B'\in\theblock{B}{\db}$; 
\item
if $A\mhook B$ and  $A\mhook B'$ and $\qatom{B}{q}=\qatom{B'}{q}$, then $B\keyequal B'$.
\end{enumerate}
\end{lemma}


\section{Garbage Sets}\label{sec:garbageset}

Let $\db$ be a database that is an input to $\cqa{q}$ with $q\in\sjfbcq$.
In this section, we show that it is generally possible to downsize $\db$ by deleting blocks from it without changing the answer to $\cqa{q}$.
That is, if the downsized database has a repair falsifying $q$, then so does the original database (the converse holds trivially true).
Intuitively, the deleted blocks can be considered as ``garbage'' for the problem $\cqa{q}$.

\begin{definition}\label{def:garbage}
The following definition is relative to a fixed query~$q$ in $\sjfbcq$.
Let $q_{0}\subseteq q$.
Let $\db$ be a database.
We say that a subset $\bfo$ of $\db$ is a {\em garbage set for $q_{0}$ in $\db$\/} if the following conditions are satisfied:
\begin{enumerate}
\item
for every $A\in\bfo$, we have that $\qatom{A}{q}\in q_{0}$ and $\theblock{A}{\db}\subseteq\bfo$; and
\item
there exists a repair $\rep$ of $\bfo$ such that for every valuation $\theta$ over $\queryvars{q}$, if $\theta(q)\subseteq\formula{\db\setminus\bfo}\cup\rep$,
then $\theta(q_{0})\cap\rep=\emptyset$ (and thus $\theta(q_{0})\cap\bfo=\emptyset$).
\end{enumerate}
\end{definition}

The first condition in the above definition says that the relation names of facts in $\bfo$ must occur in $q_{0}$, and that every block of $\db$ is either included in or disjoint with $\bfo$.
The second condition captures the crux of the definition and was illustrated in Section~\ref{sec:gt}.


We now show a number of useful properties of garbage sets that are quite intuitive.
In particular, by Lemma~\ref{lem:closedunderunion}, there exists a unique maximum (with respect to~$\subseteq$) garbage set for $q_{0}$ in $\db$, which will be called {\em the maximum garbage set for $q_{0}$ in $\db$\/}. 

\begin{lemma}\label{lem:closedunderunion}
Let $q$ be a query in $\sjfbcq$, and let $q_{0}\subseteq q$.
Let $\db$ be a database.
If $\bfo_{1}$ and $\bfo_{2}$ are garbage sets for $q_{0}$ in $\db$,
then $\bfo_{1}\cup\bfo_{2}$ is a garbage set for $q_{0}$ in $\db$.  
\end{lemma}


\begin{lemma}\label{lem:nomenestomen}
Let $q$ be a query in $\sjfbcq$, and let $q_{0}\subseteq q$.
Let $\db$ be a database.
Let $\bfo$ be a garbage set for $q_{0}$ in $\db$.
Then, every repair of $\db$ satisfies $q$ if and only if every repair of $\db\setminus\bfo$ satisfies $q$ (i.e., $\db$ and $\db\setminus\bfo$ agree on their answer to $\cqa{q}$).
\end{lemma}


\begin{lemma}\label{lem:together}
Let $q$ be a query in $\sjfbcq$, and let $q_{0}\subseteq q$.
Let $\db$ be a database.
Let $\bfo$ be a garbage set for $q_{0}$ in $\db$.
Then, 
every garbage set for $q_{0}$ in $\db\setminus\bfo$ is empty
if and only if
$\bfo$ is the maximal garbage set for $q_{0}$ in $\db$.
\end{lemma}

\section{Garbage Sets for \mCycles}\label{sec:scrub}

In this section, we bring together notions of the two preceding sections.
We focus on queries $q$ in $\sjfbcq$ whose \mgraph has a cycle $C$.
From here on, if $C$ is an elementary cycle in the \mgraph of some query $q$ in $\sjfbcq$, then the subset of $q$ that contains all (and only) the atoms of $C$, is also denoted by $C$.

Section~\ref{sec:garbageofmcycle} shows a procedural characterization of the maximal garbage set for $C$. 
Section~\ref{sec:garbageindatalog} shows that the maximal garbage set for $C$ can be computed in symmetric stratified Datalog.
Finally, Section~\ref{sec:eliminationofmcycle} shows a reduction, expressible in symmetric stratified Datalog, that replaces $C$ with a single atom. 

\subsection{Characterizing Garbage Sets for \mCycles}\label{sec:garbageofmcycle}

We define how a given \mcycle $C$ of length~$k$ can be instantiated by cycles in the $\mhook$-graph, called embeddings, whose lengths are multiples of~$k$.

\begin{definition}\label{def:cmhook}
Let $q$ be a query in $\sjfbcq$.
Let $\db$ be a database.
Let $C$ be an elementary directed cycle in the \mgraph of $q$.
The cycle $C$ naturally induces a subgraph of the $\mhook$-graph, as follows:
the vertex set of the subgraph contains all (and only) the facts $A$ of $\db$ such that $\qatom{A}{q}$ is an atom in $C$; there is a directed edge from $A$ to $B$, denoted $A\cmhook{C}B$, if $A\mhook B$ and the cycle $C$ contains a directed edge from $\qatom{A}{q}$ to $\qatom{B}{q}$.

Let $k$ be the length of $C$.
Obviously, the length of every $\cmhook{C}$-cycle must be a multiple of~$k$. 
Let $n$ be a positive integer.
An {\em $n$-embedding of $C$ in $\db$\/} (or simply {\em embedding\/} if the value $n$ is not important) is an elementary $\cmhook{C}$-cycle of length $nk$ containing no two distinct key-equal facts. 
A $1$-embedding of $C$ in $\db$ is said to be {\em relevant\/} if there exists a valuation $\theta$ over $\queryvars{q}$ such that $\theta(q)\subseteq\db$ and $\theta(q)$ contains every fact of the $1$-embedding; otherwise the  $1$-embedding is said to be {\em irrelevant\/}.
\end{definition}

Let $C$ and $q$ be as in Definition~\ref{def:cmhook}, 
and let $\db$ be a database.
There exists an intimate relationship between garbage sets for $C$ in $\db$ and different sorts of embeddings. 
\begin{itemize}
\item
Let $A\in\db$ such that $\qatom{A}{q}$ belongs to $C$.
If $A$ belongs to some relevant $1$-embedding of $C$ in $\db$,
then $A$ will have an outgoing edge in the $\cmhook{C}$-graph.
If $A$ does not belong to some relevant $1$-embedding of $C$ in $\db$,
then $A$ will have no outgoing edge in the $\cmhook{C}$-graph,
and $\theblock{A}{\db}$ is obviously a garbage set for $C$ in $\db$.
\item
Every irrelevant $1$-embedding of $C$ in $\db$ gives rise to a garbage set.
To illustrate this case, let $C=\{R(\underline{x},y,z)$, $S(\underline{y},x,z)\}$.
Assume that $R(\underline{a},b,1)\cmhook{C} S(\underline{b},a,2)\cmhook{C} R(\underline{a},b,1)$ is a $1$-embedding of $C$ in $\db$.
This $1$-embedding is irrelevant, because $1\neq 2$.
It can be easily seen that $R(\underline{a},b,\blockfiller)\cup S(\underline{b},a,\blockfiller)$ is a garbage set for $q$ in $\db$.
\item
Every $n$-embedding of $C$ in $\db$ with $n\geq 2$ gives rise to a garbage set.
This was illustrated in Section~\ref{sec:gt} by means of the outermost cycle of length~$6$ in Fig.~\ref{fig:gthook}, which is a $2$-embedding of $\{R(\underline{x},y)$, $S(\underline{y},z)$,  $T(\underline{z},x)\}$. 
\end{itemize}
These observations lead to the following lemma which provides a procedural characterization of the maximal garbage set for $C$ in a given database.

\begin{lemma}\label{lem:scrubalgo}
Let $q$ be a query in $\sjfbcq$.
Let $C=F_{0}\markov F_{1}\markov\dotsm\markov F_{k-1}\markov F_{0}$ be an elementary cycle of length~$k$ $(k\geq 2$) in the \mgraph of~$q$.
Let $\db$ be a database.
Let $\bfo$ be a minimal (with respect to $\subseteq$) subset of $\db$ satisfying the following conditions:
\begin{enumerate}
\item\label{it:scrubnotdangling}
the set $\bfo$ contains every fact $A$ of $\db$ with $\qatom{A}{q}\in\{F_{0},\dots,F_{k-1}\}$ such that 
$A$ has zero outdegree in the $\cmhook{C}$-graph;
\item\label{it:scruboneembedding}
the set $\bfo$ contains every fact that belongs to some irrelevant $1$-embedding of $C$ in $\db$;
\item\label{it:scrubmultiple}
the set $\bfo$ contains every fact that belongs to some $n$-embedding of $C$ in $\db$ with $n\geq 2$;
\item\label{it:recursive}\emph{Recursive condition:}
if $\bfo$ contains some fact of a relevant $1$-embedding of $C$ in $\db$, 
then $\bfo$ contains every fact of that $1$-embedding; and
\item\label{it:keyequalclosure}\emph{Closure under ``is key-equal to'':}
if $\bfo$ contains some fact $A$, then $\bfo$ includes $\theblock{A}{\db}$. 
\end{enumerate}
Then, $\bfo$ is the maximal garbage set for $C$ in $\db$.
\end{lemma}

\begin{corollary}\label{cor:scrubcomponent}
Let $C$ be an elementary cycle in the \mgraph of a query $q$ in $\sjfbcq$.
Let $\isc$ be a strong component in the $\cmhook{C}$-graph of a database $\db$.
If some fact of $\isc$ belongs to the maximal garbage set for $C$ in $\db$,
then every fact of $\isc$ belongs to the maximal garbage set for $C$ in~$\db$.
\end{corollary}




\subsection{Computing Garbage Sets for \mCycles}\label{sec:garbageindatalog}

In this section, we translate  Lemma~\ref{lem:scrubalgo} into a Datalog program that computes, in deterministic logspace, the maximal garbage set for an \mcycle $C$. 
The main computational challenge lies in condition~\ref{it:scrubmultiple} of Lemma~\ref{lem:scrubalgo}, which adds to the maximal garbage set all facts belonging to some $n$-embedding with $n\geq 2$, where the value of $n$ is not upper bounded.
Such $n$-embeddings can obviously be computed in nondeterministic logspace by using directed reachability in the $\cmhook{C}$-graph. 
This section shows a trick that allows doing the computation by using only \emph{undirected} reachability, which, by the use of Reingold's algorithm~\cite{DBLP:journals/jacm/Reingold08}, will lead to an algorithm that runs in deterministic logspace. 

By Corollary~\ref{cor:scrubcomponent}, instead of searching for $n$-embeddings, $n\geq 2$, it suffices to search for strong components of the $\cmhook{C}$-graph containing such $n$-embeddings. 
These strong components can be recognized by a first-order reduction to the following problem, called $\problem{LONGCYCLE}(k)$, which is in logspace by Lemma~\ref{lem:longcycle}.


\begin{definition}\label{def:longcycle}
A directed graph is {\em $k$-partite\/} if its vertices can be partitioned into $k$ sets $V_{0}$, $V_{1}$, \dots, $V_{k-1}$ such that for every directed edge $(u,v)$ in $E$, if $u\in V_{i}$, then $v\in V_{\formula{i+1}\mod k}$. 
For every positive integer $k$, $\problem{LONGCYCLE}(k)$ is the following problem.
\begin{description}
\item[Problem] $\problem{LONGCYCLE}(k)$
\item[Instance]
A connected $k$-partite directed graph $G=(V,E)$ such that every edge of $E$ belongs to a directed cycle of length~$k$.
\item[Question]
Does $G$ have an elementary directed cycle of length at least~$2k$?
\end{description}
\end{definition}

\begin{lemma}\label{lem:longcycle}
$\problem{LONGCYCLE}(k)$ is in $\logspace$ for every positive integer $k$.
\end{lemma}
\begin{proof}\emph{(Sketch)}
Let $G=(V,E)$ be an instance of $\problem{LONGCYCLE}(k)$.
A cycle of length~$k$ in $G$ is called a {\em $k$-cycle\/}.
Let $\po{G}$ be the undirected graph whose vertices are the $k$-cycles of $G$;
there is an undirected edge between two vertices if their $k$-cycles have an element in common.
The full proof in Appendix~\ref{app:scrub} shows that $G$ has an elementary directed cycle of length $\geq 2k$ if and only if one of the following conditions is satisfied:
\begin{itemize}
\item
for some $n$ such that $2\leq n\leq 2k-3$,
$G$ has an elementary directed cycle of length $nk$; or
\item
$\po{G}$ has a chordless cycle (i.e., a cycle without cycle chord) of length $\geq 2k$.
\end{itemize}
The first condition can be tested in $\FO$;
the second condition can be reduced to an undirected connectivity problem, which is in logspace~\cite{DBLP:journals/jacm/Reingold08}.
\end{proof}

To use Lemma~\ref{lem:longcycle}, we take a detour via the quotient graph of the $\cmhook{C}$-graph relative to the equivalence relation ``is key-equal to.''

\begin{definition}\label{def:quotient}
Let $q$ be a query in $\sjfbcq$.
Let $\db$ be a database.
Let $C$ be an elementary directed cycle of length $k\geq 2$ in the \mgraph of $q$.
The {\em block-quotient graph\/} is the quotient graph of the $\cmhook{C}$-graph of $\db$ with respect to the equivalence relation~$\keyequal$. 
\footnote{The quotient graph of a directed graph $G=(V,E)$ with respect to an equivalence relation~$\equiv$ on $V$ is a directed graph whose vertices are the equivalence classes of~$\equiv$; there is a directed edge from class $A$ to class $B$ if $E$ has a directed edge from some vertex in $A$ to some vertex in $B$. }
\end{definition}

The block-quotient graph of a database can obviously be constructed in $\FO$.
The strong components of the $\cmhook{C}$-graph that contain some $n$-embedding of $C$, $n\geq 2$, can then be recognized in logspace by executing the algorithm for $\problem{LONGCYCLE}(k)$ on the block-quotient graph.
Moreover, an inspection of the proof of Lemma~\ref{lem:longcycle} reveals that $\problem{LONGCYCLE}(k)$ can be expressed in symmetric stratified Datalog, which is the crux in the proof of the following lemma.

\begin{lemma}\label{lem:scrub}
Let $q$ be a query in $\sjfbcq$.
Let $C$ be an elementary cycle of length~$k$ $(k\geq 2$) in the \mgraph of~$q$.
There exists a program in $\ssdatalog$ that takes a database $\db$ as input and returns, as output, the maximal garbage set for $C$ in $\db$. 
\end{lemma}

\subsection{Elimination of \mCycles}\label{sec:eliminationofmcycle}

Given a database $\db$,
the Datalog program of Lemma~\ref{lem:scrub} allows us to compute the maximal garbage set $\bfo$ for $C$ in $\db$.
The $\cmhook{C}$-graph of $\db'\defeq\db\setminus\bfo$ will be a set of strong components, all initial, each of which is a collection of relevant $1$-embeddings of $C$ in $\db'$. 
The following Lemma~\ref{lem:toT} introduces a reduction that encodes this $\cmhook{C}$-graph by means of a fresh atom $T(\underline{u},\vec{w})$, where $\sequencevars{\vec{w}}=\queryvars{C}$ and $u$ is a fresh variable. 
Whenever $\theta(q)\subseteq\db'$ for some valuation $\theta$ over $\queryvars{q}$, the reduction will add to the database a fact $T(\underline{\cid},\theta(\vec{w}))$ where $\cid$ is an identifier for the strong component (in the $\cmhook{C}$-graph) that contains $\theta(C)$.
The construction is illustrated by Fig.~\ref{fig:toT}.
The following lemma captures this reduction and states that it \emph{(i)}~is expressible in symmetric stratified Datalog, and \emph{(ii)}~does not result in an increase of computational complexity.

\begin{figure}[t]\centering
\begin{tabular}{ll}
\raisebox{-16ex}{
\includegraphics[scale=0.8]{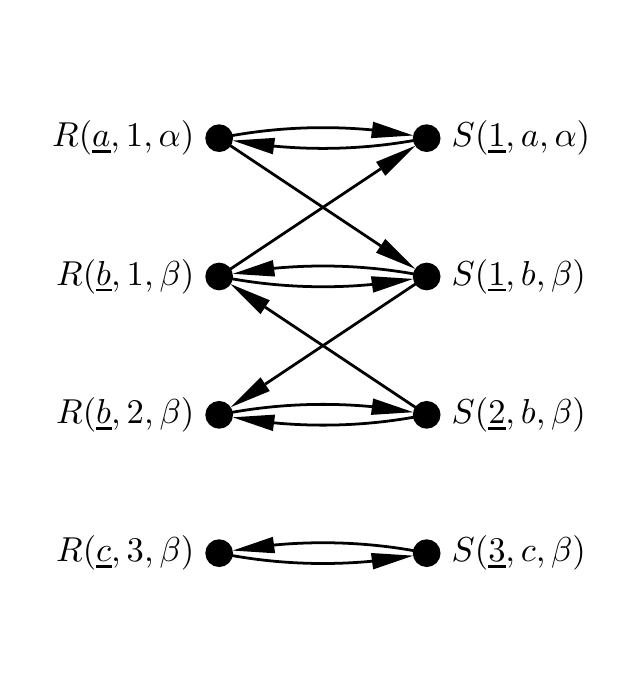}}
&
\begin{small}
$
\setlength{\arraycolsep}{1ex}
\begin{array}{lll}
\begin{array}{c|c;{1pt/1pt}*{3}{c}}
T & \underline{u} & x & y & z\bigstrut\\\cline{2-5}
  & a & a & 1 & \alpha\bigstrut\\
  & a & b & 1 & \beta\bigstrut\\
  & a & b & 2 & \beta\bigstrut\\\cdashline{2-5}
  & c & c & 3 & \beta\bigstrut
\end{array}
&
\begin{array}{c|cc}
N_{1}^{\mc} & \underline{x} & u\bigstrut\\\cline{2-3}
& a & a\\
& b & a\\
& c & c
\end{array}
&
\begin{array}{c|cc}
N_{2}^{\mc} & \underline{y} & u\bigstrut\\\cline{2-3}
& 1 & a\\
& 2 & a\\
& 3 & c
\end{array}
\end{array}
$
\end{small}
\end{tabular}
\vspace{-2ex}
\caption{
\emph{Left:}
Two strong components in  the $\cmhook{C}$-graph of a database for an \mcycle $R(\underline{x},y,z)\markov S(\underline{y},x,z)\markov R(\underline{x},y,z)$.
The maximal garbage set is empty.
\emph{Right:} Encoding of the relevant $1$-embeddings in each strong component.
The $u$-values $a$ and $c$ are used to identify the strong components, and are chosen as the smallest $x$-values in each strong component.
}\label{fig:toT}
\end{figure}

\begin{lemma}\label{lem:toT}
Let $q$ be a query in $\sjfbcq$.
Let $C=F_{0}\markov F_{1}\markov\dotsm\markov F_{k-1}\markov F_{0}$ with $k\geq 2$ be an elementary cycle in the \mgraph of $q$.
Let~$u$ be a variable such that $u\not\in\queryvars{q}$.
Let~$T$ be an atom with a fresh relation name such that
$\keyvars{T}=\{u\}$ and $\atomvars{T}=\queryvars{C}\cup\{u\}$.
Let~$p$ be a set containing, for every $i\in\{1,\dots,k\}$, an atom $N_{i}$ of mode~$\mc$ with a fresh relation name such that $\keyvars{N_{i}}=\keyvars{F_{i}}$ and $\atomvars{N_{i}}=\keyvars{F_{i}}\cup\{u\}$.
Then, 
\begin{enumerate}
\item
there exists a reduction from $\cqa{q}$ to $\cqa{\lrformula{q\setminus C}\cup\{T\}\cup p}$ that is expressible in $\ssdatalogmin$; and
\item
if the attack graph of $q$ contains no strong cycle
and some initial strong component of the attack graph contains every atom of $\{F_{0},F_{1},\dots,F_{k-1}\}$, 
then the attack graph of $\formula{q\setminus C}\cup\{T\}\cup p$ contains no strong cycle either.
\end{enumerate}
\end{lemma}
\begin{proof} \emph{(Crux)}
The full proof is in Appendix~\ref{app:scrub}.
The crux in the proof of the first item is the deterministic choice of $u$-values for $T$-blocks.
In Fig.~\ref{fig:toT}, for example, the $T$-block encoding the top strong component uses $u=a$, and the $T$-block encoding the bottom strong component uses $u=c$.
These $u$-values are the smallest $x$-values in the strong components, which can be obtained by the query~(\ref{eq:groupby}) introduced in Section~\ref{sec:preliminaries}.
\end{proof}

The proof of the main theorem, Theorem~\ref{the:effectivedichotomy}, is now fairly straightforward.


\begin{proof}[Proof of Theorem~\ref{the:effectivedichotomy}]
Let $q$ be a query in $\sjfbcq$.
We can assume that $q$ is saturated; if not, we first apply the reduction of Lemma~\ref{lem:corsaturate}.
The first item follows from~\cite[Theorem~3.2]{DBLP:journals/tods/KoutrisW17}.
In the remainder of the proof, we treat the case that the attack graph of $q$ contains no strong cycle.
The proof runs by induction on the number of atoms in $q$ that are of mode~$\mi$.
The desired result is obvious if $q$ contains no atom of mode~$\mi$.
Assume next that $q$ contains an atom of mode~$\mi$.
We distinguish two cases.

\myparagraph{Case that the attack graph contains an unattacked atom of mode~$\mi$}
If the attack graph of $q$ contains an unattacked atom of mode~$\mi$, say $R(\underline{\vec{x}},\vec{y})$, then it is known (see, e.g., \cite[Lemma~4.4]{DBLP:journals/tods/KoutrisW17}) that $q$ is true in every repair only if there exists a valuation $\theta$ over $\sequencevars{\vec{x}}$ such that $\theta(q)$ is true in every repair.
Obviously, if $q$ contains an atom $R(\underline{\vec{a}},\vec{y})$, where $\vec{a}$ contains no variables, then $q$ is true in every repair only if the input database contains a fact $R(\underline{\vec{a}},\vec{b})$ such that for every $A\in R(\underline{\vec{a}},\blockfiller)$,
there exists a valuation $\theta$ over $\sequencevars{\vec{y}}$ such that $R(\underline{\vec{a}},\theta(\vec{y}))=A$ and $\theta(q')$ is true in every repair, where $q'=q\setminus\{R(\underline{\vec{x}},\vec{y})\}$.
All this is expressible in first-order logic, and the induction hypothesis applies to $\theta(q')$.

\myparagraph{Case that all atoms of mode $\mi$ are attacked}
Then, every initial strong component of the attack graph contains at least two atoms.
By Lemma~\ref{lem:happy},
the \mgraph of $q$ has a cycle $C$ all of whose atoms belong to one and the same initial strong component of the attack graph of $q$.
By Lemma~\ref{lem:toT}, there exists a reduction, expressible in $\ssdatalogmin$, from $\cqa{q}$ to $\cqa{\formula{q\setminus C}\cup\{T\}\cup p}$ such that the attack graph of $\formula{q\setminus C}\cup\{T\}\cup p$ contains no strong cycle.
Since the number of atoms of mode~$\mi$ in $\formula{q\setminus C}\cup\{T\}\cup p$ is strictly less than in $q$ (because $C$ is replaced with $T$ and all atoms in $p$ have mode~$\mc$), by the induction hypothesis, $\cqa{\formula{q\setminus C}\cup\{T\}\cup p}$ is expressible in $\ssdatalogmin$.
It follows that $\cqa{q}$ is expressible in $\ssdatalogmin$.
\end{proof}

\section{Joins on Primary Keys}\label{sec:keyjoin}

It is common that the join condition in a join of two tables expresses a foreign-to-primary key match, i.e., the columns (called the foreign key) of one table  reference the primary key of another table.
In our setting, we have primary keys but no foreign keys.
Nevertheless, foreign keys can often be inferred from the query.
For example, in the following query, the variable $d$ in $\rmov$ references the primary key of $\rdir$:
$$
\{\rmov(\underline{m},t,\constant{1963},d), 
\rdir(\underline{d},\constant{Hitchcock},b)\}.
$$
Given relation schemas 
$\rmov(\underline{\att{M\#}},\att{Title},\att{Year}, \att{Director})$ and
$\rdir(\underline{\att{D\#}}, \att{Name}, \att{BirthYear})$,
this query asks whether there exists a movie released in 1963 and directed by Hitchcock.

The {\em key-join property} that we define below captures this common type of join. Informally, a query has the key-join property if whenever two atoms have a variable in common, then their set of shared variables is either equal to the set of primary-key variables of one of the atoms, or contains all primary-key variables of both atoms.

\begin{definition}\label{def:keyjoin}
We say that a query $q$ in $\sjfbcq$ has the {\em key-join property\/} if 
for all $F,G\in q$, either $\atomvars{F}\cap\atomvars{G}\in\{\emptyset,\keyvars{F},\keyvars{G}\}$ or $\atomvars{F}\cap\atomvars{G}\supseteq\keyvars{F}\cup\keyvars{G}$.
\end{definition}

Theorem~\ref{the:warmup} shows that if a query $q$ in $\sjfbcq$ has the key-join property,
then $\cqa{q}$ falls on the logspace side of the dichotomy of Theorem~\ref{the:effectivedichotomy}.

\begin{theorem}\label{the:warmup}
For every query $q$ in $\sjfbcq$ that has the key-join property, $\cqa{q}$ is expressible in $\ssdatalogmin$ (and is thus in $\logspace$).
\end{theorem}

It is worth noting that many of the queries covered by Theorem~\ref{the:warmup} have an acyclic attack graph as well, and thus even have a consistent first-order rewriting.  

\section{Conclusion}\label{sec:conclusion}

The main result of this paper is a theorem stating that for every query $q$ in $\sjfbcq$ (i.e., the class of self-join-free Boolean conjunctive queries),  $\cqa{q}$ is $\coNP$-complete or expressible in symmetric stratified Datalog (and thus in $\logspace$).
Since there exist queries $q\in\sjfbcq$ such that $\cqa{q}$ is $\logspace$-complete,
the logspace upper bound in Theorem~\ref{the:dichotomy} is tight.
The theorem thus culminates a long line of research that started with the ICDT~2005 paper of Fuxman and Miller~\cite{FUXMAN2005}.


An intriguing open problem is to extend these complexity results to Boolean conjunctive queries with self-joins and to $\ucq$. Progress in the latter problem may deepen our understanding of relationships between CQA and CSP, which were first discovered in~\cite{DBLP:conf/lics/Fontaine13}.


\newcommand{\etalchar}[1]{$^{#1}$}

\newpage
\appendix
 
\section{Overview of Different Graphs and Notations}\label{sec:notation}

 
 \noindent
\begin{tabular}{*{4}{|l}|}\hline
Graph & Vertices & Edge     & Short Description\bigstrut\\
      &          & Notation &                  \\\hline\hline
\smallcell{attack graph} & \smallcell{query atoms\bigstrut} & $F\attacks{q}G$ & \cell{See Section~\ref{sec:preliminaries}}\bigstrut\\\hline
\mgraph & \smallcell{query atoms\bigstrut} & $F\markov G$ & \cell{Definition~\ref{def:mbis}}\\\hline
$\mhook$-graph & \smallcell{database facts} & $A\mhook B$ & \cell{Definition~\ref{def:mhook}, data-level intantiation of the \mgraph}\bigstrut\\\hline
$\cmhook{C}$-graph & \smallcell{database facts} & $A\cmhook{C} B$ & \cell{Definition~\ref{def:cmhook}, subgraph of the $\mhook$-graph induced by an \mcycle $C$}\bigstrut\\\hline
\smallcell{block-quotient graph} & \smallcell{database blocks} & $(\block,\block')$ & \cell{Definition~\ref{def:quotient}, quotient graph of the $\cmhook{C}$-graph relative to the equivalence relation ``is key-equal to''}\bigstrut\\\hline
\end{tabular}

\medskip

{\@ifundefined{bigstrutjot}{\newdimen\bigstrutjot}{}\bigstrutjot=1.5pt
\noindent
\begin{tabular}{|c|l|}\hline
Notation & Meaning\bigstrut\\\hline\hline
$\keyvars{F}$ & the set of all variables occurring in the primary key of atom $F$\bigstrut\\
$\atomvars{F}$ & the set of all variables occurring in atom $F$\bigstrut\\
$\queryvars{q}$ & the set of all variables occurring in query $q$\bigstrut\\
$\keyequal$ & the equivalence relation ``is key-equal to'', e.g., $R(\underline{a},1)\keyequal R(\underline{a},2)$\bigstrut\\
$\repairs{\db}$ & the set of all repairs of a database $\db$\bigstrut\\
$\theblock{A}{\db}$ & the set of all facts in $\db$ that are key-equal to the fact $A$\bigstrut\\
$R(\underline{\vec{a}},\blockfiller)$ & the set of all database facts of the form $R(\underline{\vec{a}},\vec{b})$, for some $\vec{b}$\bigstrut\\
$\sjfbcq$ & the class of self-join-free Boolean conjunctive queries\bigstrut\\
$\ucq$ & the class of unions of conjunctive queries\bigstrut\\
$R^{\mc}$ & a relation name of mode~$\mc$, which must be interpreted by a consistent relation\bigstrut\\
$\catoms{q}$ & the set of all atoms of query $q$ having a relation name of mode~$\mc$\bigstrut\\
$\FD{q}$ & the set containing $\fd{\keyvars{F}}{\atomvars{F}}$ for every $F\in q$\bigstrut\\
$\keycl{F}{q}$ & the closure of $\keyvars{F}$ with respect to the FDs in $\FD{q\setminus\{F\}}\cup\FD{\catoms{q}}$\bigstrut\\
$\qatom{A}{q}$ & the atom of $q$ with the same relation name as the fact $A$\bigstrut\\
$V(G)$ & the vertex set of a graph $G$\bigstrut\\
$E(G)$ & the edge set of a graph $G$\bigstrut\\
$\uplus$ & a set union that happens to be disjoint\bigstrut\\
\hline
\end{tabular}}
 
\section{Proofs of Section~\ref{sec:glimpse}}

We use the following helping lemma.

\begin{lemma}\label{lem:saturate}
Let $q$ be a query in $\sjfbcq$.
Let $\fd{Z}{w}$ be a functional dependency that is internal to $q$.
Let $\vec{z}$ be a sequence of distinct variables such that $\sequencevars{\vec{z}}=Z$.
Let $q'=q\cup\{N^{\mc}(\underline{\vec{z}},w)\}$ where $N$ is a fresh relation name of mode~$\mc$. 
Then, 
\begin{enumerate}
\item
there exists a first-order reduction from $\cqa{q}$ to $\cqa{q'}$; and
\item
if the attack graph of $q$ contains no strong cycle,
then the attack graph of $\cqa{q'}$ contains no strong cycle.
\end{enumerate} 
\end{lemma}
\begin{proof}
\myparagraph{Proof of the first item}
By the second condition in Definition~\ref{def:saturated}, we can assume an atom $F\in q$ such that $Z\subseteq\atomvars{F}$.
Let $F_{1},F_{2},\dots,F_{\ell}$ be a sequential proof for $\FD{q}\models\fd{Z}{w}$
such that for every $i\in\{1,\dots,\ell\}$, for every $u\in Z\cup\{w\}$, $F_{i}\nattacks{q}u$.
It can be easily seen that for every $i\in\{0,\dots,\ell-1\}$, we have
\begin{equation}\label{eq:fdimpl}
\FD{\{F_{j}\}_{j=1}^{i}}\models\fd{Z}{\keyvars{F_{i+1}}}.
\end{equation}

Let $\db$ be a database that is the input to $\cqa{q}$.
We repeat the following ``purification" step:
If for two valuations over $\queryvars{q}$, denoted $\beta_{1}$ and $\beta_{2}$,  
we have $\beta_{1}(q),\beta_{2}(q)\subseteq\db$ and $\{\beta_{1},\beta_{2}\}\not\models\fd{Z}{w}$,
then we remove both the $F$-block containing $\beta_{1}(F)$ and
the $F$-block containing $\beta_{2}(F)$.
Note that $\beta_{1}(F)$ and $\beta_{2}(F)$ may be key-equal, and hence belong to the same $F$-block.

Assume that we apply this step on $\db'$ and obtain $\db''$.
We show that some repair of $\db'$ falsifies $q$ if and only if some repair of $\db''$ falsifies $q$.
The $\implies$-direction trivially holds true.
For the $\impliedby$-direction, let $\rep''$ be a repair of $\db''$ that falsifies $q$.
Assume, toward a contradiction, that every repair of $\db'$ satisfies $q$.
For every repair $\rep$,
define $\rifi{\rep}$ as the set of valuations over $Z\cup\{w\}$ containing $\theta$ if $\rep\models\theta(q)$. 
Let 
\[
\rep'=
\left\{
\begin{array}{ll}
\rep''\cup\{\beta_{j}(F)\} \mbox{\ for some $j\in\{1,2\}$} & \mbox{if $\beta_{1}(F)$ and $\beta_{2}(F)$ are key-equal}\\
 \rep''\cup\{\beta_{1}(F),\beta_{2}(F)\} & \mbox{otherwise}
\end{array}
\right.
\]
Note that if $\beta_{1}(F)$ and $\beta_{2}(F)$ are key-equal,
then we can choose either 
$\rep'=\rep''\cup\{\beta_{1}(F)\}$ or 
$\rep'=\rep''\cup\{\beta_{2}(F)\}$;
the actual choice does not matter.
Obviously, $\rep'$ is a repair of $\db'$.
Since we assumed that every repair of $\db'$ satisfies $q$,
we can assume a valuation $\alpha$ over $\queryvars{q}$ such that $\alpha(q)\subseteq\rep'$.
Since $\alpha(q)\nsubseteq\rep''$ (because $\rep''\not\models q$), it must be the case that for some $j\in\{1,2\}$, 
$\alpha(F)=\beta_{j}(F)$.
From $\sequencevars{\vec{z}}=Z\subseteq\atomvars{F}$,
it follows that $\alpha(\vec{z})=\beta_{j}(\vec{z})$.
From $\beta_{1}(\vec{z})=\beta_{2}(\vec{z})$,
it follows $\alpha(\vec{z})=\beta_{1}(\vec{z})$ and $\alpha(\vec{z})=\beta_{2}(\vec{z})$.
Since $\beta_{1}(w)\neq\beta_{2}(w)$, either $\alpha(w)\neq\beta_{1}(w)$ or $\alpha(w)\neq\beta_{2}(w)$ (or both).
Therefore, we can assume $b\in\{1,2\}$ such that $\alpha(w)\neq\beta_{b}(w)$. 
It will be the case that $\rifi{\rep'}=\{\alpha[Z\cup\{w\}]\}$.\footnote{Here, $\alpha[Z\cup\{w\}]$ is the restriction of $\alpha$ to $Z\cup\{w\}$;
this restriction is the identity on variables not in $Z\cup\{w\}$.}
Indeed, since $\alpha$ is an arbitrary valuation over $\queryvars{q}$ such that $\alpha(q)\subseteq\rep'$, it follows that for all valuations $\alpha_{1},\alpha_{2}$ over $\queryvars{q}$, if $\alpha_{1}(q),\alpha_{2}(q)\subseteq\rep'$,
then $\alpha_{1}(\vec{z})=\alpha_{2}(\vec{z})$ and thus, by~\cite[Lemma~4.3]{DBLP:journals/tods/Wijsen12}) and using that $\FD{q}\models\fd{Z}{w}$, 
we have  $\alpha_{1}(w)=\alpha_{2}(w)$.

We show that for all $i\in\{0,1,\dots,\ell\}$, there exists a pair
$(\rep'^{i},\alpha^{i})$
such that
\begin{enumerate}
\item\label{it:pone}
$\rep'^{i}$ is a repair of $\db'$;
\item \label{it:ptwo}
$\alpha^{i}$ is a valuation over $\queryvars{q}$ such that $\alpha^{i}(q)\subseteq\rep'^{i}$; 
\item\label{it:pthree}
$\alpha^{i}(\{F_{j}\}_{j=1}^{i})=\beta_{b}(\{F_{j}\}_{j=1}^{i})$ 
and 
$\alpha^{i}(\vec{z})=\beta_{b}(\vec{z})$ 
(thus $\alpha^{i}(\vec{z})=\alpha(\vec{z})$); 
\item\label{it:pfour}
$\alpha^{i}(w)=\alpha(w)$; and
\item\label{it:pfive}
$\rifi{\rep'^{i}}=\{\alpha[Z\cup\{w\}]\}$.
\end{enumerate}
The third condition entails that $\{\alpha^{i},\beta_{b}\}\models\FD{\{F_{j}\}_{j=1}^{i}}$.
From Equation~(\ref{eq:fdimpl}), it follows 
$\{\alpha^{i},\beta_{b}\}\models\fd{Z}{\keyvars{F_{i+1}}}$.
Then, from $\alpha^{i}(\vec{z})=\beta_{b}(\vec{z})$,
it follows that $\alpha^{i}$ and $\beta_{b}$ agree on all variables of $\keyvars{F_{i+1}}$.

The proof runs by induction on increasing $i$.
For the basis of the induction, $i=0$,
the desired result holds by choosing $\rep'^{0}=\rep'$ and $\alpha^{0}=\alpha$.

For the induction step, $i\rightarrow i+1$,
the induction hypothesis is that the desired pair $(\rep'^{i},\alpha^{i})$ exists for some $i\in\{0,1,\dots,\ell-1\}$.
Since $\alpha^{i}$ and $\beta_{b}$ agree on all variables of $\keyvars{F_{i+1}}$,
we have that $\alpha^{i}(F_{i+1})$ and $\beta_{b}(F_{i+1})$ are key-equal.
From $\beta_{b}(q)\subseteq\db'$,
it follows that $\beta_{b}(F_{i+1})\in\db'$.
Let $\rep'^{i+1}=\lrformula{\rep'^{i}\setminus\{\alpha^{i}(F_{i+1})\}}\cup\{\beta_{b}(F_{i+1})\}$, which is obviously a repair of $\db'$.
Since $F_{i+1}\nattacks{q}u$ for all $u\in Z\cup\{w\}$, $\rifi{\rep'^{i+1}}\subseteq\rifi{\rep'^{i}}$ by~\cite[Lemma~B.1]{DBLP:journals/tods/KoutrisW17}.
Since we assumed that every repair of $\db'$ satisfies $q$,
we have that $\rifi{\rep'^{i+1}}\neq\emptyset$, 
and thus $\rifi{\rep'^{i+1}}=\{\alpha[Z\cup\{w\}]\}$.
Hence, there exists a valuation $\alpha^{i+1}$ over $\queryvars{q}$ such that $\alpha^{i+1}(q)\subseteq\rep'^{i+1}$ and $\alpha^{i+1}[Z\cup\{w\}]=\alpha[Z\cup\{w\}]$, that is,
$\alpha^{i+1}(\vec{z})=\alpha(\vec{z})$ and $\alpha^{i+1}(w)=\alpha(w)$.
Since $\alpha(\vec{z})=\beta_{b}(\vec{z})$, we have $\alpha^{i+1}(\vec{z})=\beta_{b}(\vec{z})$.
We have thus shown that the pair $(\rep'^{i+1},\alpha^{i+1})$ satisfies items~\ref{it:pone}, \ref{it:ptwo}, \ref{it:pfour}, and~\ref{it:pfive} in the above five-item list;
we also have shown the second conjunct of item~\ref{it:pthree}.
In the next paragraph, we show that 
$\alpha^{i+1}(\{F_{j}\}_{j=1}^{i+1})=\beta_{b}(\{F_{j}\}_{j=1}^{i+1})$, i.e.,
the first conjunct of item~\ref{it:pthree}.

By the induction hypothesis, 
$\alpha^{i}(\{F_{j}\}_{j=1}^{i})=\beta_{b}(\{F_{j}\}_{j=1}^{i})$
and
$\alpha^{i}(q)\subseteq\rep'^{i}$,
which implies 
$\beta_{b}(\{F_{j}\}_{j=1}^{i})\subseteq\rep'^{i}$.
Since $\rep'^{i}$ and $\rep'^{i+1}$ include the same set of $F_{j}$-facts for every $j\in\{1,\dots,i\}$, we have $\beta_{b}(\{F_{j}\}_{j=1}^{i})\subseteq\rep'^{i+1}$.
Since $\beta_{b}(F_{i+1})\in\rep'^{i+1}$ by construction,
we obtain $\beta_{b}(\{F_{j}\}_{j=1}^{i+1})\subseteq\rep'^{i+1}$.
Since also $\alpha^{i+1}(\{F_{j}\}_{j=1}^{i+1})\subseteq\rep'^{i+1}$ 
(because $\alpha^{i+1}(q)\subseteq\rep'^{i+1}$),
it is correct to conclude that $\{\beta_{b},\alpha^{i+1}\}\models\FD{\{F_{j}\}_{j=1}^{i+1}}$.
We are now ready to show that $\alpha^{i+1}(F_{j})=\beta_{b}(F_{j})$ for all $j\in\{1,\dots,i+1\}$.
To this extent, pick any $k\in\{1,\dots,i+1\}$.
We have $\FD{\{F_{j}\}_{j=1}^{k-1}}\models\fd{Z}{\keyvars{F_{k}}}$ by Equation~(\ref{eq:fdimpl}).
Since $\{\beta_{b},\alpha^{i+1}\}\models\FD{\{F_{j}\}_{j=1}^{k-1}}$,
we have $\{\beta_{b},\alpha^{i+1}\}\models\fd{Z}{\keyvars{F_{k}}}$.
Then, from $\alpha^{i+1}(\vec{z})=\beta_{b}(\vec{z})$ (the second conjunct of item~\ref{it:pthree}),
it follows that $\alpha^{i+1}$ and $\beta_{b}$ agree on all variables of $\keyvars{F_{k}}$.
Since $\alpha^{i+1}(F_{k}),\beta_{b}(F_{k})\in\rep'^{i+1}$,
it must be the case that $\alpha^{i+1}(F_{k})=\beta_{b}(F_{k})$.
This concludes the induction step.

For the pair $(\rep'^{\ell},\alpha^{\ell})$,
we have that $\alpha^{\ell}(\{F_{j}\}_{j=1}^{\ell})=\beta_{b}(\{F_{j}\}_{j=1}^{\ell})$, and thus, since $w$ occurs in some $F_{j}$, $\alpha^{\ell}(w)=\beta_{b}(w)$.
Since also $\alpha^{\ell}(w)=\alpha(w)$,
we obtain $\alpha(w)=\beta_{b}(w)$, a contradiction.
We conclude by contradiction that some repair of $\db'$ falsifies $q$.

We repeat the ``purification" step until it can no longer be applied.
Let the final database be $\widehat{\db}$.
By the above reasoning,
we have that  every repair of $\widehat{\db}$ satisfies $q$ if and only if every repair of $\db$ satisfies $q$.
Let $\sep$ be the smallest set of $N$-facts containing $N(\underline{\beta(\vec{z})},\beta(w))$ for every valuation $\beta$ over $\queryvars{q}$ such that $\beta(q)\subseteq\db$.
We show that $\sep$ is consistent.
To this extent,
let $\beta_{1},\beta_{2}$ be valuations over $\queryvars{q}$ such that $\beta_{1}(q),\beta_{2}(q)\subseteq\db$ and $\beta_{1}(\vec{z})=\beta_{2}(\vec{z})$.
If $\beta_{1}(w)\neq\beta_{2}(w)$, then a purification step can remove the block containing $\beta_{1}(F)$,
contradicting our assumption that no purification step is applicable on $\widehat{\db}$.
We conclude by contradiction that $\beta_{1}(w)=\beta_{2}(w)$.

Since $N$ has mode~$\mc$ and $\sep$ is consistent,
we have that $\widehat{\db}\cup\sep$ is a legal database.
It can now be easily seen that every repair of $\db$ satisfies $q$ if and only if every repair of $\widehat{\db}\cup\sep$ satisfies $q'=q\cup\{N^{\mc}(\underline{\vec{z}},w)\}$.


It remains to be argued that the reduction is in $\FO$, i.e.,
that the result of the repeated ``purification" step can be obtained by a single first-order query.
Let $\queryvars{q}=\{x_{1},\dots,x_{n}\}$.
Let $q^{*}(x_{1},\dots,x_{n})$ denote the quantifier-free part of the Boolean query $q$.
For every $i\in\{1,\dots,n\}$, let $x_{i}'$ be a fresh variable.
Let $\vec{u}$ be a sequence of distinct variables such that $\sequencevars{\vec{u}}=\atomvars{F}$.
The following query finds all $F$-facts whose blocks can be removed:
$$\left \{\vec{u}\mid\exists^{*}\lrformula{q^{*}(x_{1},\dots,x_{n})\land q^{*}(x_{1}',\dots,x_{n}')\land\lrformula{\bigwedge_{z\in Z}z=z'}\land w\neq w'}\right\},$$
where the existential quantification ranges over all variables not in $\vec{u}$.
The $F$-facts that are to be preserved are not key-equal to a fact in the preceding query and can obviously be computed in $\FO$.
This concludes the proof of the first item.

\medskip
\myparagraph{Proof of the second item}
Assume that the attack graph of $q$ contains no strong cycle.  We will show that the attack graph of $q'$ contains no strong cycle either.
By the second item in Definition~\ref{def:saturated}, we can assume an atom $G\in q$ such that $Z\subseteq\atomvars{G}$.
It is sufficient to show that for every $F,H\in q$,
if there exists a witness for $F\attacks{q'}H$,
then there exists a witness for $F\attacks{q'}H$ that does not contain $N^{\mc}(\underline{\vec{z}},w)$.
To this extent,
assume that a witness for $F\attacks{q'}H$ contains 
\begin{equation}\label{eq:uprime}
\dotsm F'\step{u'}N^{\mc}(\underline{\vec{z}},w)\step{u''}F''\dotsm,
\end{equation}
where $u'$ and $u''$ are distinct variables.
We can assume without loss of generality that this is the only occurrence of $N^{\mc}(\underline{\vec{z}},w)$ in the witness.
In this case, we have $F\attacks{q}u'$.
If $u',u''\in Z$, then we can replace $N^{\mc}(\underline{\vec{z}},w)$ with $G$.
So the only nontrivial case is where either $u'=w$ or $u''=w$ (but not both).
Then, it must be the case that $\FD{q'\setminus\{F\}}\not\models\fd{\keyvars{F}}{w}$, thus also
\begin{equation}\label{eq:nm}
\FD{q\setminus\{F\}}\not\models\fd{\keyvars{F}}{w}.
\end{equation}

Since $\fd{Z}{w}$ is internal to $q$, there exists a sequential proof for $\FD{q}\models\fd{Z}{w}$ such that no atom in the proof attacks a variable in $Z \cup \{w\}$. Let $J_{1},J_{2},\dots,J_{\ell}$ be a shortest such proof.
Because $F\attacks{q}u'$ and $u' \in Z \cup \{w\}$, it must be that $F\not\in\{J_{1},\dots,J_{\ell}\}$.
We can assume that $w$ occurs at a non-primary-key position in $J_{\ell}$.
Because of~(\ref{eq:nm}), we can assume the existence of a variable $v\in\keyvars{J_{\ell}}$ such that 
$\FD{q\setminus\{F\}}\not\models\fd{\keyvars{F}}{v}$.
If $v\not\in Z$, then there exists $k<\ell$ such that $v$ occurs at a non-primary-key position in $J_{k}$.
Again, we can assume a variable $v'\in\keyvars{J_{k}}$  such that 
$\FD{q\setminus\{F\}}\not\models\fd{\keyvars{F}}{v'}$.
By repeating the same reasoning,
there exists a sequence
$$
\step{z_{i_{0}}}J_{i_{0}}
\step{z_{i_{1}}}J_{i_{1}}
\step{z_{i_{2}}}
\dots
\step{z_{i_{m}}}J_{i_{m}}
\step{w}
$$
where $1\leq i_{0}<i_{1}<\dotsm<i_{m}=\ell$
such that 
\begin{itemize}
\item
$z_{i_{0}}\in Z$;
\item
for all $j\in\{0,\dots,m\}$,
$\FD{q\setminus\{F\}}\not\models\fd{\keyvars{F}}{z_{i_{j}}}$; and
\item
for all $j\in\{1,\dots,m\}$,
$z_{i_{j}}\in\atomvars{J_{i_{j-1}}}\cap\atomvars{J_{i_{j}}}$.
In particular, $z_{i_{j}}\in\keyvars{J_{i_{j}}}$.
\end{itemize}
We can assume $G\in q$ such that $Z\subseteq\atomvars{G}$.
Let $u\in\{u',u''\}$ such that $u\neq w$.
Thus, $\{u,w\}=\{u',u''\}$.
It can now be easily seen that a witness for $F\attacks{q'}H$ can be obtained by
replacing $N^{\mc}(\underline{\vec{z}},w)$ in~(\ref{eq:uprime}) with the following sequence or its reverse:
$$
\step{u}G
\step{z_{i_{0}}}J_{i_{0}}
\step{z_{i_{1}}}J_{i_{1}}
\step{z_{i_{2}}}
\dots
\step{z_{i_{m}}}J_{i_{m}}
\step{w}
$$

This concludes the proof of Lemma~\ref{lem:saturate}.
\end{proof}

The proof of Lemma~\ref{lem:corsaturate} is now straightforward.

\begin{proof}[Proof of Lemma~\ref{lem:corsaturate}]
Repeated application of Lemma~\ref{lem:saturate}.
\end{proof}

\section{Proofs of Section~\ref{sec:mgraph}}

We will use the following helping lemma.

\begin{lemma}\label{lem:extend}
Let $q$ be a query in $\sjfbcq$ such that $q$ is saturated and the attack graph of $q$ contains no strong cycle.
Let $\isc$ be an initial strong component in the attack graph of $q$ with $\card{\isc}\geq 2$. 
For every atom $F \in \isc$, there exists an atom $H \in \isc$ such that $F \markov H$.
\end{lemma}
\begin{proof}
Assume $F \in \isc$.
Since $F$ belongs in an initial strong component with at least two atoms, there exists $G \in \isc$ such that
$F\attacks{q}G$ and the attack is weak. Thus, $\FD{q}\models\fd{\keyvars{F}}{\keyvars{G}}$.
It follows that $\FD{q\setminus\{F\}}\models\fd{\atomvars{F}}{\keyvars{G}}$.
Let $\sigma = H_1, H_2, \dots, H_\ell$ be a sequential proof for $\FD{q\setminus\{F\}}\models\fd{\atomvars{F}}{\keyvars{G}}$, and thus $F \notin \{H_1, \dots, H_\ell\}$. We can assume without loss of generality that $H_\ell = G$.

Let $j$ be the smallest index in $\{1, \dots, \ell\}$ such that $H_{j} \in \isc$. 
Since $H_\ell \in \isc$, such an index always exists.
Then, $\sigma = H_1, H_2, \dots, H_{j-1}$ is a sequential proof for $\FD{q\setminus\{F\}}\models\fd{\atomvars{F}}{\keyvars{H_j}}$ (observe that this proof may be empty). By our choice of $j$, for every $i\in\{1,\dots,j-1\}$,
we have $H_i \notin \isc$, and hence $H_i$ cannot attack $F$ or $H_j$ (since $\isc$ is an initial strong component).
It follows that no atom in $\sigma$ attacks a variable in $\atomvars{F}\cup\keyvars{H_{j}}$.
Since $q$ is saturated, this implies that $\FD{\catoms{q}}\models\fd{\atomvars{F}}{\keyvars{H_j}}$, and so
$F \markov H_j$.
\end{proof}

The proof of Lemma~\ref{lem:happy} can now be given.

\begin{proof}[Proof of Lemma~\ref{lem:happy}]
Starting from some atom $F_0 \in \isc$, by applying repeatedly Lemma~\ref{lem:extend}, we can create an infinite sequence 
$F_0 \markov F_1 \markov F_2 \markov \dotsm$ such that for every $i\geq 1$, $F_i \in \isc$ and $F_i \neq F_{i+1}$. Since the atoms in $\isc$ are finitely many, there will exist some $i,j$ such that $i<j$ and $F_{i}=F_{j+1}$.
It follows that the \mgraph of $q$ contains a cycle all of whose atoms belong to $\isc$.
\end{proof}

\begin{proof}[Proof of Lemma~\ref{lem:hockey}]
The first item is trivial.
For the second item, assume $A\mhook B$, $A\mhook B'$, and $\qatom{B}{q}=\qatom{B'}{q}$.
We can assume $F,G\in q$ such that $F\markov G$, $\qatom{A}{q}=F$, and $\qatom{B}{q}=G$.
Then, there exist valuations $\theta_{1},\theta_{2}$ over $\queryvars{q}$ such that
$A\in\theta_{1}(q)\subseteq\db$, $A\in\theta_{2}(q)\subseteq\db$,
$B\keyequal\theta_{1}(G)$, and $B'\keyequal\theta_{2}(G)$.
Since $\theta_{1}[\atomvars{F}]=\theta_{2}[\atomvars{F}]$
and $\FD{\catoms{q}}\models\fd{\atomvars{F}}{\keyvars{G}}$ (because $F\markov G$),
it follows $\theta_{1}[\keyvars{G}]=\theta_{2}[\keyvars{G}]$,
hence $B$ and $B'$ must be key-equal.
\end{proof}

\section{Proofs of Section~\ref{sec:garbageset}}

\subsection{Proofs of Lemmas~\ref{lem:closedunderunion} and~\ref{lem:nomenestomen}}

\begin{proof}[Proof of Lemma~\ref{lem:closedunderunion}]
Let $\bfo_{1}$ and $\bfo_{2}$ be garbage sets for $q_{0}$ in $\db$.
For every $i\in\{1,2\}$, we can assume a repair $\rep_{i}$ of $\bfo_{i}$ such that
\begin{quote}
\emph{Garbage Condition:}
for every valuation $\theta$ over $\queryvars{q}$ such that $\theta(q)\subseteq\formula{\db\setminus\bfo_{i}}\cup\rep_{i}$, we have $\theta(q_{0})\cap\rep_{i}=\emptyset$.
\end{quote}
Let $\bfo_{2}^{-}  =  \bfo_{2}\setminus\bfo_{1}$ and $\rep_{2}^{-}  =  \rep_{2}\setminus\bfo_{1}$.
Then, $\rep_{1}\uplus\rep_{2}^{-}$ is a repair of $\bfo_{1}\uplus\bfo_{2}^{-}$,
where the use of $\uplus$ (instead of $\cup$) indicates that the operands of the union are disjoint.
Let $\theta$ be an arbitrary valuation over $\queryvars{q}$ such that
$$\theta(q)\subseteq\lrformula{\db\setminus\lrformula{\bfo_{1}\uplus\bfo_{2}^{-}}}\cup\formula{\rep_{1}\uplus\rep_{2}^{-}}.$$
Then, $\theta(q)\subseteq\formula{\db\setminus\bfo_{1}}\cup\rep_{1}$.
Consequently, by the \emph{Garbage Condition} for $i=1$,
$\theta(q_{0})\cap\rep_{1}=\emptyset$, and thus $\theta(q_{0})\cap\bfo_{1}=\emptyset$.
It follows $\theta(q)\subseteq\lrformula{\db\setminus\lrformula{\bfo_{1}\cup\bfo_{2}}}\cup\rep_{2}^{-}$,
hence $\theta(q)\subseteq\lrformula{\db\setminus\bfo_{2}}\cup\rep_{2}^{-}$.
Consequently, by the \emph{Garbage Condition} for $i=2$,
$\theta(q_{0})\cap\rep_{2}^{-}=\emptyset$.

It follows that $\bfo_{1}\uplus\bfo_{2}^{-}$=$\bfo_{1}\cup\bfo_{2}$ is a garbage set for $q_{0}$ in $\db$. 
\end{proof}

\medskip

\begin{proof}[Proof of Lemma~\ref{lem:nomenestomen}]
The $\impliedby$-direction is trivial.
For the $\implies$-direction, 
assume that every repair of $\db$ satisfies~$q$.
We can assume a repair $\rep_{0}$ of $\bfo$ such that for every valuation $\theta$ over $\queryvars{q}$, if $\theta(q)\subseteq\formula{\db\setminus\bfo}\cup\rep_{0}$, then $\theta(q_{0})\cap\rep_{0}=\emptyset$.
Let $\rep$ be an arbitrary repair of $\db\setminus\bfo$.
It suffices to show $\rep\models q$.
Since $\rep\cup\rep_{0}$ is a repair of $\db$,
we can assume a valuation $\theta$ over $\queryvars{q}$ such that $\theta(q)\subseteq\rep\cup\rep_{0}$.
Since $\theta(q)\subseteq\formula{\db\setminus\bfo}\cup\rep_{0}$ is obvious,
it follows $\theta(q)\cap\rep_{0}=\emptyset$.
Consequently, $\theta(q)\subseteq\rep$, hence $\rep\models q$.
This concludes the proof.
\end{proof}

\subsection{Proof of Lemma~\ref{lem:together}}

We will use two helping lemmas.

\begin{lemma}\label{lem:substraction}
Let $q$ be a query in $\sjfbcq$, and let $q_{0}\subseteq q$.
Let $\bfo$ be a garbage set for $q_{0}$ in $\db$.
If $\bfp$ is the union of one or more blocks of $\bfo$,
then $\bfo\setminus\bfp$ is a garbage set for $q_{0}$ in $\db\setminus\bfp$.
\end{lemma}
\begin{proof}
Let $\bfp$ be the union of one or more blocks of $\bfo$.
We can assume a repair $\rep$ of $\bfo$ such that for every valuation $\theta$ over $\queryvars{q}$, if $\theta(q)\subseteq\formula{\db\setminus\bfo}\cup\rep$, then $\theta(q)\cap\rep=\emptyset$.
Let $\sep=\rep\setminus\bfp$.
Obviously, $\sep$ is a repair of $\bfo\setminus\bfp$.

Let $\theta$ be a valuation over $\queryvars{q}$ such that $\theta(q)\subseteq\lrformula{\lrformula{\db\setminus\bfp}\setminus\lrformula{\bfo\setminus\bfp}}\cup\sep$.
It suffices to show $\theta(q)\cap\sep=\emptyset$.
Since $\lrformula{\db\setminus\bfp}\setminus\lrformula{\bfo\setminus\bfp}\subseteq\db\setminus\bfo$ and $\sep\subseteq\rep$, it follows $\theta(q)\subseteq\formula{\db\setminus\bfo}\cup\rep$, hence $\theta(q)\cap\rep=\emptyset$.
It follows $\theta(q)\cap\sep=\emptyset$.
\end{proof}


\begin{corollary}\label{cor:substraction}
Let $q$ be a query in $\sjfbcq$, and let $q_{0}\subseteq q$.
Let $\bfo$ be a garbage set for $q_{0}$ in $\db$.
If every garbage set for $q_{0}$ in $\db\setminus\bfo$ is empty,
then $\bfo$ is the maximal garbage set for $q_{0}$ in $\db$.
\end{corollary}
\begin{proof}
Proof by contraposition.
Assume that $\bfo$ is not the maximal garbage set for $q_{0}$ in $\db$.
Let $\bfo_{0}$ be the maximal garbage set for $q_{0}$ in $\db$.
By Lemma~\ref{lem:substraction}, $\bfo_{0}\setminus\bfo$ is a nonempty garbage set for $q_{0}$ in $\db\setminus\bfo$. 
\end{proof}

\begin{lemma}\label{lem:chainunion}
Let $q$ be a query in $\sjfbcq$, and let $q_{0}\subseteq q$.
Let $\db$ be a database.
If $\bfo$ is a garbage set for $q_{0}$ in $\db$,
and $\bfp$ is a garbage set for $q_{0}$ in $\db\setminus\bfo$,
then $\bfo\cup\bfp$ is a garbage set for $q_{0}$ in $\db$.
\end{lemma}
\begin{proof}
Assume the hypothesis holds.
Note that $\bfo\cap\bfp=\emptyset$.
We can assume a repair $\rep$ of $\bfo$ such that for every valuation $\theta$ over $\queryvars{q}$, if $\theta(q)\subseteq\formula{\db\setminus\bfo}\cup\rep$, then $\theta(q)\cap\rep=\emptyset$.
Likewise, we can assume a repair $\sep$ of $\bfp$ such that for every valuation $\theta$ over $\queryvars{q}$, if $\theta(q)\subseteq\lrformula{\lrformula{\db\setminus\bfo}\setminus\bfp}\cup\sep$, then $\theta(q)\cap\sep=\emptyset$.
Obviously, $\rep\cup\sep$ is a repair of $\bfo\cup\bfp$.

Let $\theta$ be a valuation over $\queryvars{q}$ such that $\theta(q)\subseteq\lrformula{\db\setminus\lrformula{\bfo\cup\bfp}}\cup\formula{\rep\cup\sep}$.
From the set inclusion
$\lrformula{\db\setminus\lrformula{\bfo\cup\bfp}}\cup\formula{\rep\cup\sep}
\subseteq
\formula{\db\setminus\bfo}\cup\rep$,
it follows
$\theta(q)\subseteq\formula{\db\setminus\bfo}\cup\rep$,
hence $\theta(q)\cap\rep=\emptyset$.
Then,
$\theta(q)\subseteq\lrformula{\db\setminus\lrformula{\bfo\cup\bfp}}\cup\sep
=
\lrformula{\lrformula{\db\setminus\bfo}\setminus\bfp}\cup\sep$,
hence $\theta(q)\cap\sep=\emptyset$.
It follows $\theta(q)\cap\formula{\rep\cup\sep}=\emptyset$.
\end{proof}

\begin{corollary}\label{cor:chainunion}
Let $q$ be a query in $\sjfbcq$, and let $q_{0}\subseteq q$.
Let $\db$ be a database, and let $\bfo$ be the maximal garbage set for $q_{0}$ in $\db$.
Then, every garbage set for $q_{0}$ in $\db\setminus\bfo$ is empty.  
\end{corollary}
\begin{proof}
Immediate from Lemma~\ref{lem:chainunion}.
\end{proof}

The proof of Lemma~\ref{lem:together} can now be given.

\begin{proof}[Proof of Lemma~\ref{lem:together}]
Immediate from Corollaries~\ref{cor:substraction} and~\ref{cor:chainunion}.
\end{proof}

\section{Proofs of Section~\ref{sec:scrub}}
\label{app:scrub}

\begin{figure}\centering
\includegraphics[scale=0.8]{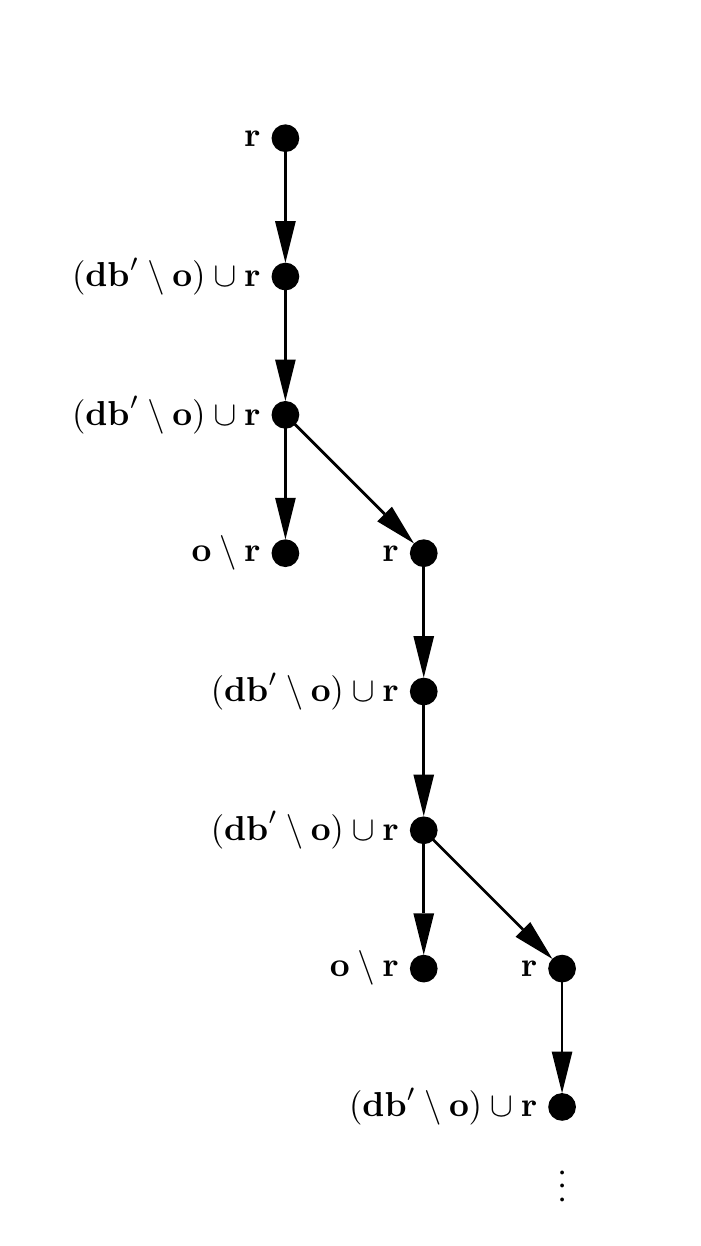}
\caption{Illustration of the $\cmhook{C}$-graph in the proof of Lemma~\ref{lem:scrubalgo}.
Every vertex is a fact, and the vertex labels indicate the set to which each vertex belongs. Vertices on the same horizontal line are key-equal.
}\label{fig:maximal}
\end{figure}

\medskip

\begin{proof}[Proof of Lemma~\ref{lem:scrubalgo}]
We will write $\oplus$ for addition modulo~$k$.
We first consider garbage sets respecting the first three conditions.
\begin{itemize}
\item
Let $A$ be a fact of $\db$ such that $\qatom{A}{q}\in\{F_{0},\dots,F_{k-1}\}$ and $A$ has zero outdegree in the $\cmhook{C}$-graph.
Then, there exists no valuation $\theta$ over $\queryvars{q}$ such that $A\in\theta(q)\subseteq\db$.
It is obvious that $\theblock{A}{\db}$ is a garbage set for $C$ in $\db$.
\item
Let $A_{0}\cmhook{C}A_{1}\cmhook{C}\dotsm\cmhook{C}A_{k-1}\cmhook{C}A_{0}$ be an irrelevant $1$-embedding of $C$ in $\db$.
Assume without loss of generality that for every $i\in\{0,\dots,k-1\}$,
$\qatom{A_{i}}{q}=F_{i}$.
Let $\bfo=\bigcup_{i=0}^{k-1}\theblock{A_{i}}{\db}$.
Let $\rep=\{A_{0},\dots,A_{k-1}\}$, which is obviously a repair of $\bfo$.
We show that $\bfo$ is a garbage set for $C$ in $\db$.
Assume, toward a contradiction, the existence of a valuation $\theta$ over $\queryvars{q}$ such that for some $i\in\{0,\dots,k-1\}$, $A_{i}\in\theta(q)\subseteq\formula{\db\setminus\bfo}\cup\rep$.
Then, $\theta(F_{i})\cmhook{C}\theta(F_{i\oplus 1})$.
Since $\theta(F_{i})=A_{i}$, we have  $A_{i}\cmhook{C}\theta(F_{i\oplus 1})$. 
From $A_{i}\cmhook{C}\theta(F_{i\oplus 1})$ and $A_{i}\cmhook{C}A_{i\oplus 1}$, it follows $\theta(F_{i\oplus 1})\keyequal A_{i\oplus 1}$ by Lemma~\ref{lem:hockey}.
Since $\theta(F_{i\oplus 1})\in\formula{\db\setminus\bfo}\cup\rep$,
it follows  $\theta(F_{i\oplus 1})=A_{i\oplus 1}$.
By repeated application of the same reasoning,
for every $j\in\{0,\dots,k-1\}$,  $\theta(F_{j})=A_{j}$.
But then $A_{0}\cmhook{C}A_{1}\cmhook{C}\dotsm\cmhook{C}A_{k-1}\cmhook{C}A_{0}$ is a relevant $1$-embedding of $C$ in $\db$, a contradiction.
\item
Let $\rep$ be a set containing all (and only) the facts of some $n$-embedding of $C$ in $\db$ with $n\geq 2$.
Let $\bfo=\bigcup_{A\in\rep}\theblock{A}{\db}$.
It can be shown that $\bfo$ is a garbage set for $C$ in $\db$;
the argumentation is analogous to the reasoning in the previous paragraph.
\end{itemize}
Let $\bfo_{0}$ be the minimal subset of $\db$ that satisfies all conditions in the statement of the lemma except the recursive condition~\ref{it:recursive}.
By Lemma~\ref{lem:closedunderunion} and our reasoning in the previous items, it follows that $\bfo_{0}$ is a garbage set for $C$ in $\db$. 

Note that the first three conditions do not recursively depend on $\bfo_{0}$.
Starting with $\bfo_{0}$, construct a maximal sequence
$$\bfo_{0},\mu_{0},\bfo_{1},\mu_{1},\bfo_{2},\mu_{2},\dots,\bfo_{m},\mu_{m},\bfo_{m+1}$$
such that $\bfo_{0}\subsetneq \bfo_{1}\subsetneq \bfo_{2}\subsetneq\dotsm\subsetneq \bfo_{m+1}$
and for every $h\in\{1,\dots,m\}$, 
\begin{enumerate}
\item
$\mu_{h}$ is a valuation over $\queryvars{q}$ such that $\mu_{h}(q)\subseteq\db$ and $\mu_{h}(q)\cap \bfo_{h}\neq\emptyset$.
Thus, $\mu(F_{0})\cmhook{C}\mu(F_{1})\cmhook{C}\dotsm\cmhook{C}\mu(F_{k-1})$ is a relevant $1$-embedding of $C$ in $\db$; and
\item
$\bfo_{h+1}=\bfo_{h}\cup\lrformula{\bigcup_{i=0}^{k-1}\theblock{\mu_{h}(F_{i})}{\db}}$.
\end{enumerate}
It is clear that the final set $\bfo_{m+1}$ is a minimal set satisfying all conditions in the statement of the lemma.
We show by induction on increasing $h$ that for all $h\in\{0,1,\dots,m\}$, $\bfo_{h}$ is a garbage set for $C$ in $\db$.
We have already showed that $\bfo_{0}$ is a garbage set for $C$ in $\db$.
For the induction step, $h\rightarrow h+1$, the induction hypothesis is that $\bfo_{h}$ is a garbage set for $C$ in $\db$.
Then, there exists a repair $\rep$ of $\bfo_{h}$ such that for every valuation $\theta$ over $\queryvars{q}$, if $\theta(q)\subseteq\formula{\db\setminus\bfo_{h}}\cup\rep$, then $\theta(q)\cap\rep=\emptyset$.
For every $i\in\{0,\dots,k-1\}$, define $A_{i}\defeq\mu_{h}(F_{i})$.
Let $\sep=\{A_{0},\dots,A_{k-1}\}\setminus\bfo_{h}$.
We have $\bfo_{h+1}=\bfo_{h}\uplus\lrformula{\bigcup_{A_{j}\in\sep}\theblock{A_{j}}{\db}}$.
Let $\rep'=\rep\uplus\sep$.
Obviously, $\rep'$ is a repair of $\bfo_{h+1}$.
Here, we use $\uplus$, instead of $\cup$, to make clear that the operands of the union are disjoint.
Assume, toward a contradiction, the existence of a valuation $\theta$ over $\queryvars{q}$ such that $\theta(q)\subseteq\formula{\db\setminus\bfo_{h+1}}\cup\rep'$ and $\theta(q)\cap\rep'\neq\emptyset$.
Since $\formula{\db\setminus\bfo_{h+1}}\cup\rep'\subseteq\formula{\db\setminus\bfo_{h}}\cup\rep$, it follows $\theta(q)\subseteq\formula{\db\setminus\bfo_{h}}\cup\rep$, hence $\theta(q)\cap\rep=\emptyset$ by our initial hypothesis.
It must be the case that $\theta(q)\cap\sep\neq\emptyset$.
We can assume $i\in\{0,\dots,k-1\}$ such that $A_{i}\in\theta(q)\cap\sep$.
We have $\theta(F_{i})\cmhook{C}\theta(F_{i\oplus 1})$.
Since $\theta(F_{i})=A_{i}$, we have  $A_{i}\cmhook{C}\theta(F_{i\oplus 1})$. 
From $A_{i}\cmhook{C}\theta(F_{i\oplus 1})$ and $A_{i}\cmhook{C}A_{i\oplus 1}$, it follows $\theta(F_{i\oplus 1})\keyequal A_{i\oplus 1}$ by Lemma~\ref{lem:hockey}.
Since $\theta(F_{i\oplus 1})\in\formula{\db\setminus\bfo_{h+1}}\cup\rep'$,
it follows that either $\theta(F_{i\oplus 1})=A_{i\oplus 1}\in\sep$ (this happens if $A_{i\oplus 1}\not\in\bfo_{h}$) or $\theta(F_{i\oplus 1})\in\rep$.
Thus, either $A_{i\oplus 1}\in\theta(q)\cap\sep$ or $\theta(F_{i\oplus 1})\in\rep$.
If $A_{i\oplus 1}\in\theta(q)\cap\sep$, then, by the same reasoning, either $A_{i\oplus 2}\in\theta(q)\cap\sep$ or $\theta(F_{i\oplus 2})\in\rep$.
By repeating the same reasoning, we obtain that for all $j\in\{0,\dots,k-1\}$, either $A_{j}\in\theta(q)\cap\sep$ or $\theta(F_{j})\in\rep$.
Since $\mu_{h}(q)\cap\bfo_{h}\neq\emptyset$ by our construction, we can assume the existence of $\ell\in\{0,\dots,k-1\}$ such that $A_{\ell}\in\bfo_{h}$, hence $A_{\ell}\not\in\sep$.
Since $A_{\ell}\not\in\theta(q)\cap\sep$,
it follows $\theta(F_{\ell})\in\rep$, contradicting that $\theta(q)\cap\rep=\emptyset$.
This concludes the induction step.
It is correct to conclude that $\bfo_{m+1}$ is a garbage set for $C$ in $\db$.

Let $\db'=\db\setminus\bfo_{m+1}$.
We show that the garbage set for $C$ in $\db'$ is empty.
Assume, toward a contradiction, that $\bfo$ is a nonempty garbage set for $C$ in $\db'$.
We can assume a repair $\rep$ of $\bfo$ such that for every valuation $\theta$ over $\queryvars{q}$, if $\theta(q)\subseteq\formula{\db'\setminus\bfo}\cup\rep$, then $\theta(q)\cap\rep=\emptyset$.

We show that for any $A\in\rep$,
the $\cmhook{C}$-graph contains an infinite path that starts from $A$ such that any
vertex on the path belongs to $\formula{\db'\setminus\bfo}\cup\rep$ and any (contiguous) subpath of length~$k$ contains some fact from $\rep$. 
To this extent, let $A$ be a fact of $\rep$.
By our construction, there exists a valuation $\mu$ over $\queryvars{q}$ such that $A\in\mu(q)\subseteq\db'$ (otherwise $A$ would belong to $\bfo_{m+1}$).
Hence, $\mu(F_{0})\cmhook{C}\mu(F_{1})\cmhook{C}\dotsm\cmhook{C}\mu(F_{k-1})\cmhook{C}\mu(F_{0})$ is a relevant $1$-embedding of $C$ in $\db'$ that contains $A$.
Then, for some $i\in\{0,\dots,k-1\}$, it must be the case that $\mu(F_{i})\not\in\formula{\db'\setminus\bfo}\cup\rep$ (or else $\mu(q)\subseteq\formula{\db'\setminus\bfo}\cup\rep$ and $\mu(q)\cap\rep\neq\emptyset$, a contradiction).
Thus, the $\cmhook{C}$-graph contains a shortest path $\pi$ of length $<k$ from $A$ to some fact $B\in\bfo\setminus\rep$.
Then, there exists $B'\in\rep$ such that $B'\keyequal B$ and the $\cmhook{C}$-graph contains a path of length $<k$ from $A$ to $B'$.
This path is obtained by substituting $B'$ for $B$ in $\pi$.
Since  $B'\in\rep$, we can continue the path by applying the same reasoning as for $A$.
The path is illustrated by Fig.~\ref{fig:maximal}.
Since the directed path is infinite, it has a shortest finite subpath of length $\geq k$ whose first vertex is key-equal to its last vertex. 
Let $D$ be the last but one vertex on this subpath.
Since the $\cmhook{C}$-graph contains a directed edge from $D$ to the first vertex of the subpath, it contains a cycle of some length $nk$ with $n\geq 1$.
Since this cycle is obviously an $n$-embedding of $C$ in $\db'=\db\setminus\bfo_{m+1}$, it must be a relevant $1$-embedding of $C$ in $\db'$ which, moreover, contains some fact of~$\rep$.
Thus, there exists a valuation $\mu$ over $\queryvars{q}$ such that $\mu(q)\subseteq\formula{\db'\setminus\bfo}\cup\rep$ and $\mu(q)\cap\rep\neq\emptyset$, a contradiction.

Since the garbage set for $\db\setminus\bfo_{m+1}$ is empty,
it follows by Lemma~\ref{lem:together} that $\bfo_{m+1}$ is the maximal garbage set for $C$ in $\db$.  
This concludes the proof.
\end{proof}

\medskip

\begin{proof}[Proof of Corollary~\ref{cor:scrubcomponent}]
We first show the following property: if $A\cmhook{C}B$ and $B$ belongs to the maximal garbage set for $C$ in $\db$, then $A$ also belongs to the maximal garbage set for $C$ in $\db$. 
To this extent, assume $A\cmhook{C}B$ such that $B$ belongs to the maximal garbage set for $C$ in $\db$.
We can assume an edge $F_{0}\markov F_{1}$ in $C$ and a valuation $\theta$ over $\queryvars{q}$ such that $\theta(q)\subseteq\db$, $A=\theta(F_{0})$, and $B\keyequal\theta(F_{1})$.
Since $B$ belongs to the maximal garbage set for $C$ in $\db$,
we have that $\theta(F_{1})$ belongs to the maximal garbage set by Definition~\ref{def:garbage}.
From $\theta(C)\subseteq\db$, it follows that $A\cmhook{C}\theta(F_{1})$ is an edge of a relevant $1$-embedding of $C$ in $\db$.
It follows from the recursive condition~\ref{it:recursive} in Lemma~\ref{lem:scrubalgo} that $A$ belongs to  the maximal garbage set for $C$ in $\db$.

The proof of the lemma can now be given.
Assume that $B\in\isc$ belongs to the maximal garbage set for $C$ in $\db$.
Let $A\in\isc$.
Since $\isc$ is a strong component,
there exists a path $A_{0}\cmhook{C}A_{1}\cmhook{C}\dotsm\cmhook{C}A_{\ell}$ such that $A_{0}=A$ and $A_{\ell}=B$.
By repeating the property of the previous paragraph,
we find that the maximal garbage set for $C$ in $\db$ contains $A_{\ell}$, $A_{\ell-1}$, \dots, $A_{0}$.
\end{proof}

\medskip

\begin{proof}[Proof of Lemma~\ref{lem:longcycle}]
Let $G=(V,E)$ be an instance of $\problem{LONGCYCLE}(k)$.
A directed cycle in $G$ of length~$k$ is called a \emph{$k$-cycle}.
Since the graph $G$ is $k$-partite, every $k$-cycle is elementary.

We denote by $\po{G}=(\po{V},\po{E})$ the undirected graph whose vertices are the $k$-cycles of $G$.
There is an undirected edge between any two distinct $k$-cycles $P_{1}$ and $P_{2}$ if $V(P_{1})\cap V(P_{2})\neq\emptyset$. 
We show that the following are equivalent:
\begin{enumerate}
\item\label{it:polyp}
$\po{G}$ has a chordless cycle of length $\geq 2k$
or $G$ has an elementary directed cycle of length $nk$ with $2\leq n\leq 2k-3$.
\item\label{it:pylop}
$G$ contains an elementary directed cycle of length~$\geq 2k$.
\end{enumerate}

\framebox{\ref{it:polyp}$\implies$\ref{it:pylop}}
Assume that \ref{it:polyp} holds true.
The result is obvious if there exists $n$ such that $2\leq n\leq 2k-3$ and $G$ has an elementary cycle of length $nk$.
Assume next that $\po{G}$ has a chordless elementary cycle $(P_0, P_1, \dots, P_{m-1}, P_0)$ of length $m \geq 2k$.
We construct a labeled cycle $C$ in $G$ using the following procedure.
The construction will define a labeling function $\ell$ from the vertices in $C$ to $\{0,1,\dots,m-1\}$.
It will be the case that $w\in V(P_{\ell(w)})$ for every vertex $w$ in $C$.
We start with any vertex $v_0\in V(P_{m-1})\cap V(P_0)$ and define its label as $\ell(v_{0})\defeq 0$. 
At any point of the procedure,
if we are at vertex $u$ with label $\ell(u)$, we choose the next vertex $w$ in $C$ to be the next vertex in the $k$-cycle $P_{\ell(u)}$.
If $\ell(u)<m-1$ and $w$ also belongs to $P_{\ell(u)+1}$, we let $\ell(w)\defeq\ell(u)+1$; otherwise $\ell(w)\defeq\ell(u)$.
The procedure terminates when we attempt to add a vertex that already exists in $C$,
and thus $C$ will be elementary.

We first show that the termination condition will not be met for any vertex distinct from $v_{0}$. 
Suppose, toward a contradiction, that the sequence constructed so far is $C = \tuple{v_0, v_1, \dots, v_n}$, $\ell(v_n) = i \leq m-1$, and the next vertex in $P_i$ is some $v_j$ with $j\in\{1, \dots, n-1\}$. 
Since $v_{j}$ belongs to both $P_{i}$ and $P_{\ell(v_{j})}$, 
it must be the case that $\ell(v_{j})\geq i-1$, because otherwise $\{P_{i},P_{\ell(v_{j})}\}$ is a chord in $(P_0, P_1, \dots, P_{m-1}, P_0)$, a contradiction.
%
We now distinguish two cases:
\begin{description}
\item[Case $\ell(v_j) = i-1$.] 
Then, $v_j\in V(P_{i-1})\cap V(P_i)$. By the procedure, this means that $\ell(v_{j-1})=i-2$.
Indeed, if $\ell(v_{j-1})=i-1$,
then the procedure would have set $\ell(v_{j})$ to $i$, because $v_{j}$ also belongs to $P_{i}$.
But then this also implies that $v_j \in V(P_{i-2})$, a contradiction to the fact that the cycle is chordless.
\item[Case $\ell(v_j) = i$.]
Informally, the procedure reaches a vertex on $P_{i}$ that has been visited before. 
Then, $C$ contains all the vertices of $P_i$, and none of them are in $P_{i+1}$, a contradiction.
\end{description}
It is now clear that at some point we will reach $v_0$. Indeed, when the label becomes $m-1$, the procedure will follow the edges of $P_{m-1}$ until it reaches $v_0$.
The cycle $C$ has length $\geq 2k$, because every label from $\{0, \dots, m-1\}$ occurs in some vertex of $C$, hence $C$ contains at least  $m\geq 2k$ vertices.

\framebox{\ref{it:pylop}$\implies$\ref{it:polyp}}
We first introduce some notions that will be useful in the proof.
A subpath of a path is a consecutive subsequence of edges of that path. 
Every path is a subpath of itself.
We write $\start{\pi}$ and $\fin{\pi}$ to denote, respectively, the first and the last vertex of a path $\pi$.
If $\fin{\pi}=\start{\pi'}$,
then $\pi\cdot\pi'$ denotes the concatenation of paths $\pi$ and $\pi'$.

\myparagraph{Covering}
Let $O$ be an elementary cycle in $G$ of size $\geq 2k$.
A \emph{seam in $O$} is a subpath of $O$ that is also a subpath of some $k$-cycle.
Obviously, every seam in $O$ has length $<k$.
A \emph{covering of $O$} is a set of seams in $O$ such that every edge of $O$ is an edge of some seam in the set.
Since every edge of $G$ belongs to some $k$-cycle by our hypothesis, $O$ has a covering.

\myparagraph{Cyclic ordering of the seams in a minimal covering}
Let $C=\{S_{0},S_{1},\dots,S_{\ell-1}\}$ be a minimal (with respect to cardinality) covering of $O$.
For every $i\in\{0,\dots,\ell-1\}$,
pick an edge $e_{i}\in E(O)$ such that $e_{i}\in E(S_{i})$ and $e_{i}\notin E(S_{j})$ for all $j\in\{0,\dots,\ell-1\}$ such that $j\neq i$.
Notice that if such $e_{i}$ would not exist for some $i$, then every edge of $O$ would belong to some seam in $C\setminus\{S_{i}\}$, contradicting $C$'s minimality. 
From here on, we will assume that if $C$ is a minimal covering, then its seams are listed such that a traversal of $O$ that starts with $e_{0}$ traverses these $e_{i}$'s in the order $e_{0}$, $e_{1}$, \dots, $e_{\ell-1}$.
It can be seen that this cyclic order does not depend on which $e_{i}$ is picked from $S_{i}$ when more than one choice would be possible.
Thus, if we traverse the edges of $O$ starting from $\start{S_{0}}$, then for every $i\in\{1,\dots,\ell-1\}$, we will traverse $e_{i}\in E(S_{i})$ after $S_{i-1}$ and before $S_{(i+1)\mod\ell}$.
For the following definition, it is useful to note that if $S$ and $S'$ are consecutive seams in a minimal covering, then $S$ and $S'$ can overlap on a common subpath.

\myparagraph{Preferred coverings}
Let $O_{1}$ and $O_{2}$ be two elementary cycles, both of length $\geq 2k$.
These elementary cycles need not have the same length and need not be distinct. 
Let $C_{1}$ be a covering of $O_{1}$,
and $C_{2}$ a covering of $O_{2}$.
The covering $C_{1}$ is said to be \emph{preferred over} $C_{2}$ if $\card{C_{1}}<\card{C_{2}}$.
If $\card{C_{1}}=\card{C_{2}}$,
then $C_{1}$ is said to be \emph{preferred over} $C_{2}$ if $\sum_{S\in C_{1}}\card{E(S)}>\sum_{S\in C_{2}}\card{E(S)}$.
Informally, we prefer coverings of smaller cardinality;
if cardinalities are equal, we prefer the covering with greater average seam length.

We are now ready to write down the proof for \ref{it:pylop}$\implies$\ref{it:polyp}.
Assume that $G$ contains an elementary directed cycle of length~$\geq 2k$.
We can assume the existence of an elementary cycle $O$ in $G$ of length $\geq 2k$ having a covering $C^{\ast}$ such that no elementary cycle of length $\geq 2k$ has a covering that is preferred over  $C^{\ast}$.

Let $C^{\ast}=\{S_{0},S_{1},\dots,S_{\ell-1}\}$.
Note that it must be the case that $\ell\geq 3$.
For every $i\in\{0,\dots,\ell-1\}$, we can assume a $k$-cycle $P_{i}$ such that $S_{i}$ is a subpath of $P_{i}$.
It will be the case that $(P_{0},P_{1},\dots,P_{\ell-1},P_{0})$ is a cycle in $\hat{G}$.
We will show that this cycle in $\hat{G}$ is chordless.
The proof is by contradiction.
Assume that the $\hat{G}$-cycle  $(P_{0},P_{1},\dots,P_{\ell-1},P_{0})$ has a chord.
Then we can assume two $k$-cycles $P,P'\in\{P_{0},P_{1},\dots,P_{\ell-1}\}$ such that $P$ and $P'$ are not adjacent in the $\hat{G}$-cycle and $V(P)\cap V(P')\neq\emptyset$ (and thus $\{P,P'\}$ is a chord in the $\hat{G}$-cycle). 
Let $S$ and $S'$ be the two seams that are subpaths of $P$ and $P'$, respectively. 
Informally, the covering $O^{\ast}$ uses disjoint subpaths, $S$ and $S'$, of two intersecting $k$-cycles. 
This situation is sketched in Fig.~\ref{fig:sketch}.

\begin{figure}[h]\centering
\begin{tabular}{cc}
\includegraphics[scale=0.50,trim={10cm 5cm 10cm 2cm},clip]{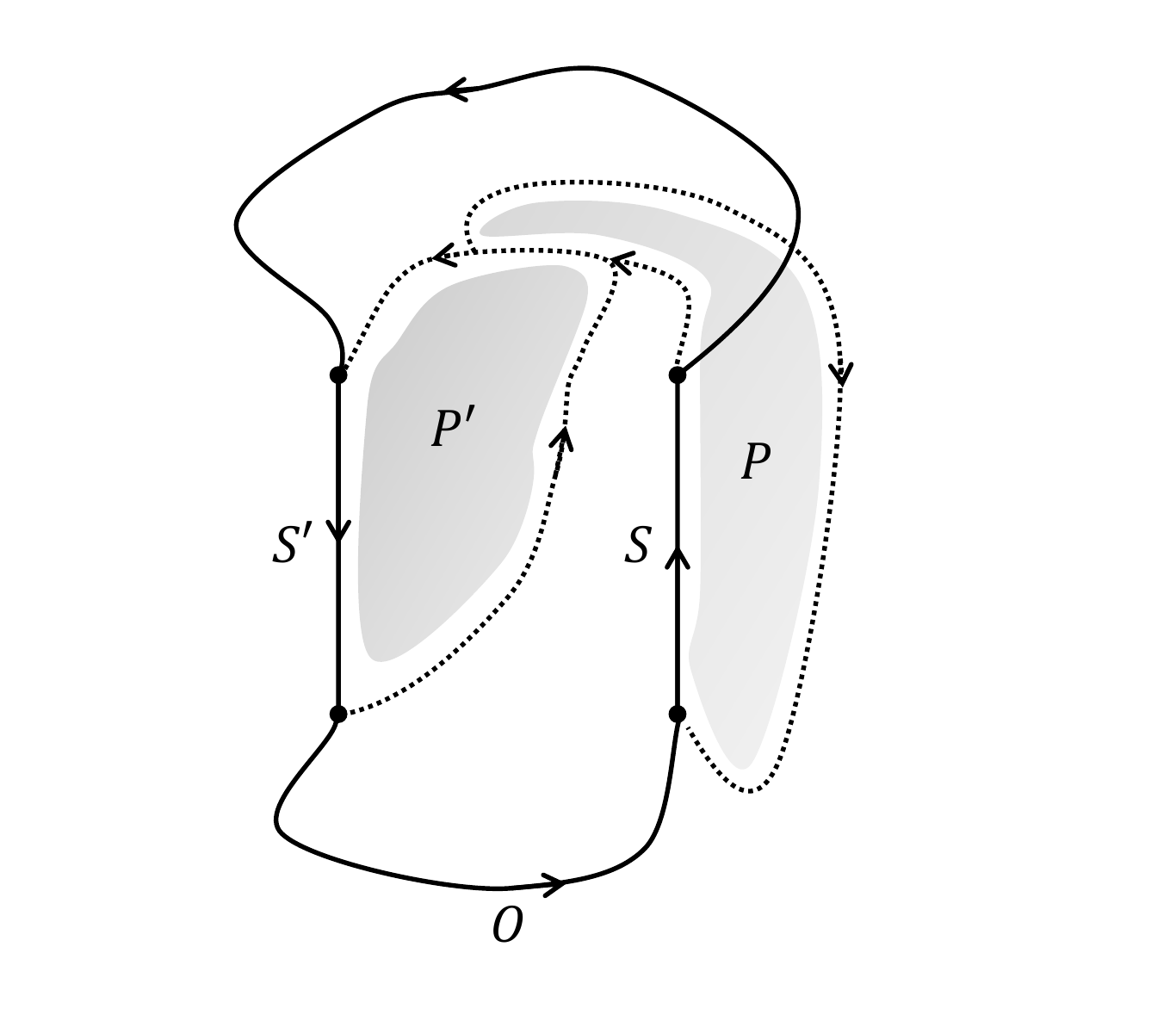}
&
\includegraphics[scale=0.50,trim={0cm 8cm 20cm 0cm},clip]{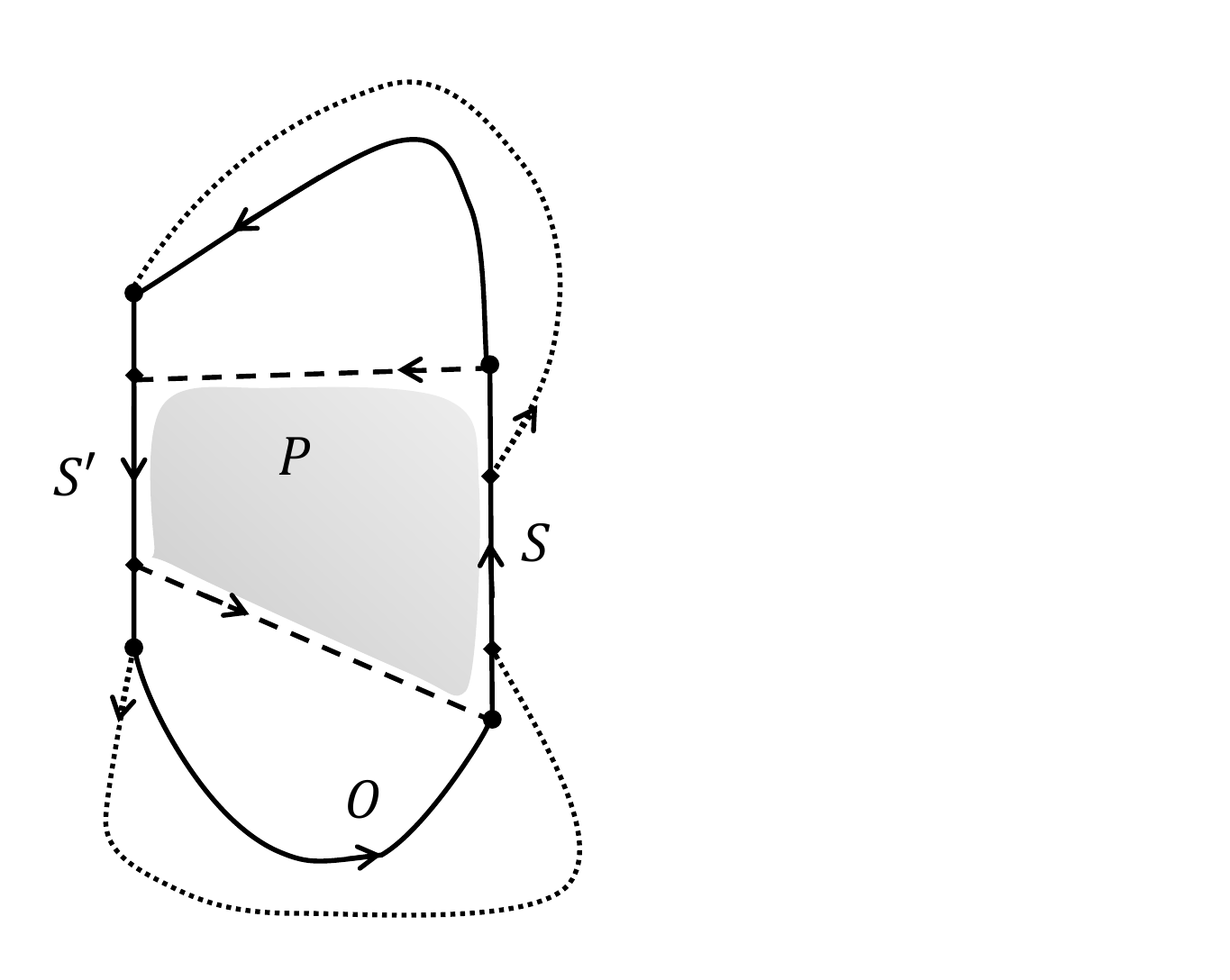}
\end{tabular}
\caption{
The seams $S$ and $S'$ are subpaths of the intersecting $k$-cycles $P$ and $P'$, respectively.
The solid line is the elementary cycle $O$ of length $\geq 2k$.
In the right diagram, the $k$-cycle $P'$ uses the dotted paths.
}\label{fig:sketch}
\end{figure} 

Assume without loss of generality that $S=S_{0}$.
We can assume $m\in\{2,\dots,\ell-2\}$ such that $S'=S_{m}$.
Thus,
\begin{equation}\label{eq:cstar}
C^{\ast}=\{S,S_{1},\dots,S_{m-1},S',S_{m+1},\dots,S_{\ell-1}\}.
\end{equation}
It will be the case that $S$ and $S'$ are disjoint paths (or else $C^{\ast}$ would not be a preferred covering, a contradiction).

Let $\pi$ be the subpath of $O$ from $\fin{S'}$ to $\start{S}$.
Let $T$ be the shortest subpath of $P$ that starts in $\fin{S}$ and ends in a vertex (call it $x$) that belongs to $P'$.
Such $x$ will always be reached because $P$ and $P'$ intersect. 
\begin{figure}[h]\centering
\includegraphics[scale=0.50,trim={0cm 8cm 10cm 2cm},clip]{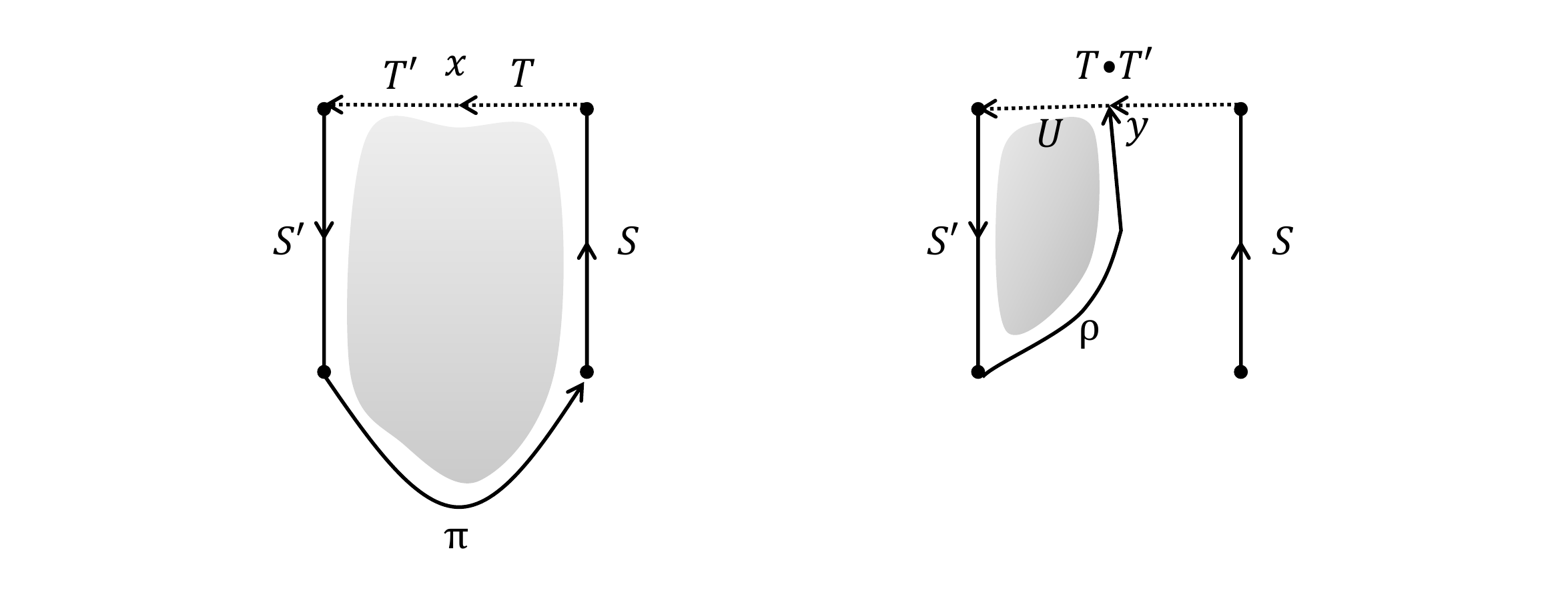}
\caption{Construction in the proof of Lemma~\ref{lem:longcycle}.
The $k$-cycles $P$ and $P'$ intersect in a vertex $x$ that does not belong to $S'$.
The path $S'\cdot\pi\cdot S$ is a subpath of $O$.
In the left diagram, $\pi$ does not intersect $T\cdot T'$.
In the right diagram, $\pi$ intersects $T\cdot T'$ in $y$.
}\label{fig:clutter}
\end{figure} 
\begin{figure}[h]\centering
\includegraphics[scale=0.50,trim={0cm 8cm 24cm 0cm},clip]{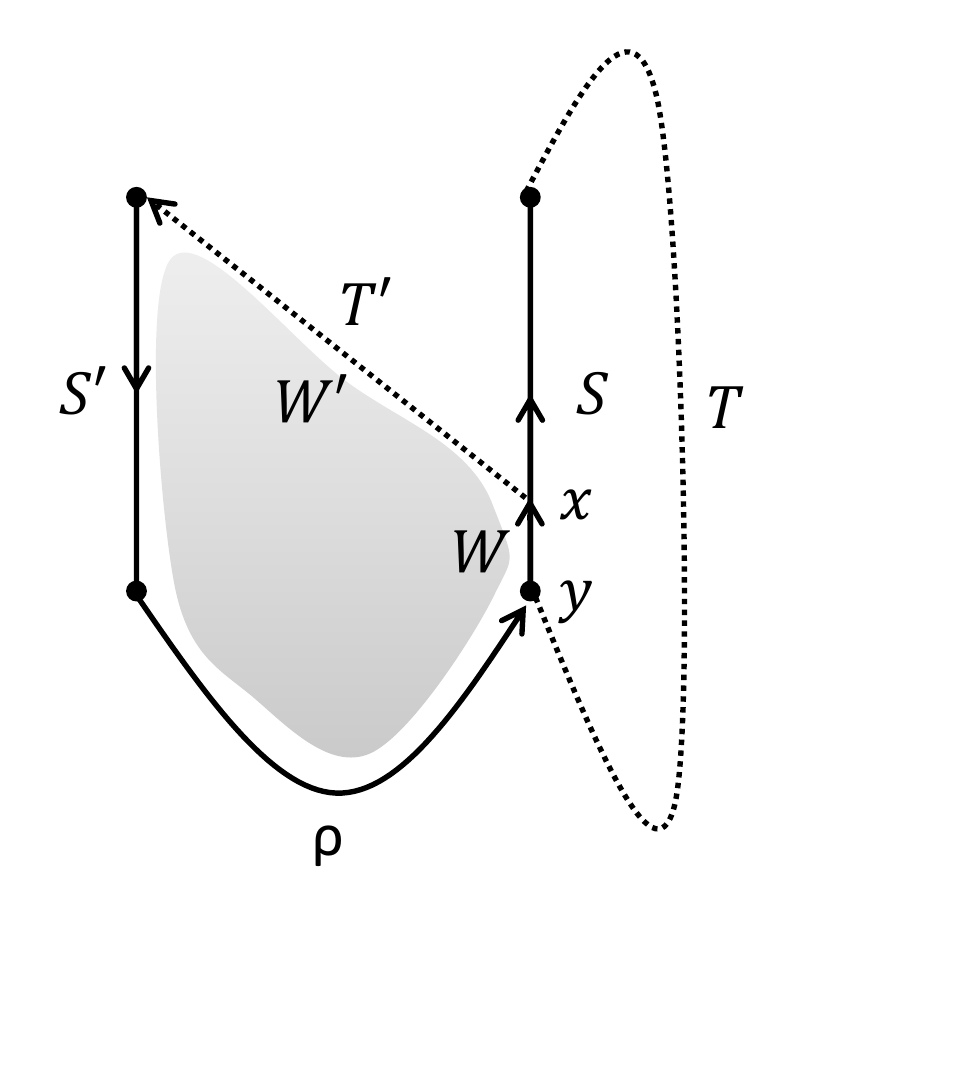}
\caption{Construction in the proof of Lemma~\ref{lem:longcycle}.
The $k$-cycles $P$ and $P'$ intersect in a vertex $x$ that belongs to $S$.
}\label{fig:clutterxonS}
\end{figure} 
We distinguish two cases, which are illustrated by Figures~\ref{fig:clutter} and~\ref{fig:clutterBis}.
\paragraph*{Case that $x\notin V(S')$.}
Let $T'$ be the subpath of $P'$ from $x$ to $\start{S'}$.
The path $T\cdot T'$ is elementary, or else $x$ would not be the first vertex on $T$ that belongs to $P'$, a contradiction.
We distinguish two cases, which are illustrated by the left and right diagrams of~Fig.~\ref{fig:clutter}.
\begin{description}
\item[Case that $\pi$ does not intersect $T\cdot T'$.]
Let $O'$ be the cycle $S\cdot T\cdot T'\cdot S'\cdot\pi$, which will be elementary.
The interior of $O'$ is shaded in the left diagram of Fig.~\ref{fig:clutter}.
Assume, toward a contradiction, that $O'$ has length $\geq 2k$.
Then, $S\cdot T$ and $T'\cdot S'$ are two seams of $O'$.
Then 
\begin{equation*}
\{S\cdot T, T'\cdot S',S_{m+1},\dots,S_{\ell-1}\}
\end{equation*}
is a covering of $O'$ that is preferred over $C^{\ast}$ (as can bee seen by comparing with Equation~\ref{eq:cstar}), a contradiction.
We conclude by contradiction that $O'$ has length~$k$.
But then $S'\cdot\pi\cdot S$ is a seam in $O$.
Then 
\begin{equation*}
\{S'\cdot\pi\cdot S,S_{1},\dots,S_{m-1}\}
\end{equation*}
is a covering of $O$ that is preferred over $C^{\ast}$, a contradiction. 
\item[Case that $\pi$ intersects $T\cdot T'$.]
Let $\rho$ be the shortest prefix of $\pi$ that ends in a vertex (call it $y$) that belongs to $T\cdot T'$.
Assume, toward a contradiction, that $\rho$ is the empty path.
Then, $y=\fin{S'}$.
Since $y\in V(S')$, it must be the case that $y=\start{T'}$, hence $x=y$.
But then $x\in V(S')$, a contradiction.
We conclude by contradiction that $\rho$ contains at least one edge.
Let $U$ denote the suffix of $T\cdot T'$ that starts in $y$. 
Let $O'$ be the cycle $U\cdot S'\cdot\rho$, which will be elementary.
The interior of $O'$ is shaded in the right diagram of Fig.~\ref{fig:clutter}.
Assume, toward a contradiction, that $O'$ has length $\geq 2k$.
We distinguish two cases.
\begin{itemize}
\item
Case that $y\in V(T')$.
Then $U\cdot S'$ is a seam in $O'$.
Then $O'$ can be covered by $U\cdot S'$ together with the seams in $C^{\ast}$ that cover $\rho$ (where the seam of $C^{\ast}$ that covers a suffix of $\rho$  may need to be truncated at $y$).
\item
Case that $y\not\in V(T')$.
Then, $U=W\cdot W'$ where $W$ is the prefix of $U$ that ends in $x$ (and thus contains $y$),
and $W'$ is the suffix of $U$ that starts from $x$.
Then $W$ and $W'\cdot S'$ are seams of $O'$.
Then $O'$ can be covered by $W$, $W'\cdot S'$, and the seams in $C^{\ast}$ that cover $\rho$ (where again the seam covering a suffix of $\rho$  may need to be truncated at $y$).
\end{itemize}
It can now be seen that $O'$ has a covering of cardinality $\ell-m+1<\ell=\card{C^{\ast}}$.
Thus, $O'$ has a covering that is preferred over $C^{\ast}$, a contradiction.
We conclude by contradiction that $O'$ has length~$k$.
Then $S'\cdot\rho$ is a seam in $O$.
Then
\begin{equation*}
\{S,S_{1},\dots,S_{m-1},S'\cdot\rho,S_{m+1},\dots,S_{\ell-1}\}
\end{equation*}
is a (not necessarily minimal) covering of $O$ that is preferred over $C^{\ast}$, a contradiction. 

We note incidentally that in this case, it is possible that $x$ belongs to $S$, as illustrated by Fig.~\ref{fig:clutterxonS}.
\end{description}

\begin{figure}[h]\centering
\includegraphics[scale=0.50,trim={0cm 8cm 10cm 2cm},clip]{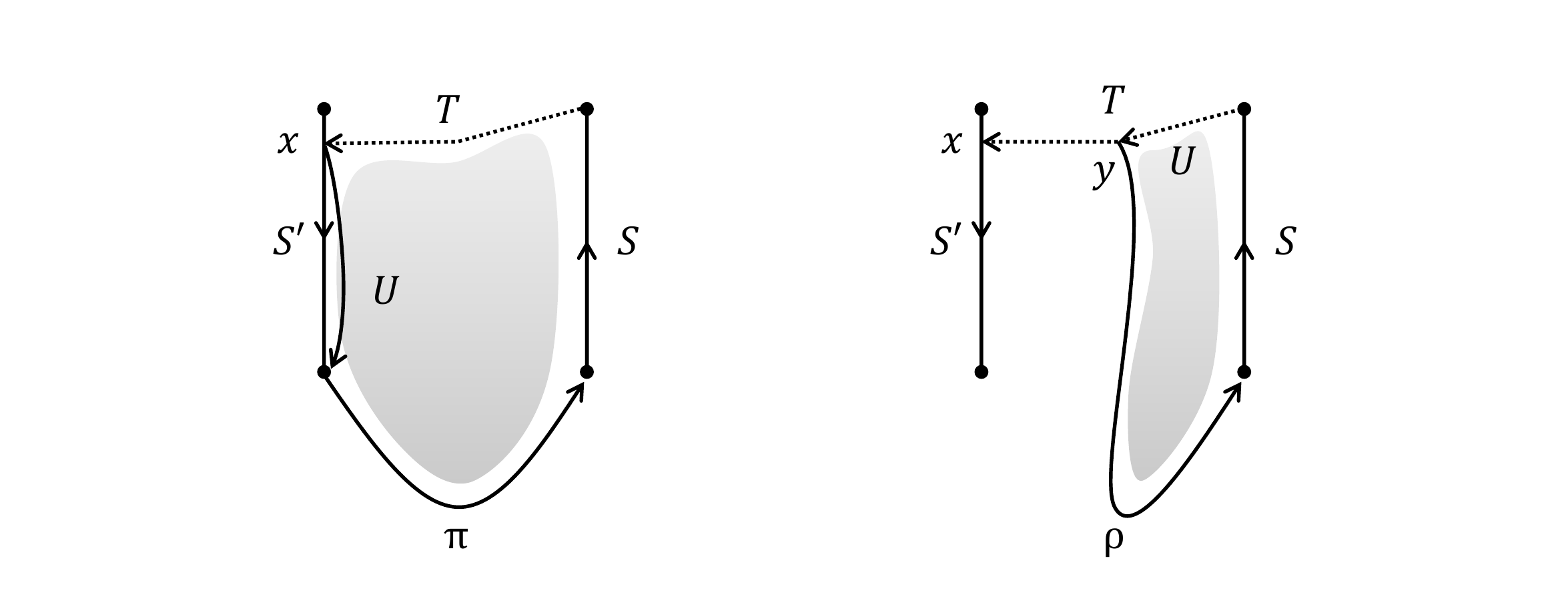}
\caption{Construction in the proof of Lemma~\ref{lem:longcycle}.
The $k$-cycles $P$ and $P'$ intersect in a vertex $x$ that belongs to $S'$.
The path $S'\cdot\pi\cdot S$ is a subpath of $O$.
In the left diagram, $\pi$ does not intersect $T$.
In the right diagram, $\pi$ intersects $T$ in $y$.
}\label{fig:clutterBis}
\end{figure}

\paragraph*{Case that $x\in V(S')$.}
We distinguish two cases, which are illustrated by the left and right diagrams of~Fig.~\ref{fig:clutterBis}.
\begin{description}
\item[Case that $\pi$ does not intersect $T$.]
Let $U$ be the suffix of $S'$ that starts in $x$.
Note incidentally that if $x=\start{S'}$, then $U=S'$.
Let $O'$ be the cycle $U\cdot\pi\cdot S\cdot T$, which will be elementary.
Assume, toward a contradiction, that $O'$ has length~$\geq~2k$.
Then, $S\cdot T$ is a seam in $O'$.
Then 
\begin{equation*}
\{S\cdot T, U,S_{m+1},\dots,S_{\ell-1}\}
\end{equation*}
is a covering of $O'$ that is preferred over $C^{\ast}$, a contradiction.
We conclude by contradiction that $O'$ has length~$k$.
Then $U\cdot\pi\cdot S$ is a seam in $O$.
Then
\begin{equation*}
\{U\cdot\pi\cdot S,S_{1},\dots,S_{m-1},S'\}
\end{equation*}
is a (not necessarily minimal) covering of $O$ that is preferred over $C^{\ast}$, a contradiction. 
\item[Case that $\pi$ intersects $T$.]
Let $U$ be the shortest prefix of $T$ that ends in a vertex (call it $y$) that belongs to $\pi$.
Let $\rho$ be the suffix of $\pi$ that starts in $y$.
Assume, toward a contradiction, that $\rho$ is the empty path.
Then, $y=\start{S}$.
Since $y$ is on the subpath of $P$ from $\fin{S}$ to $x$, it must be the case that $x=y$.
Then, $\start{S}\in V(S')$, a contradiction.
We conclude by contradiction that $\rho$ contains at least one edge.
Let $O'$ be the cycle $S\cdot U\cdot\rho$, which is elementary. 
Assume, toward a contradiction, that $O'$ has length~$\geq 2k$.
Then, $S\cdot U$ is a seam in $O'$.
It can be easily seen that $O'$ has a covering that is preferred over $C^{\ast}$ ($O'$ can be covered by the seam $S\cdot U$ together with the seams in $O^{\ast}$ that cover $\rho$), a contradiction.
We conclude by contradiction that $O'$ has length~$k$.
Then, $\rho\cdot S$ is a seam in  $O$.
Then
\begin{equation*}
\{\rho\cdot S,S_{1},\dots,S_{m-1},S',S_{m+1},\dots,S_{\ell-1}\}
\end{equation*}
is a (not necessarily minimal) covering of $O$ that is preferred over $C^{\ast}$, a contradiction.
\end{description}
It is now correct to conclude that $(P_{0},P_{1},\dots,P_{\ell-1},P_{0})$ is a chordless cycle in $\hat{G}$.
We now distinguish two cases.
\begin{description}
\item[Case $\ell\geq 2k$.]
Then $\hat{G}$ has a chordless cycle of length $\geq 2k$.
\item[Case $\ell<2k$].
Since $O$ contains edges from only $\ell$ $k$-cycles,
and since every $k$-cycle can contribute at most $k-1$ edges to $O$ (because $O$ is elementary),
it follows that the length of $O$ is at most $(2k-1)(k-1)$.
Since the length of $O$ must be a multiple of $k$ and $(2k-1)(k-1)<2(k-1)k$, the length of $O$ cannot exceed the greatest multiple of $k$ that is strictly smaller than $2(k-1)k=(2k-2)k$.
Therefore, the length of $O$ is at most $(2k-3)k$. 
\end{description}
This concludes the proof of \framebox{\ref{it:pylop}$\implies$\ref{it:polyp}}.

The equivalence \framebox{\ref{it:polyp}$\iff$\ref{it:pylop}} is now used to develop a logspace algorithm  for deciding whether $G$ contains an elementary cycle of length $\geq 2k$.

All elementary cycles of length between $2k$ and $(2k-3)k$ can obviously be found in $\FO$.
The graph $\po{G}$ can clearly be constructed in logarithmic space.
The existence of a chordless cycle can be computed in logarithmic space, as follows:
check whether there exists a path $(P_{1},P_{2},\dots,P_{2k})$ whose two subpaths of length~$2k-1$ are chordless and whose endpoints (i.e., $P_{1}$ and $P_{2k}$) are either equal or connected by a path that uses no vertex in $\{P_{2},\dots,P_{2k-1}\}$. 
Since undirected connectivity can be decided in logarithmic space~\cite{DBLP:journals/jacm/Reingold08}, it is correct to conclude that $\problem{LONGCYCLE}(k)$ is in logarithmic space.
\end{proof}

\begin{figure*}\centering
\includegraphics[scale=0.8]{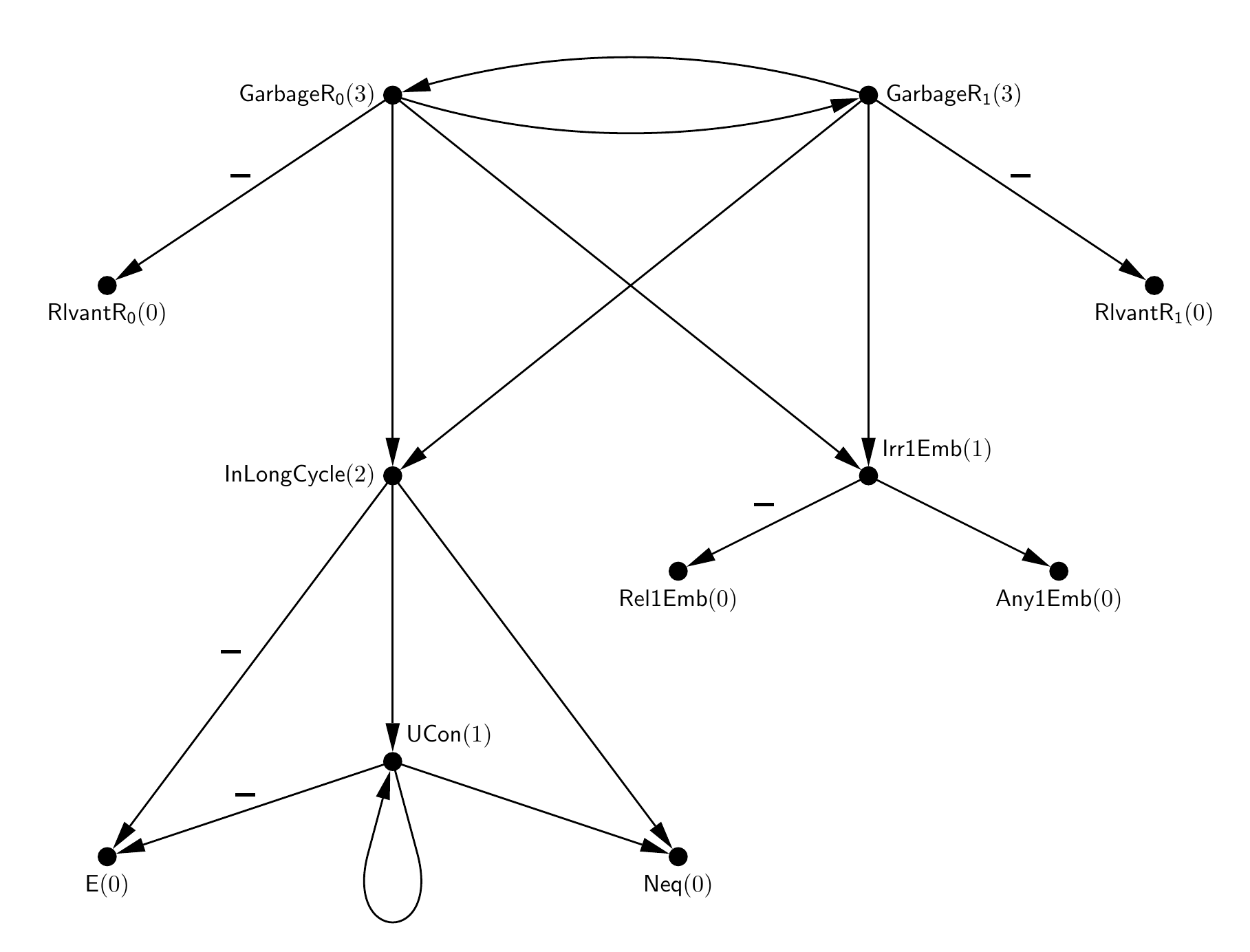}
\caption{
Precedence graph for the symmetric stratified Datalog program constructed in the proof of Lemma~\ref{lem:scrub},
for a cycle $R_{0}\markov R_{1}\markov R_{0}$.
Equality and disequality predicates have been omitted.
Edges with a $-$ label represent negative dependencies.
The numbers between parentheses are the strata.
}\label{fig:pdg}
\end{figure*}

\medskip

\begin{proof}[Proof of Lemma~\ref{lem:scrub}]
Let the elementary cycle in the \mgraph be $C=F_{0}\markov F_{1}\markov\dotsm\markov F_{k-1}\markov F_{0}$ where $k\geq 2$ is the length of the cycle.
For each $i\in\{0,\dots,k-1\}$, let $F_{i}=R_{i}(\underline{\vec{x}_{i}},\vec{y}_{i})$.
Further, for every $i\in\{0,\dots,k-1\}$ such that the signature of $R_{i}$ is $\signature{n}{\ell}$:
\begin{itemize}
\item
let $\vec{u}_{i}$ and $\vec{w}_{i}$ be sequences of fresh distinct variables of lengths~$\ell$ and~$n-\ell$ respectively.
Thus, the atom $R_{i}(\underline{\vec{u}_{i}},\vec{w}_{i})$ is syntactically well-defined;
\item
let $\good{R_{i}}$ be an IDB predicate of arity $n$;
\item
let $\del{R_{i}}$ be an IDB predicate of arity $\ell$.
\end{itemize}
Informally,
whenever a fact $\del{R_{i}}(\vec{a}_{i})$ will be derived,
then the input database contains a block $R_{i}(\underline{\vec{a}_{i}},\blockfiller)$ that belongs to the maximal garbage set for $C$.
The precedence graph of our Datalog program, with an indication of the strata, is shown in Fig.~\ref{fig:pdg}. 
We start by defining the IDB predicates $\good{R_{i}}$, where $\good{R_{i}}(\underline{\vec{a}_{i}},\vec{b_{i}})$ indicates that $R_{i}(\underline{\vec{a}_{i}},\vec{b_{i}})$ belongs to a relevant $1$-embedding of $C$.
For every $i\in\{0,1,\dots,k-1\}$, we add the rules:
$$
\begin{datalogpgm}
\good{R_{i}}(\vec{x}_{i},\vec{y}_{i}) & q\\[1.0ex]
\del{R_{i}}(\vec{u}_{i}) & R_{i}(\vec{u}_{i},\vec{w}_{i}), \neg\good{R_{i}}(\vec{u}_{i},\vec{w_{i}})\\
\end{datalogpgm}
$$
These rules implement condition~\ref{it:scrubnotdangling} in Lemma~\ref{lem:scrubalgo};
condition~\ref{it:keyequalclosure} is also captured since the argument of $\del{R_{i}}$ is limited to primary-key positions, which identify blocks rather than individual facts.
To implement condition~\ref{it:recursive} in Lemma~\ref{lem:scrubalgo}, we add,
for every $i,j\in\{0,1,\dots,k-1\}$ such that $i<j$, the rules:
$$
\begin{datalogpgm}
\del{R_{i}}(\vec{x}_{i}) & q, \del{R_{j}}(\vec{x}_{j})\\[1.0ex]
\del{R_{j}}(\vec{x}_{j}) & q, \del{R_{i}}(\vec{x}_{i})
\end{datalogpgm}
$$
These rules are each other's symmetric version.

For every variable $x$ and every $i\in\{\supone,\supthree\}\cup\{0,1,2,\dots\}$,
we write $\rename{x}{i}$ to denote a fresh variable such that $\rename{x}{i}=\rename{y}{j}$ if and only if $x=y$ and $i=j$.
This notation extends to sequences of variables and queries in the natural way.
For example, if $\vec{x}=\tuple{x_{1},x_{2},\dots,x_{n}}$,
then $\rename{\vec{x}}{i}=\tuple{\rename{x_{1}}{i},\rename{x_{2}}{i},\dots,\rename{x_{n}}{i}}$.
If $c$ is a constant, then we define $\rename{c}{i}=c$.

We will need to compare composite primary-key values for disequality.
To this extent, we add the following rules for every $i\in\{0,1,\dots,k-1\}$:
$$
\begin{datalogpgm}
\eqpred{R_{i}}(\vec{x}_{i},\vec{x}_{i}) &  R_{i}(\vec{x}_{i},\vec{y}_{i})\\[1.0ex]
\diseqpred{R_{i}}(\vec{x}_{i},\rename{\vec{x}_{i}}{\supone}) & 
\left\{
\begin{array}{l}
 R_{i}(\vec{x}_{i},\vec{y}_{i}),
 R_{i}(\rename{\vec{x}_{i}}{\supone},\rename{\vec{y}_{i}}{\supone}),\\
 \neg\eqpred{R_{i}}(\vec{x}_{i},\rename{\vec{x}_{i}}{\supone})
\end{array}
\right\} 
\end{datalogpgm}
$$
Note that the rule for $\eqpred{R_{i}}$ only applies to $R_{i}$-facts that satisfy the rule body $\{R_{i}(\vec{x}_{i},\vec{y}_{i})\}$.
This suffices, because $R_{i}$-facts falsifying  $\{R_{i}(\vec{x}_{i},\vec{y}_{i})\}$ cannot belong to a relevant $1$-embedding,
and will be added to the garbage set by previous rules.

In what follows, $\diseqpred{R_{i}}(\vec{x}_{i},\rename{\vec{x}_{i}}{\supone})$ will be abbreviated as $\vec{x}_{i}\neq_{R_{i}}\rename{\vec{x}_{i}}{\supone}$.
Likewise, $\eqpred{R_{i}}(\vec{x}_{i},\rename{\vec{x}_{i}}{\supone})$ will be abbreviated as $\vec{x}_{i}=_{R_{i}}\rename{\vec{x}_{i}}{\supone}$.
Of course, in Datalog with $\neq$, these predicates can be expressed by using disequality ($\neq$) instead of negation ($\neg$).

The predicate $\anyone$ computes all $1$-embeddings of $C$.
Then, $\relone$ computes the relevant $1$-embeddings, and $\irrone$ the irrelevant $1$-embeddings, which is needed in the implementation of condition~\ref{it:scruboneembedding} in Lemma~\ref{lem:scrubalgo}.

$$
\begin{array}{l}
\begin{datalogpgm}
\anyone(
\rename{\vec{x}_{0}}{0},\rename{\vec{y}_{0}}{0},
\rename{\vec{x}_{1}}{1},\rename{\vec{y}_{1}}{1},
\dots,
\rename{\vec{x}_{k-1}}{k-1},\rename{\vec{y}_{k-1}}{k-1}
)
&
\left\{
\begin{array}{l}
\rename{q}{0},\rename{q}{1},\dots,\rename{q}{k-1},\\[1.0ex]
\begin{array}{ccc}
\rename{\vec{x}_{0}}{0}&=_{R_{0}}&\rename{\vec{x}_{0}}{k-1},\\
\rename{\vec{x}_{1}}{1}&=_{R_{1}}&\rename{\vec{x}_{1}}{0},\\
\rename{\vec{x}_{2}}{2}&=_{R_{2}}&\rename{\vec{x}_{2}}{1},\\
&\vdots&\\
\rename{\vec{x}_{k-1}}{k-1}&=_{R_{k-1}}&\rename{\vec{x}_{k-1}}{k-2}
\end{array}
\end{array}
\right\}
\end{datalogpgm}
\\
\\
\begin{datalogpgm}
\relone(
\vec{x}_{0},\vec{y}_{0},
\vec{x}_{1},\vec{y}_{1},
\dots,
\vec{x}_{k-1},\vec{y}_{k-1}
)
&
q
\end{datalogpgm}\\[1.0ex]
\begin{datalogpgm}
\irrone(
\rename{\vec{x}_{0}}{0},
\rename{\vec{x}_{1}}{1},
\dots,
\rename{\vec{x}_{k-1}}{k-1}
)
&
\left\{
\begin{array}{l}
\anyone(
\rename{\vec{x}_{0}}{0},\rename{\vec{y}_{0}}{0},
\rename{\vec{x}_{1}}{1},\rename{\vec{y}_{1}}{1},
\dots,
\rename{\vec{x}_{k-1}}{k-1},\rename{\vec{y}_{k-1}}{k-1}
),\\[1.0ex]
\neg\relone(
\rename{\vec{x}_{0}}{0},\rename{\vec{y}_{0}}{0},
\rename{\vec{x}_{1}}{1},\rename{\vec{y}_{1}}{1},
\dots,
\rename{\vec{x}_{k-1}}{k-1},\rename{\vec{y}_{k-1}}{k-1}
)
\end{array}\right\}
\end{datalogpgm}
\end{array}
$$
To finish the implementation of condition~\ref{it:scruboneembedding} in Lemma~\ref{lem:scrubalgo}, we add, for every $i\in\{0,\dots,k-1\}$, the following rules:
$$
\begin{datalogpgm}
\del{R_{i}}(\rename{\vec{x}_{i}}{i}) & \irrone(
\rename{\vec{x}_{0}}{0},
\rename{\vec{x}_{1}}{1},
\dots,
\rename{\vec{x}_{k-1}}{k-1}
)
\end{datalogpgm}
$$

\begin{figure}\centering
\includegraphics[scale=0.8]{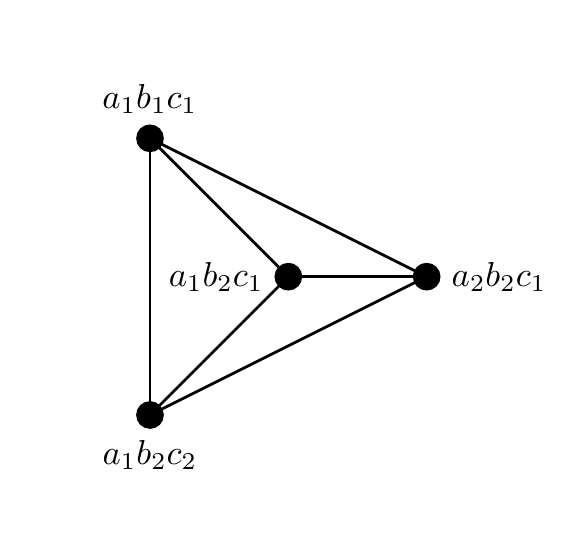}
\caption{$\polyedge$-edges for the example of Fig.~\ref{fig:gt}, where $k=3$.
There is no chordless cycle of length~$2k=6$.
However, since the inequalities $2\leq n\leq 2k-3$ have solutions $n=2$ and $n=3$,
the Datalog program will also contain non-recursive rules for detecting $2$-embeddings and $3$-embeddings. 
}\label{fig:guidedtourpoly}
\end{figure}

\begin{figure*}\centering
\captionsetup[subfigure]{justification=centering}
\begin{subfigure}[t]{\textwidth}\centering
\includegraphics[scale=0.8]{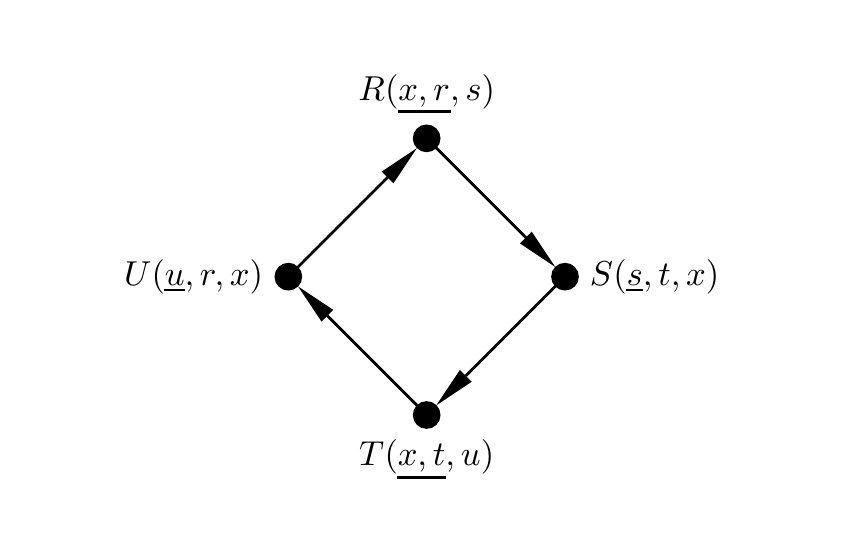}
\caption{\mgraph.}\label{fig:gcm}
\end{subfigure}
\newline
\captionsetup[subfigure]{justification=centering}
\begin{subfigure}[t]{0.5\textwidth}\centering
\includegraphics[scale=0.8]{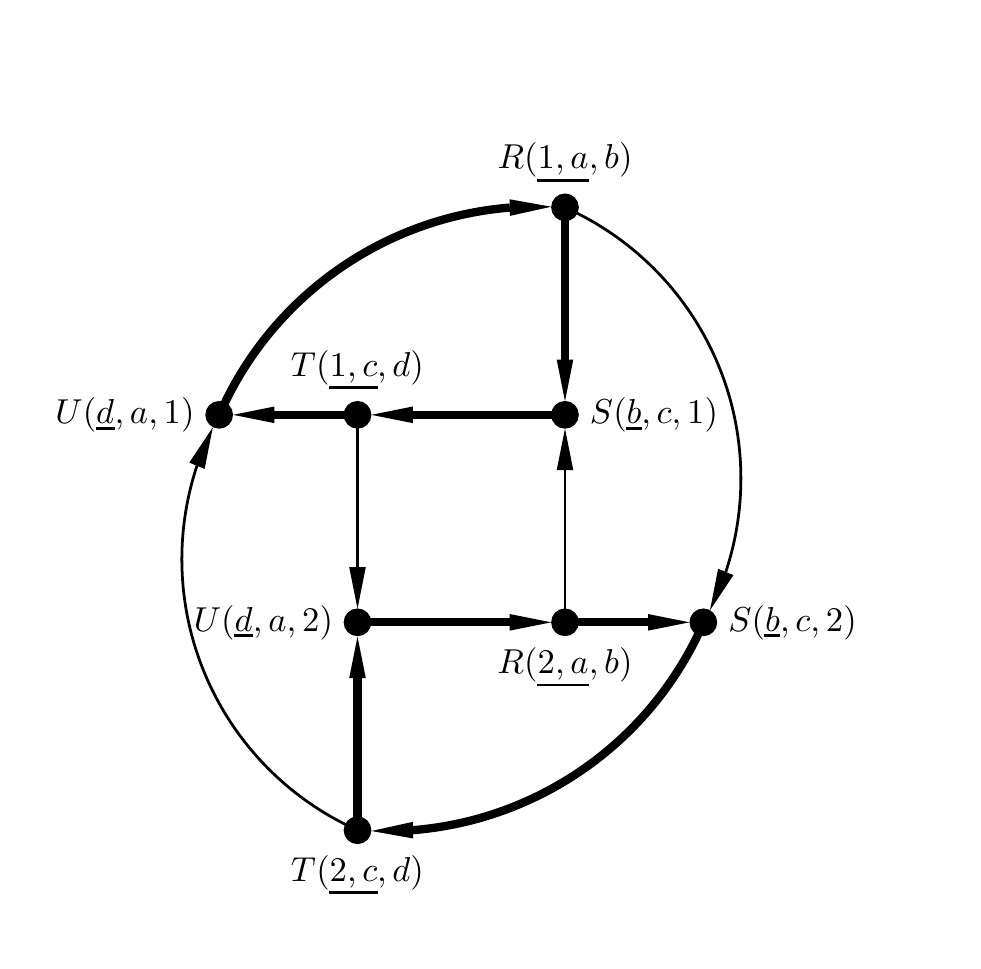}
\caption{$\cmhook{C}$-graph.}\label{fig:gcc}
\end{subfigure}
\begin{subfigure}[t]{0.4\textwidth}\centering
\includegraphics[scale=0.8]{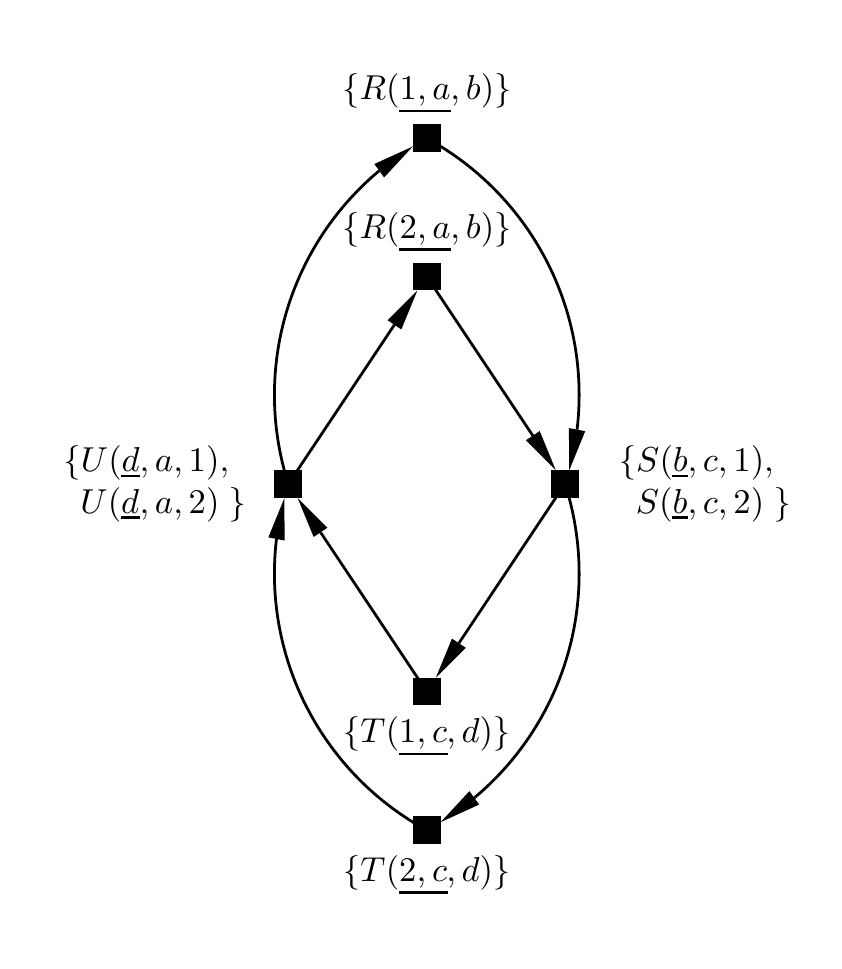}
\caption{Block-quotient graph.}\label{fig:gchook}
\end{subfigure}
\caption{
The $\cmhook{C}$-graph contains two relevant $1$-embeddings (thick arrows) and four irrelevant $1$-embeddings.
The block-quotient graph contains four elementary cycles of length~$4$; the outermost cycle (curved arrows) and the innermost cycle (straight arrows) are induced by irrelevant $1$-embeddings. }\label{fig:garbagecomputation}
\end{figure*}

We now add rules that implement the algorithm sketched in the proof of Lemma~\ref{lem:longcycle}, capturing condition~\ref{it:scrubmultiple} in Lemma~\ref{lem:scrubalgo}.
From here on, whenever $C$ occurs as an argument of a predicate,
then it is understood to be a shorthand for the sequence
$\tuple{\vec{x}_{0},\vec{x}_{1},\dots,\vec{x}_{k-1}}$.
The IDB predicate $\polyedge$ is used for undirected edges between vertices that are $k$-cycles in the block-quotient graph,
and the predicate $\polydistinct$ tests whether two vertices are distinct.
Figure~\ref{fig:guidedtourpoly} shows the $\polyedge$-edges for the example of Fig.~\ref{fig:gt}.
It suffices to consider only $k$-cycles of the block-quotient graph induced by relevant $1$-embeddings of~$C$, because irrelevant $1$-embeddings are already added to the garbage set by previous rules.
Figure~\ref{fig:garbagecomputation} illustrates that $k$-cycles in the block-quotient graph can be induced by $1$-embeddings of~$C$ that are not relevant; such $k$-cycles, however, are ignored by our Datalog program.

For all $i,j\in\{0,\dots,k-1\}$ such that $i\neq j$,
add the following rules:
$$
\begin{datalogpgm}
\polyedge(\rename{C}{0},\rename{C}{1})
&
\left\{
\begin{array}{l}
\rename{q}{0}, \rename{q}{1},\\[1.0ex]
\begin{array}{ccc}
\rename{\vec{x}_{i}}{0} & =_{R_{i}} & \rename{\vec{x}_{i}}{1},\\
\rename{\vec{x}_{j}}{0} & \neq_{R_{j}} & \rename{\vec{x}_{j}}{1}
\end{array}
\end{array}
\right\}
\end{datalogpgm}
$$

For all $j\in\{0,\dots,k-1\}$,
add the following rules:
$$
\begin{datalogpgm}
\polydistinct(\rename{C}{0},\rename{C}{1})
&
\left\{
\begin{array}{l}
\rename{q}{0}, \rename{q}{1},\\[1.0ex]
\begin{array}{ccc}
\rename{\vec{x}_{j}}{0} & \neq_{R_{j}} & \rename{\vec{x}_{j}}{1}
\end{array}
\end{array}
\right\}
\end{datalogpgm}
$$

The predicate $\polyconnected$ computes undirected connectivity in the graph defined by $\polyedge$; 
it takes $2k$ vertices as operands, and holds true if there exists an undirected path between the first two operands such that no vertex on the path is equal to or adjacent to any of the remaining $2k-2$ operands.
$$
\begin{datalogpgm}
\polyconnected(\rename{C}{1},\rename{C}{1},\rename{C}{3},\dots,\rename{C}{2k})
&
\left\{
\begin{array}{l}
\polydistinct(\rename{C}{1},\rename{C}{3}),
\neg\polyedge(\rename{C}{1},\rename{C}{3}),\\
\multicolumn{1}{c}{\vdots}\\
\polydistinct(\rename{C}{1},\rename{C}{2k}),
\neg\polyedge(\rename{C}{1},\rename{C}{2k})\\
\end{array}
\right\}
\end{datalogpgm}
$$
$$
\begin{datalogpgm}
\polyconnected(\rename{C}{1},\rename{C}{2},\rename{C}{3},\dots,\rename{C}{2k})
&
\left\{
\begin{array}{l}
\polyconnected(\rename{C}{1},\rename{C}{\supone},\rename{C}{3},\dots,\rename{C}{2k}),
\polyedge(\rename{C}{\supone},\rename{C}{2}),\\
\begin{array}{l}
\polydistinct(\rename{C}{\supone},\rename{C}{3}),\neg\polyedge(\rename{C}{\supone},\rename{C}{3}),\\
\multicolumn{1}{c}{\vdots}\\
\polydistinct(\rename{C}{\supone},\rename{C}{2k}),\neg\polyedge(\rename{C}{\supone},\rename{C}{2k}),\\
\polydistinct(\rename{C}{2},\rename{C}{3}),\neg\polyedge(\rename{C}{2},\rename{C}{3}),\\
\multicolumn{1}{c}{\vdots}\\
\polydistinct(\rename{C}{2},\rename{C}{2k}),\neg\polyedge(\rename{C}{2},\rename{C}{2k})
\end{array}
\end{array}
\right\}
\end{datalogpgm}
$$
$$
\begin{datalogpgm}
\polyconnected(\rename{C}{1},\rename{C}{\supone},\rename{C}{3},\dots,\rename{C}{2k})
&
\left\{
\begin{array}{l}
\polyconnected(\rename{C}{1},\rename{C}{2},\rename{C}{3},\dots,\rename{C}{2k}),
\polyedge(\rename{C}{\supone},\rename{C}{2}),\\
\begin{array}{l}
\polydistinct(\rename{C}{\supone},\rename{C}{3}),\neg\polyedge(\rename{C}{\supone},\rename{C}{3}),\\
\multicolumn{1}{c}{\vdots}\\
\polydistinct(\rename{C}{\supone},\rename{C}{2k}),\neg\polyedge(\rename{C}{\supone},\rename{C}{2k}),\\
\polydistinct(\rename{C}{2},\rename{C}{3}),\neg\polyedge(\rename{C}{2},\rename{C}{3}),\\
\multicolumn{1}{c}{\vdots}\\
\polydistinct(\rename{C}{2},\rename{C}{2k}),\neg\polyedge(\rename{C}{2},\rename{C}{2k})
\end{array}
\end{array}
\right\}
\end{datalogpgm}
$$
The latter two rules are each other's symmetric version.
We are now ready to encode the two conditions for the existence of an elementary directed cycle of length $\geq 2k$ in the proof of Lemma~\ref{lem:longcycle}.
We add non-recursive rules that detect $n$-embeddings for every  $n$ such that $2\leq n\leq 2k-3$.
We show here only the rules for $n=2$, i.e., for $\cmhook{C}$- cycles of length $2k$ without key-equal atoms.
We add, for every $i\in\{0,\dots,k-1\}$, the following rules:
$$
\begin{datalogpgm}
\del{R_{i}}(\vec{x}_{i})
&
\left\{
\begin{array}{l}
\rename{q}{0},\rename{q}{1},\dots,\rename{q}{k-1},
\rename{q}{k},\rename{q}{k+1},\dots,\rename{q}{2k-1},\\[1.0ex]
\begin{array}{ccc}
\rename{\vec{x}_{0}}{0}&=_{R_{0}}&\rename{\vec{x}_{0}}{2k-1},\\
\rename{\vec{x}_{1}}{1}&=_{R_{1}}&\rename{\vec{x}_{1}}{0},\\
&\vdots&\\
\rename{\vec{x}_{k-1}}{k-1}&=_{R_{k-1}}&\rename{\vec{x}_{k-1}}{k-2},\\
\rename{\vec{x}_{0}}{k}&=_{R_{0}}&\rename{\vec{x}_{0}}{k-1},\\
\rename{\vec{x}_{1}}{k+1}&=_{R_{1}}&\rename{\vec{x}_{1}}{k},\\
&\vdots&\\
\rename{\vec{x}_{k-1}}{2k-1}&=_{R_{k-1}}&\rename{\vec{x}_{k-1}}{2k-2},\\[1.0ex]
\rename{\vec{x}_{0}}{0} & \neq_{R_{0}} & \rename{\vec{x}_{0}}{k},\\
\rename{\vec{x}_{1}}{1} & \neq_{R_{1}} & \rename{\vec{x}_{1}}{k+1},\\
&\vdots&\\
\rename{\vec{x}_{k-1}}{k-1} & \neq_{R_{k-1}} & \rename{\vec{x}_{k-1}}{2k-1}\\
\end{array}
\end{array}
\right\}
\end{datalogpgm}
$$
The following rule checks whether $C$ belongs to a chordless cycle of length $\geq 2k$.
$$
\begin{datalogpgm}
\polylong(\rename{C}{1})
&
\left\{
\begin{array}{l}
\polyedge(\rename{C}{1},\rename{C}{2}),
\polyedge(\rename{C}{2},\rename{C}{3}),
\dots,
\polyedge(\rename{C}{2k-1},\rename{C}{2k}),\\
\{\neg\polyedge(\rename{C}{i},\rename{C}{j})\}_{\scriptsize
\begin{array}{l}
1\leq i\leq 2k-2,\\ 
i+2\leq j\leq 2k,\\ 
(i,j)\neq (1,2k)
\end{array}
},\\
\\
\{\polydistinct(\rename{C}{i},\rename{C}{j})\}_{\scriptsize
\begin{array}{l}
1\leq i<j\leq 2k,\\ (i,j)\neq (1,2k)
\end{array}
},\\
\\
\polyconnected(\rename{C}{1},\rename{C}{2k},\rename{C}{2},\dots,\rename{C}{2k-1})
\end{array}
\right\}
\end{datalogpgm}
$$
Finally, to finish the implementation of condition~\ref{it:scrubmultiple} in Lemma~\ref{lem:scrubalgo}, we add, for every $i\in\{0,\dots,k-1\}$, the following rules:
$$
\begin{datalogpgm}
\del{R_{i}}(\vec{x}_{i}) & \polylong(C)
\end{datalogpgm}
$$
This concludes the computation of the maximal garbage set for $C$.
\end{proof}


The following example illustrates the Datalog program in the proof of Lemma~\ref{lem:scrub}.

\begin{example}
Let $q=\{R(\underline{x},y,z), S(\underline{y},x,z), U(\underline{z},a)\}$, where $a$ is a constant.
We show a program in symmetric stratified Datalog that computes the garbage set for the \mcycle
$C=R(\underline{x},y,z)\markov S(\underline{y},x,z)\markov R(\underline{x},y,z)$.
In this example, $k=2$.
The program is constructed as in the proof of Lemma~\ref{lem:scrub} (up to some straightforward syntactic simplifications).

$R$-facts and $S$-facts belong to the maximal garbage set
if they do not belong to a relevant $1$-embedding.
This is expressed by the following rules.
$$
\begin{datalogpgm}
\good{R}(x,y,z) & R(x,y,z), S(y,x,z), U(z,a)\\[1.0ex]
\del{R}(x) & R(x,y,z), \neg\good{R}(x,y,z)\\[1.0ex]
\good{S}(y,x,z) & R(y,x,z), S(y,x,z), U(z,a)\\[1.0ex]
\del{S}(y) & S(y,x,z), \neg\good{S}(y,x,z)
\end{datalogpgm}
$$
If some $R$-fact or $S$-fact of a relevant $1$-embedding belongs to the maximal garbage set, then every fact of that $1$-embedding belongs to the maximal garbage set.
This is expressed by the following rules.
$$
\begin{datalogpgm}
\del{R}(x) & R(x,y,z), S(y,x,z), U(z,a), \del{S}(y)\\[1.0ex]
\del{S}(y) & R(x,y,z), S(y,x,z), U(z,a), \del{R}(x)
\end{datalogpgm}
$$
Note that the predicates $\del{R}$ and $\del{S}$ refer to blocks:
whenever a fact is added to the garbage set, its entire block is added.
The following rules compute irrelevant $1$-embeddings.
$$
\begin{datalogpgm}
\anyone(x,y,z,z') & 
\left\{
\begin{array}{l}
R(x,y,z), S(y,x,z), U(z,a),\\
R(x,y,z'), S(y,x,z'), U(z',a)
\end{array}
\right\}\\\spacebetweenrules
\relone(x,y,z,z) & R(x,y,z), S(y,x,z), U(z,a)\\[1.0ex]
\irrone(x,y) & 
\anyone(x,y,z,z'),
\neg\relone(x,y,z,z')
\end{datalogpgm}
$$
The predicate $\polyedge$ is used for edges between vertices;
each vertex is a $(x,y)$-value.
The predicate $\polydistinct$ expresses disequality of vertices. 
$$
\begin{datalogpgm}
\cpolyedge{x}{y}{x}{y'}
& 
\left\{
\begin{array}{l}
R(x,y,z), S(y,x,z), U(z,a),\\
R(x,y',z'), S(y',x,z'), U(z',a),
y\neq y'
\end{array}
\right\}\\\spacebetweenrules
\cpolyedge{x}{y}{x'}{y}
& 
\left\{
\begin{array}{l}
R(x,y,z), S(y,x,z), U(z,a),\\
R(x',y,z'), S(y,x',z'), U(z',a),
x\neq x'
\end{array}
\right\}
\end{datalogpgm}
$$

$$
\begin{datalogpgm}
\cpolydistinct{x}{y}{x'}{y'}
& 
\left\{
\begin{array}{l}
R(x,y,z), S(y,x,z), U(z,a),\\
R(x',y',z'), S(y',x',z'), U(z',a),
x\neq x'
\end{array}
\right\}\\\spacebetweenrules
\cpolydistinct{x}{y}{x'}{y'}
& 
\left\{
\begin{array}{l}
R(x,y,z), S(y,x,z), U(z,a),\\
R(x',y',z'), S(y',x',z'), U(z',a),
y\neq y'
\end{array}
\right\}
\end{datalogpgm}
$$
The predicate $\polyconnected$ is used for undirected connectivity of the $\polyedge$-predicate.
In particular, it will be the case that $\polyconnected(a_1,b_1,a_2,b_2,a_3,b_3,a_4,b_4)$ holds true if there exists a path between vertices $(a_1,b_1)$ and $(a_2,b_2)$ such that no vertex on the path is equal to or adjacent to a vertex in $\{(a_3,b_3), (a_4,b_4)\}$.
Recall that each vertex is itself a pair.
$$
\begin{datalogpgm}
\polyconnected(\CC{1},\CC{1},\CC{3},\CC{4})
&
\left\{
\begin{array}{l}
\polydistinct(\CC{1},\CC{3}),
\neg\polyedge(\CC{1},\CC{3}),\\
\polydistinct(\CC{1},\CC{4}),
\neg\polyedge(\CC{1},\CC{4})\\
\end{array}
\right\}
\\\spacebetweenrules
\polyconnected(\CC{1},\CC{2},\CC{3},\CC{4})
&
\left\{
\begin{array}{l}
\polyconnected(\CC{1},\CC{\supone},\CC{3},\CC{4}),
\polyedge(\CC{\supone},\CC{2}),\\
\begin{array}{l}
\polydistinct(\CC{\supone},\CC{3}),\neg\polyedge(\CC{\supone},\CC{3}),\\
\polydistinct(\CC{\supone},\CC{4}),\neg\polyedge(\CC{\supone},\CC{4}),\\
\polydistinct(\CC{2},\CC{3}),\neg\polyedge(\CC{2},\CC{3}),\\
\polydistinct(\CC{2},\CC{4}),\neg\polyedge(\CC{2},\CC{4})
\end{array}
\end{array}
\right\}
\\\spacebetweenrules
\polyconnected(\CC{1},\CC{\supone},\CC{3},\CC{4})
&
\left\{
\begin{array}{l}
\polyconnected(\CC{1},\CC{2},\CC{3},\CC{4}),
\polyedge(\CC{\supone},\CC{2}),\\
\begin{array}{l}
\polydistinct(\CC{\supone},\CC{3}),\neg\polyedge(\CC{\supone},\CC{3}),\\
\polydistinct(\CC{\supone},\CC{4}),\neg\polyedge(\CC{\supone},\CC{4}),\\
\polydistinct(\CC{2},\CC{3}),\neg\polyedge(\CC{2},\CC{3}),\\
\polydistinct(\CC{2},\CC{4}),\neg\polyedge(\CC{2},\CC{4})
\end{array}
\end{array}
\right\}
\end{datalogpgm}
$$
The latter two rules are each other's symmetric version.
The following rule checks whether a vertex $(a_1,b_1)$ belongs to a chordless $\polyedge$-cycle of length $\geq 2k$.

$$
\begin{datalogpgm}
\polylong(\CC{1})
&
\left\{
\begin{array}{l}
\polyedge(\CC{1},\CC{2}),
\polyedge(\CC{2},\CC{3}),
\polyedge(\CC{3},\CC{4}),\\
\neg\polyedge(\CC{1},\CC{3}),
\neg\polyedge(\CC{2},\CC{4}),\\
\polydistinct(\CC{1},\CC{2}),
\polydistinct(\CC{1},\CC{3}),\\
\polydistinct(\CC{2},\CC{3}),
\polydistinct(\CC{2},\CC{4}),\\
\polydistinct(\CC{3},\CC{4}),\\
\polyconnected(\CC{1},\CC{4},\CC{2},\CC{3})
\end{array}
\right\}
\end{datalogpgm}
$$

The following rules add to the maximal garbage sets all $R$-facts and $S$-facts that belong to an irrelevant $1$-embedding or to a strong component of the $\cmhook{C}$-graph that contains an elementary $\cmhook{C}$-cycle of length $\geq 2k$.
Whenever a fact is added, all facts of its block are added.
$$
\begin{datalogpgm}
\del{R}(x) & \polylong(x,y)\\[1.0ex]
\del{S}(y) & \polylong(x,y)\\[1.0ex]
\del{R}(x) & \irrone(x,y)\\[1.0ex]
\del{S}(y) & \irrone(x,y)
\end{datalogpgm}
$$
This terminates the computation of the garbage set.
In general, we have to check the existence of elementary $\cmhook{C}$-cycles of length $nk$ with $2\leq n\leq 2k-3$.
However, for $k=2$, no such $n$ exists.
\end{example}

\medskip

\begin{proof}[Proof of Lemma~\ref{lem:toT}]
Let  $q'=\formula{q\setminus C}\cup\{T\}$.
For every $i\in\{0,1,\dots,k-1\}$, let $F_{i}=R_{i}(\underline{\vec{x}_{i}},\vec{y}_{i})$.

\myparagraph{Proof of the first item}
We show the existence of a reduction from $\cqa{q}$ to the problem $\cqa{q'\cup p}$ that is expressible in $\ssdatalogmin$.
We first describe the reduction, and then show that it can be expressed in $\ssdatalogmin$.

Let $\db_{0}$ be a database that is input to $\cqa{q}$.
By Lemma~\ref{lem:scrub}, we can compute in symmetric stratified Datalog the maximal garbage set $\bfo$ for $C$ in $\db_{0}$.
Let $\db=\db_{0}\setminus\bfo$.
We know, by Lemma~\ref{lem:nomenestomen}, that the problem $\cqa{q}$ has the same answer on instances $\db_{0}$ and $\db$.
Moreover, by Lemma~\ref{lem:together}, every garbage set for~$C$ in $\db$ is empty, which implies, by Lemma~\ref{lem:scrubalgo}, that \emph{(i)}~every $n$-embedding of $C$ in $\db$ must be a relevant $1$-embedding, and \emph{(ii)}~every fact $A$ with $\qatom{A}{q}\in C$ belongs to a $1$-embedding. 
The reduction will now encode all these $1$-embeddings as $T$-facts.

We show that every directed edge of the $\cmhook{C}$-graph belongs to a directed cycle.
To this extent, take any edge $A\cmhook{C}B$.
Since every garbage set for $C$ in $\db$ is empty,
the $\cmhook{C}$-graph contains a relevant $1$-embedding containing $A$,
and a relevant $1$-embedding containing $B$.
Let $A'$ be the fact such that $A'\cmhook{C}B$ is a directed edge in the $1$-embedding containing $B$.
Let $B'$ be the fact such that $A\cmhook{C}B'$ is a directed edge in the $1$-embedding containing $A$.
Since $A\cmhook{C}B$ and $A\cmhook{C}B'$, it follows $B\sim B'$ by Lemma~\ref{lem:hockey}.
From $A'\cmhook{C}B$ and $B\sim B'$, it follows $A'\cmhook{C}B'$.
Thus, the $\cmhook{C}$-graph contains a directed path from $B$ to $A'$,
an edge from $A'$ to $B'$, and a directed path from $B'$ to $A$.
Consequently, the $\cmhook{C}$-graph contains a directed path from $B$ to $A$.

It follows that every strong component of the $\cmhook{C}$-graph is initial.
It can be easily seen that if an initial strong component contains some fact $A$,
then it contains every fact that is key-equal to~$A$.
Let $\rep$ be a repair of $\db$.
For every fact $A\in\rep$, there exists a unique fact $B\in\rep$ such that $A\cmhook{C}B$.
It follows that $\rep$ must contain an elementary $\cmhook{C}$-cycle, which must be a relevant $1$-embedding (because every garbage set for $C$ in $\db$ is empty) belonging to the same initial strong component as~$A$.
It can also be seen that there exists a repair that contains exactly one such $1$-embedding for every strong component of the $\cmhook{C}$-graph.

We define an undirected graph $G$ as follows:
for each valuation $\mu$ over $\queryvars{q}$ such that $\mu(q)\subseteq\db$,
we introduce a vertex $\theta$ with $\theta=\mu[\queryvars{C}]$.
We add an edge between two vertices $\theta$ and $\theta'$ if for some $i\in\{0,\dots,k-1\}$,
$\theta(\vec{x}_{i})=\theta'(\vec{x}_{i})$.
The graph $G$ can clearly be constructed in logarithmic space (and even in $\FO$).
We define a set $\db_{T}$ of $T$-facts and, for every $i\in\{0,\dots,k-1\}$, a set $\db_{i}$ as follows:
for all two vertices $\theta$, $\theta'$ of~$G$,
if
$$
\theta'(\vec{x}_{0})=\min\left\{\theta''(\vec{x}_{0})\mid
\mbox{$\theta''\in V(G)$ belongs to the same strong component as $\theta$}\right\},
$$
then we add to $\db_{T}$ the fact $\substitute{\theta}{u}{\theta'(\vec{x}_{0})}(T)$, 
and we add to $\db_{i}$ the fact $\substitute{\theta}{u}{\theta'(\vec{x}_{0})}(N_{i})$.
In this way, every $\db_{i}$ is consistent.
Informally, if $T$ is the atom $T(\underline{u},\vec{w})$, then we add to $\db_{T}$ the $T$-fact $T(\underline{\theta'(\vec{x}_{0})},\theta(\vec{w}))$, where $\theta'(\vec{x}_{0})$ is treated as a single value.
This fact represents that $\theta$ belongs to the strong component that is identified by $\theta'(\vec{x}_{0})$.
Since undirected connectivity can be computed in logarithmic space~\cite{DBLP:journals/jacm/Reingold08}, $\db_{T}$ and each $\db_{i}$ can be constructed in logarithmic space.

Let $\db_{C}$ be the set of all $F_{i}$-facts in $\db$ ($0\leq i\leq k-1$), 
and let $\db_{\common}\defeq\db\setminus\db_{C}$,
the part of the database $\db$ that is preserved by the reduction. 
Let $\db_{N}=\bigcup_{i=0}^{k-1}\db_{i}$.
Since $\db_{N}$ is consistent, $\db_{\common}\uplus\db_{T}\uplus\db_{N}$ is a legal input to $\cqa{q'\cup p}$, where the use of $\uplus$ (instead of $\cup$) indicates that the operands of the union are disjoint. 

We show that the following are equivalent:
\begin{enumerate}
\item\label{it:crepairone}
Every repair of $\db$ satisfies $q$.
\item\label{it:crepairtwo}
For every $\sep\in\repairs{\db_{\common}}$,
for every repair $\rep_{T}$ of $\db_{T}$,
$\sep\uplus\rep_{T}\uplus\db_{N}\models q'\cup p$.
\item\label{it:crepairthree}
Every repair of $\db_{\common}\uplus\db_{T}\uplus\db_{N}$ satisfies $q'\cup p$.
\end{enumerate}
The equivalence \ref{it:crepairtwo}$\iff$\ref{it:crepairthree} is straightforward.
We show next the equivalence \ref{it:crepairone}$\iff$\ref{it:crepairtwo}.
\framebox{\ref{it:crepairone}$\implies$\ref{it:crepairtwo}}
Let $\sep\in\repairs{\db_{\common}}$ and let $\rep_{T}$ be a repair of $\db_{T}$.
By our construction of $\db_{T}$, there exists a repair $\rep_{C}$ of $\db_{C}$ such that for every valuation $\theta$ over $\queryvars{q}$, if $\theta(q)\subseteq\sep\cup\rep_{C}$,
then for some value~$c$, $\substitute{\theta}{u}{c}(q'\cup p)\subseteq\sep\cup\rep_{T}\cup\db_{N}$.
Informally, $\rep_{C}$ contains all (and only) the relevant $1$-embeddings of $C$ in $\sep\cup\rep_{C}$ that are encoded by the $T$-facts of $\rep_{T}$.
Since $\sep\cup\rep_{C}$ is a repair of $\db$, by the hypothesis~\ref{it:crepairone}, we can assume a valuation $\theta$ over $\queryvars{C}$ such that $\theta(q)\subseteq\sep\cup\rep_{C}$.
Consequently, for some value $c$, $\substitute{\theta}{u}{c}(q'\cup p)\subseteq\sep\cup\rep_{T}\cup\db_{N}$.
\framebox{\ref{it:crepairtwo}$\implies$\ref{it:crepairone}}
Let $\rep$ be a repair of $\db$.
There exist $\sep\in\repairs{\db_{\common}}$ and $\rep_{C}\in\repairs{\db_{C}}$ such that $\rep=\sep\cup\rep_{C}$.
By the construction of $\db_{T}$, there exists a repair $\rep_{T}$ of $\db_{T}$ such that for every valuation $\theta$ over $\queryvars{q}$, if $\substitute{\theta}{u}{c}(q'\cup p)\subseteq\sep\cup\rep_{T}\cup\db_{N}$ for some $c$, then $\theta(q)\subseteq\sep\cup\rep_{C}$ (note incidentally that the converse does not generally hold). 
Informally, for every strong component $\isc$ of the $\cmhook{C}$-graph of $\db$ such that $\sep\cup\formula{\rep_{C}\cap V(\isc)}\models q$, 
the set $\rep_{T}$ encodes one $1$-embedding of $C$ in $\sep\cup\formula{\rep_{C}\cap V(\isc)}$.
Here, $V(\isc)$ denotes the vertex set of the strong component~$\isc$.
Since $\sep\cup\rep_{T}\cup\db_{N}$ is a repair of $\db_{\common}\uplus\db_{T}\uplus\db_{N}$,
it follows by the hypothesis~\ref{it:crepairtwo} that there exists a valuation $\theta$ over $\queryvars{q}$ such that $\substitute{\theta}{u}{c}(q'\cup p)\subseteq\sep\cup\rep_{T}\cup\db_{N}$ for some $c$.
Consequently, $\theta(q)\subseteq\sep\cup\rep_{C}$.

We still have to argue that $\db_{T}$ and $\db_{N}$ can be computed in $\ssdatalogmin$.
For every $i\in\{0,\dots,k-1\}$, let $\keep{R_{i}}$ be an IDB predicate of the same arity as~$R_{i}$.
For every $i\in\{0,\dots,k-1\}$, we add the following rules:
$$
\begin{datalogpgm}
\keep{R_{i}}(\vec{x}_{i},\vec{y}_{i}) & R_{i}(\vec{x}_{i},\vec{y}_{i}), \neg\del{R_{i}}(\vec{x}_{i})
\end{datalogpgm}
$$ 
where $\del{R_{i}}$ is the IDB predicate defined in the proof of Lemma~\ref{lem:scrub}.
Each predicate $\keep{R_{i}}$ is used to compute the $R_{i}$-facts that are not in the maximal garbage set. 

We now introduce rules for computing the relations for $T$ and for each $N_{i}$. 
For every $i\in\{0,\dots,k-1\}$, add the rule:
$$
\begin{datalogpgm}
\conedge(\vec{x}_{0},\rename{\vec{x}_{0}}{\supone}) &
\left\{
\begin{array}{l}
\keep{R_{0}}(\vec{x}_{0},\vec{y}_{0}),\\
\phantom{\keep{R_{0}}}\vdots\\
\keep{R_{k-1}}(\vec{x}_{k-1},\vec{y}_{k-1}),\\[1.5ex]
\keep{R_{0}}(\rename{\vec{x}_{0}}{\supone},\rename{\vec{y}_{0}}{\supone}),\\
\phantom{\keep{R_{0}}}\vdots\\
\keep{R_{k-1}}(\rename{\vec{x}_{k-1}}{\supone},\rename{\vec{y}_{k-1}}{\supone}),\\[1.5ex]
\begin{array}{ccc}
\vec{x}_{i} & =_{R_{i}} & \rename{\vec{x}_{i}}{\supone}
\end{array}
\end{array}\right\}
\end{datalogpgm}
$$
Informally, a fact $\conedge(\vec{a},\vec{a'})$ tells us that the blocks $R_{0}(\underline{\vec{a}},\blockfiller)$ and $R_{0}(\underline{\vec{a}'},\blockfiller)$ belong to the same strong component of the $\cmhook{C}$-graph.
Obviously, $\conedge$ defines a reflexive and symmetric binary relation on sequences of constants.
The predicate $\concomp$ computes undirected connectivity in the $\conedge$ relation.
$$
\begin{datalogpgm}
\concomp(\vec{x}_{0},\rename{\vec{x}_{0}}{\supone}) &
\conedge(\vec{x}_{0},\rename{\vec{x}_{0}}{\supone})\\[1.0ex]
\concomp(\vec{x}_{0},\rename{\vec{x}_{0}}{\supone}) &
\concomp(\vec{x}_{0},\rename{\vec{x}_{0}}{\supthree}),
\conedge(\rename{\vec{x}_{0}}{\supthree},\rename{\vec{x}_{0}}{\supone})
\\[1.0ex]
\concomp(\vec{x}_{0},\rename{\vec{x}_{0}}{\supthree}) &\concomp(\vec{x}_{0},\rename{\vec{x}_{0}}{\supone}),
\conedge(\rename{\vec{x}_{0}}{\supthree},\rename{\vec{x}_{0}}{\supone})
\end{datalogpgm}
$$
The latter two rules are each other's symmetric version.
The following rule picks a single identifier for each connected component of $G$,
using the abbreviated syntax for the query~(\ref{eq:groupby}) introduced in Section~\ref{sec:preliminaries}.
$$
\begin{datalogpgm}
\pick(\vec{x}_{0},\min(\rename{\vec{x}_{0}}{\supthree}))
&
\concomp(\vec{x}_{0},\rename{\vec{x}_{0}}{\supthree})
\end{datalogpgm}
$$
Informally, $\pick(\vec{a},\vec{a}')$ means that $\vec{a}'$, rather than $\vec{a}$, will serve to uniquely identify the strong component.
The following rule computes all $T$-facts:
$$
\begin{datalogpgm}
\encodet(\underline{\rename{\vec{x}_{0}}{\supone}},\vec{x}_{0},\vec{y}_{0},\dots,\vec{x}_{k-1},\vec{y}_{k-1})
&
\left\{
\begin{array}{l}
\keep{R_{0}}(\vec{x}_{0},\vec{y}_{0}),\\
\phantom{\keep{R_{0}}}\vdots\\
\keep{R_{k-1}}(\vec{x}_{k-1},\vec{y}_{k-1}),\\[1.5ex]
\pick(\vec{x}_{0},\rename{\vec{x}_{0}}{\supone})
\end{array}
\right\}
\end{datalogpgm}
$$
Finally, for every $i\in\{0,\dots,k-1\}$, add the rule:
$$
\begin{datalogpgm}
\detu{i}(\underline{\vec{x}_{i}},\rename{\vec{x}_{0}}{\supone}) &
\encodet(\underline{\rename{\vec{x}_{0}}{\supone}},\vec{x}_{0},\vec{y}_{0},\dots,\vec{x}_{k-1},\vec{y}_{k-1})
\end{datalogpgm}
$$
Note that in this encoding, the cardinality of the primary key of $\encodet$ can be greater than~$1$.
This is not a problem, because we can treat values for $u$ as composite values.

\myparagraph{Proof of the second item}
Since $\atomvars{N_{i}}\subseteq\atomvars{T}$ for every atom $N_{i}\in p$, 
we can limit our analysis to witnesses for attacks that do not contain any $N_{i}$.
Indeed, if $N_{i}$ would occur in a witness, it can be replaced with $T$.
Let $\isc$ be an initial strong component of the attack graph of $q$ that contains every atom of $\{F_{0},F_{1},\dots,F_{k-1}\}$.
It can be easily seen that for all $i\in\{0,1,\dots,k-1\}$,
\begin{equation}\label{eq:sp}
\FD{q'\cup p}\models\fd{\keyvars{F_{i}}}{\atomvars{F_{i}}}.
\end{equation}
We will use the following properties:
\begin{enumerate}[label=(\alph*)]
\item\label{it:alpha}
For every $H\in q\setminus C$, we have $\keycl{H}{q}\subseteq\keycl{H}{q'\cup p}$.
Immediate consequence of~(\ref{eq:sp}).
\item\label{it:beta}
For every $H\in q\setminus C$, if $H\attacks{q'\cup p}T$, then $H\in\isc$.
To show this result, let $H\in q\setminus C$ such that $H\attacks{q'\cup p}T$.
We can assume without loss of generality the existence of a witness for $H\attacks{q'\cup p}T$ of the form $\omega\step{v}T$ 
such that $v\neq u$.
We can assume the existence of $j\in\{0,\dots,k-1\}$ such that $v\in\queryvars{F_{j}}$. 
Then the sequence $\omega\step{v}F_{j}$ is a witness for $H\attacks{q}F_{j}$,
thus $H\in\isc$.
\end{enumerate}

\medskip
\noindent
We know by~\cite[Lemma~3.6]{DBLP:journals/tods/KoutrisW17} that if the attack graph contains a strong cycle, then it contains a strong cycle of length~$2$.
Assume that the attack graph of $q'\cup p$ contains an attack cycle $H\attacks{q'\cup p}J\attacks{q'\cup p}H$.
Then, either $H\neq T$ or $J\neq T$ (or both).
We assume without loss of generality that $H\neq T$.
We show that the attack cycle $H\attacks{q'\cup p}J\attacks{q'\cup p}H$ is weak.
We distinguish three cases.
\begin{description}
\item[Case $H\nattacks{q'\cup p}T$ (thus $J\neq T$) and $J\nattacks{q'\cup p}T$.]
Then no witness for $H\attacks{q'\cup p}J$ or $J\attacks{q'\cup p}H$ can contain $T$.
By property~\ref{it:alpha}, $H\attacks{q}J\attacks{q}H$.
Since the attack graph of $q$ contains no strong attack cycle,
$\FD{q}\models\fd{\keyvars{H}}{\keyvars{J}}$ and $\FD{q}\models\fd{\keyvars{J}}{\keyvars{H}}$.
Then, $\FD{q'\cup p}\models\fd{\keyvars{H}}{\keyvars{J}}$ and $\FD{q'\cup p}\models\fd{\keyvars{J}}{\keyvars{H}}$.
It follows that the attack cycle $H\attacks{q'\cup p}J\attacks{q'\cup p}H$ is weak.
\item[Case $H\attacks{q'\cup p}T$.]
By the property~\ref{it:beta}, $H\in\isc$.
We distinguish two cases.
\begin{description}
\item[Case $J=T$.]
The attack cycle $H\attacks{q'\cup p}J\attacks{q'\cup p}H$ is weak
because $\FD{q'\cup p}\models\fd{\keyvars{H}}{u}$ and $\FD{q'\cup p}\models\fd{u}{\keyvars{H}}$.
Recall that $\{u\}=\keyvars{T}$. 
\item[Case $J\neq T$.]
We show that $J\in\isc$ by distinguishing two cases:
\begin{itemize}
\item
if $J\nattacks{q'\cup p}T$, 
then no witness for $J\attacks{q'\cup p}H$ can contain $T$;
then $J\attacks{q}H$, and thus $J\in\isc$; and
\item
if $J\attacks{q'\cup p}T$, then $J\in\isc$ by the property~\ref{it:beta}.
\end{itemize}
From $H,J\in\isc$, it follows $\FD{q}\models\fd{\keyvars{H}}{\keyvars{J}}$ and $\FD{q}\models\fd{\keyvars{J}}{\keyvars{H}}$.
Then, $\FD{q'\cup p}\models\fd{\keyvars{H}}{\keyvars{J}}$ and $\FD{q'\cup p}\models\fd{\keyvars{J}}{\keyvars{H}}$.
It follows that the attack cycle $H\attacks{q'\cup p}J\attacks{q'\cup p}H$ is weak.
\end{description}
\item[Case $J\attacks{q'\cup p}T$ (thus $J\neq T$).]
This case is symmetrical to a case that has already been treated.
\end{description}
\end{proof}


\section{Proofs of Section~\ref{sec:keyjoin}}

We will use the following helping lemma.

\begin{lemma}\label{lem:warmup}
Let $q$ be a query in $\sjfbcq$ that has the key-join property. 
Then, for all $F,G\in q$,
if $F\attacks{q}G$,
there exists a sequence $F_{0},F_{1},\dots,F_{\ell}$
such that $F_{0}=F$, $F_{\ell}=G$,
and for all $i\in\{1,2,\dots,\ell\}$, $\keyvars{F_{i}}\subseteq\atomvars{F_{i-1}}$.
\end{lemma}
\begin{proof}
Since $q$ has the key-join property, 
for all $F,G\in q$ one of the following cases holds:
\begin{enumerate}
\item\label{it:myempty}
$\atomvars{F}\cap\atomvars{G}=\emptyset$;
\item\label{it:myleft}
$\atomvars{F}\cap\atomvars{G}=\keyvars{F}$;
\item\label{it:myright}
$\atomvars{F}\cap\atomvars{G}=\keyvars{G}$; or
\item\label{it:myboth}
$\atomvars{F}\cap\atomvars{G}\supseteq\keyvars{F}\cup\keyvars{G}$.
\end{enumerate}

Assume $F\attacks{q}G$.
We can assume a shortest sequence
$$
F_{0}\step{x_{1}}F_{1}\step{x_{2}}F_{2}\dotsm\step{x_{\ell-1}}F_{\ell-1}\step{x_{\ell}}F_{\ell}
$$
that is a witness for $F\attacks{q}G$.
We can assume that for all $0\leq i\neq j\leq \ell-1$,
$\atomvars{F_{i}}\cap\atomvars{F_{i+1}}$ and $\atomvars{F_{j}}\cap\atomvars{F_{j+1}}$ are not comparable by $\subseteq$, 
or else the witness can be shortened, contradicting that it is the shortest witness possible.

We show by induction on increasing $i$ that for all $i\in\{1,\dots,\ell\}$,
$\keyvars{F_{i}}\subseteq\atomvars{F_{i-1}}$.
This holds obviously true for those $i$ satisfying $\atomvars{F_{i-1}}\cap\atomvars{F_{i}}\supseteq\keyvars{F_{i}}$, which happens in cases~\ref{it:myright} and~\ref{it:myboth}. 
Also, for all $i\in\{1,\dots,\ell\}$, the intersection $\atomvars{F_{i-1}}\cap\atomvars{F_{i}}$ contains $x_{i}$ and is thus non-empty, 
which excludes case~\ref{it:myempty}.
So in the remainder, it suffices to show that for $i\in\{1,\dots,\ell\}$, $\atomvars{F_{i-1}}\cap\atomvars{F_{i}}\neq\keyvars{F_{i-1}}$,  which excludes case~\ref{it:myleft}.

\myparagraph{Induction basis $i=1$}
From $x_{1}\notin\keycl{F_{0}}{q}$, it follows $x_{1}\not\in\keyvars{F_{0}}$.
It follows that $\atomvars{F_{0}}\cap\atomvars{F_{1}}\neq\keyvars{F_{0}}$,
which excludes case~\ref{it:myleft}.

\myparagraph{Induction step $i\rightarrow i+1$}
The induction hypothesis is that 
$\keyvars{F_{i}}\subseteq\atomvars{F_{i-1}}$.
Consequently,
$\keyvars{F_{i-1}}\subseteq\atomvars{F_{i-1}}\cap\atomvars{F_{i}}$.
It follows that $\atomvars{F_{i}}\cap\atomvars{F_{i+1}}\nsubseteq\keyvars{F_{i}}$,
or else, as argued before, the witness would not be the shortest possible.
It follows $\atomvars{F_{i}}\cap\atomvars{F_{i+1}}\neq\keyvars{F_{i}}$,
which excludes case~\ref{it:myleft}.
\end{proof}

The proof of Theorem~\ref{the:warmup} can now be given.

\begin{proof}[Proof of Theorem~\ref{the:warmup}]
Assume that $q$ has the key-join property
We show that the attack graph of $q$ contains no strong attacks.
To this extent, assume $F\attacks{q}G$.
The sequence $F_{0},F_{1},\dots,F_{\ell-1}$ in the statement of Lemma~\ref{lem:warmup} is a sequential proof for
$\FD{q}\models\fd{\keyvars{F_{0}}}{\keyvars{F_{\ell}}}$,
and thus the attack  $F\attacks{q}G$ is weak.
The result then follows from Theorem~\ref{the:effectivedichotomy}.
\end{proof}

\end{document}